\documentclass[11pt]{article}

\usepackage{style}
\usepackage{authblk}
\usepackage{setspace}

\title{Fundamental Limits and Optimal Methods for Sharp Analytical Causal Bounds in Instrumental Variable Models}

\usepackage{authblk}

\renewcommand{\thefootnote}{\fnsymbol{footnote}}

\usepackage{authblk}

\renewcommand{\thefootnote}{\fnsymbol{footnote}}

\author[1]{Arefe Boushehrian$^*$}
\author[1]{Mohammad Reza Badri$^*$}
\author[2]{Sina Akbari}
\author[1]{Negar Kiyavash}

\affil[1]{EPFL, Switzerland}
\affil[2]{Statistical Laboratory, University of Cambridge, UK}
\date{}
\begin{document}
\maketitle

\begingroup
\renewcommand{\thefootnote}{\fnsymbol{footnote}}
\setcounter{footnote}{1}
\footnotetext{Authors contributed equally.}
\endgroup

\setcounter{footnote}{0}
\renewcommand{\thefootnote}{\arabic{footnote}}

\noindent

\begin{abstract}
    Bounding causal effects analytically, rather than numerically, is appealing for its interpretability and conceptual clarity.
    Existing sharp methods rely on optimization-based approaches such as the Balke–Pearl framework, whose computational complexity grows rapidly. An alternative line of work derives bounds heuristically using probability laws and generic inequalities, and some recent papers have claimed or conjectured that this approach can yield sharp analytical bounds with substantially lower complexity.
    In this paper, we show that this perceived advantage is illusory.
    In particular, in a discrete instrumental variable setting, we show that any sharp analytical bound for the average treatment effect 
    must be expressible as a maximum (minimum) over a collection of linear terms whose cardinality grows exponentially in the number of values taken by the outcome.
    In parallel, we show that the number of instrumental variable inequalities itself also grows exponentially. 
    Consequently, bounds and inequalities expressed using only polynomially many such terms cannot be sharp.
    As a constructive complement, the paper is accompanied by codes implemented in python and R to derive sharp analytical bounds and sharp inequalities with optimal efficiency, matching the lower bounds proven in this paper.
    These codes are available \href{https://github.com/ArefeBoushehrian/Analytical-Causal-Bounds-in-Instrumental-Variable-Models}{online}.
\end{abstract}
\vspace{2em}
\section{INTRODUCTION}
\label{sec: intro}

Bounding causal effects under partial identification has a long history in statistics, econometrics, and causal inference. 
When point identification fails, sharp bounds provide the tightest possible characterization of causal estimands consistent with observed data and maintained assumptions. 
Seminal contributions by \citet{manski1990nonparametric} and \citet{manski2003partial} formalized partial identification as a coherent framework, while subsequent work on instrumental variable (IV) models clarified how structural assumptions can result in narrower identification regions \citep{balke1995probabilistic, Balke1997BoundsOT, robins1989analysis, pearl1995testability}.

Among the available tools for bounding causal estimands, \emph{analytical bounds}, closed-form expressions written directly as functionals of observed probabilities, are especially appealing due to their interpretability and transparency.
However, the computational complexity of deriving sharp analytical bounds remains poorly understood.

Two broad approaches have been developed to derive analytical bounds in IV models.
The first formulates the bounding problem as an optimization over the space of full data distributions consistent with the observed data and the structural assumptions.
This approach originates in \cite{balke1995probabilistic}, and has been extended in multiple directions \citep{cheng2006bounds, Richardson2014ACE, Sachs2020AGM, duarte2024automated}.
These methods are guaranteed to yield sharp bounds, but their computational cost grows rapidly with the cardinalities of the observed variables, making them difficult to apply in more complex discrete settings.

The second approach derives inequalities algebraically from the observed data law using logical implications of the causal model, the laws of probability, and generic inequalities, e.g. Fréchet inequalities \citep{frechet1935generalisation, manski1990nonparametric, kedagni2020generalized, finkelstein2020deriving}.
While these bounds are not guaranteed to be sharp, they are often easier to interpret and compute.
Because these approaches do not explicitly enumerate vertices of the full-data polytope, it has been conjectured, implicitly or explicitly \cite{bonet2001instrumentality, kedagni2020generalized, shu2025identification} that sharp analytical bounds and complete testable implications might be obtained through these approaches without incurring the limiting computational complexity of the optimization-based approaches.

This paper shows that such a computational separation is impossible.
Focusing on the discrete instrumental variable (IV) model, we establish fundamental lower bounds on the complexity of \emph{both} sharp partial identification bounds for the average treatment effect (ATE) and sharp testable implications of the IV model.


First, for partial identification, we prove that when the outcome variable $Y$ takes $n$ distinct values, any sharp upper (or lower) bound for the ATE must be representable as the minimum (or maximum) over a collection of linear expressions whose cardinality grows exponentially in $n$. 
Consequently, no polynomial-size family of linear analytical expressions can yield sharp bounds in general.

Second, we provide a complete and sharp characterization of the testable implications of the discrete IV model. 
We show that the set of observed data distributions compatible with the model is defined by a family of linear inequalities corresponding to the extreme rays of a dual cone. 
Crucially, we prove that the number of necessary and sufficient IV inequalities also grows exponentially in the support size of the outcome. 
Thus, any complete system of testable implications must be exponentially large in the worst case. 
This extends and sharpens earlier results on instrumental inequalities \citep{pearl1995testability, bonet2001instrumentality, kedagni2020generalized} by providing necessity, sufficiency, and explicit lower bounds on their number in multi-valued settings.

Both results show that the exponential growth observed in linear programming approaches, therefore, reflects an intrinsic combinatorial feature of the IV model rather than an artifact of a particular computational technique.

As a constructive complement to our theoretical results, we develop a software package that computes sharp upper and lower bounds on the ATE as well as sharp IV inequalities (testable implications) with \emph{optimal output-sensitive complexity}, that is, time linear in the number of sharp terms (or inequalities). 
Rather than enumerating extreme points of the dual linear program, our implementation exploits the explicit structural characterizations derived in this paper to generate sharp bounds and IV inequalities directly. 
This avoids the combinatorial overhead inherent in off-the-shelf linear programming or vertex-enumeration routines. 
Simulation studies demonstrate substantial computational gains. 
While existing packages based on generic optimization or polyhedral enumeration become practically intractable even for moderate outcome support sizes, our method computes the full collection of sharp bounds and necessary and sufficient IV inequalities in milliseconds, even when the outcome support is large.



\subsection{RELATED WORK}
\label{sec: related works}
The identification of causal effects in IV models dates back to classical econometric analyses of simultaneous equations and compliance behavior. 
Under additional structural assumptions such as monotonicity, point identification of effects is possible \citep{angrist1995identification, angrist1991sources, chamberlain1986asymptotic, manski1998monotone}. 
Without such strengthening assumptions, causal effects are generally only partially identified.

Sharp bounds for IV models were first derived via linear programming in \citet{balke1995probabilistic}. 
Extensions to multi-valued instruments appeared in \citet{cheng2006bounds} and \citet{Richardson2014ACE}, building on earlier nonparametric bounds in \citet{robins1989analysis} and \citet{manski1990nonparametric}. 
Symbolic and algorithmic characterizations of causal bounds from general DAGs were developed in \citet{Sachs2020AGM}. 
Our results complement these works by proving that their observed time complexity is unavoidable for sharp analytical solutions, rather than an artifact of their approaches.



On the testability side, \citet{pearl1995testability} derived the original instrumental inequalities for the binary case, and \citet{bonet2001instrumentality} strengthened these results. 
\citet{kedagni2020generalized} provided necessary and sufficient generalized instrumental inequalities for the binary treatment and outcome setting under joint independence. 
Most recently, \citet{song2025categoricalinstrumentalvariablemodel} characterized the restrictions on the joint distribution of potential outcomes in a categorical IV model. 
We extend this literature by (i) providing necessary and sufficient inequalities for multi-valued outcomes, and (ii) establishing exponential lower bounds on the number of inequalities required for completeness.



Finally, continuous IV models have been studied under structural or semiparametric assumptions \citep{angrist1995identification, kitagawa2009identification, beresteanu2012partial}. 
However, without additional restrictions, continuous IV models impose no testable constraints on the observed distribution \citep{bonet2001instrumentality, gunsilius2021nontestability}. 
Our focus is therefore on the discrete setting, where rich geometric structure emerges.

\paragraph{Contributions.}
We make three main contributions:

\begin{enumerate}
    \item \textbf{Sharp analytical ATE bounds.}
    We prove that any sharp representation of the ATE bounds the discrete IV model must involve exponentially many linear terms in the outcome support size, showing that the computational burden is unavoidable when deriving sharp bounds.
    We provide an explicit and complete characterization of these bounds in the binary instrument setting.
    
    \item \textbf{Sharp testable implications.} We 
    establish exponential (in the outcome support size) lower bounds on the number of sharp testable implications (IV inequalities) in the discrete IV model. 
    We provide these sharp inequalities explicitly in the binary instrument setting.
    
    \item \textbf{Computationally optimal code.} We provide codes to compute all sharp bounds and inequalities with optimal time complexity in the binary instrument setting, which enables tractable computation of bounds and inequalities in discrete IV models.
\end{enumerate}




\section{Problem setup}
\label{sec: problem setup}
For $n\in\N$, we denote the set $\left\{0,\cdots,n-1\right\}$ by $[n]$. 
Throughout the paper, vectors are shown by bold letters (e.g., $\vp$), random variables 
by capital letters (e.g., $Y$), and probability measures by calligraphic letters (e.g., $\mathcal{P}$).

We consider the instrumental variable model with categorical variables.
We let $D$ denote a dichotomous treatment variable.
We assume without loss of generality that $D\in\{0,1\}$.
We denote the observed outcome variable and the instrument  by $Y$ and $Z$, respectively.
Following Neyman-Rubin potential outcome model, let $Y^{(d,z)}$ denote the potential outcome of $Y$, if the treatment and the instrument (possibly contrary to the fact) were set to $D=d$ and $Z=z$, respectively.
We assume these variables exist throughout.
The variables $\Dv{z}$ are defined similarly, as the potential outcome of $D$, had (possibly contrary to the fact) the instrumental variable been set to $Z=z$.
The existence of these variables, however, is only essential to some of our results.


We assume the outcome variable $Y$ takes values in a finite real-valued set
$ \left \{ \gamma_0,\gamma_1, \cdots , \gamma_{n-1} \right \}$, where 
$\gamma_0 < \gamma_1 < \cdots < \gamma_{n-1}$
without loss of generality.
We further assume that the instrument $Z$ takes values in $[\ell]$.

We require an exclusion restriction in our setting.
\begin{assumption} [Individual-level exclusion]
\label{asm: exclusion}
    $\Y{d, z} = \Y{d, z'}$ almost surely for all $z, z' \in [\ell]$ and every $d \in \{0, 1\}$.
\end{assumption}
The individual-level exclusion restriction posits that 
    there is no direct effect of $Z$ on the outcome $Y$ other than through the treatment of interest $D$.
We maintain this assumption throughout the manuscript and consequently simplify the notation and use $\Y{d}$ to denote potential outcomes.

In addition to the exclusion restriction, we will work under an IV independence assumption.
We consider the following two versions.
\begin{assumption} [Random assignment] The variables $\Dv{z}$ for $z\in[\ell]$ exist, and,
\label{asm: independence}
    \[Z \independent \left (  \Y0, \Y1, \Dv0, \Dv1,\dots,\Dv{\ell-1} \right ).\]
\end{assumption}

\begin{assumption}[Joint independence]
\label{asm: jointind}
    \[Z \independent  (  \Y0, \Y1).\]
\end{assumption}
\Cref{asm: independence} is stronger than \Cref{asm: jointind}, in the sense that the former implies the latter.
In the remainder of this paper, we work under \Cref{asm: independence} for ease of exposition.
However, we will show that most of our results are valid under \Cref{asm: jointind}.
Finally, we make the standard consistency assumption:
\begin{assumption}[Consistency]\label{asm:consistency}
\begin{align*}
\refstepcounter{assumpeq}\label{asm:consistency_a}
Y &= (1-D)Y^{(0)} + DY^{(1)} \tag{\theassumpeq}, \\
\refstepcounter{assumpeq}\label{asm:consistency_b}
D &= \textstyle\sum_{z\in[\ell]}\mathbbm{1}(Z=z)D^{(z)}, \tag{\theassumpeq}
\end{align*}

where $\mathbbm{1}(\cdot)$ is the indicator function.
\end{assumption}


Let $\mathcal{Q}^f$ and $\mathcal{P}$ represent the full data and the observed data laws, respectively.
For example, when $\ell=2$ (binary instrument), 
$\mathcal{Q}^f=\mathcal{L}(\Y{0},\Y{1},\Dv{0},\Dv{1},Z)$,
and $\mathcal{P}=\mathcal{L}(Y,D,Z)$.
We denote by $\mathcal{Q}$ the marginal law of $(\Y{0},\Y{1},\Dv{0},\dots,\Dv{\ell-1})$, after marginalizing out $Z$,   induced by $\mathcal{Q}^f$.
Moreover, for every $z\in[\ell]$, we define the stochastic kernel $\mathcal{P}_z$ as the regular conditional distribution\footnote{We assume without loss of generality that $\mathcal{P}(Z)>0$.
Otherwise, one can discard the values $z$ with zero probability and rearrange.} of $(Y,D)$ given $Z=z$ under $\mathcal{P}$. That is, $\mathcal{P}_z\coloneqq\mathcal{P}(Y,D\mid Z=z)$.
Finally, we use $q_{i j, \mathbf{d}}$ and $p_{i d,z}$ as shorthand for
$\mathcal{Q}\left(\Y{0}=\gamma_i,\Y{1}=\gamma_j, (\Dv{0},\dots, \Dv{\ell-1})=\mathbf{d}\right)$ and $\mathcal{P}_z(Y=\gamma_i, D=d)$, respectively.

In this work, we consider the following two problems.
\paragraph{Sharp Partial Identification Bounds for the Average Treatment Effect.}
We consider the problem of deriving sharp bounds for the average treatment effect, defined as 
\begin{equation}\label{eq:ate}
\begin{split}\textit{ATE}(\mathcal{Q}^f) &:= \E_{\mathcal{Q}^f}[\Y{1} - \Y{0}]
{\color{white}:}=
\textstyle\sum_{i,j\in[n]}
\sum_{\mathbf{d}\in\{0,1\}^\ell}
(\gamma_j - \gamma_i)\,
q_{ij,\mathbf{d}}
.\end{split}
\end{equation}
Specifically, the objective is to derive functionals $U(\cdot)$ and $L(\cdot)$ of the observed data law $\mathcal{P}$ such that the following inequalities are valid and sharp
\begin{equation}\label{eq: U and L}
L(\mathcal{P})\leq ATE(\mathcal{Q}^f)\leq U(\mathcal{P}),
\end{equation}


meaning that (i) they contain the true value of the estimand for every full data law that is compatible with the maintained assumptions and the observed distribution (valid); and, (ii)
they are the tightest possible among all valid bounds (sharp). 
Specifically, for every observed law $\mathcal{P}$ conforming to the IV model,
there exist full data distributions $\mathcal{Q}^f_L$ and $\mathcal{Q}^f_U$
satisfying the maintained assumptions and having the marginal $\mathcal{P}$ that attain these bounds, i.e.
\begin{equation*}
ATE(\mathcal{Q}^f_L)=L(\mathcal{P}), 
\qquad 
ATE(\mathcal{Q}^f_U)=U(\mathcal{P}).
\end{equation*}

In general, one has to derive the functional $L(\cdot)$ and $U(\cdot)$ separately.
However, in the following proposition
we show that using the symmetry of the ATE functional, $U(\mathcal{P})$ can be expressed as $L(\cdot)$ evaluated at a different point (and vise versa).



\begin{proposition}\label{prop: uno reverse}
For any observed data law $\mathcal{P}$ we have $U(\mathcal{P})\equiv -L(\bar{\mathcal{P}})$,
    where the law $\bar{\mathcal{P}}$ can be constructed based on $\mathcal{P}$ as follows:
\[
    \bar{\mathcal{P}}(Y=y, D=d, Z=z) = \mathcal{P} (Y=y, D=1-d, Z=z).
\]
\end{proposition}
\begin{proof}
Let $\bar{\mathcal{Q}}^f$ be defined as
\begin{align*}
    \bar{\mathcal{Q}}^f(\Y0=\gamma_i,\Y1=\gamma_j,(\Dv0,\dots,\Dv\ell)=\mathbf{d}, Z=z) =\\ 
    \mathcal{Q}^f(\Y0=\gamma_j,\Y1=\gamma_i,(\Dv0,\dots,\Dv\ell)=\mathbf{d}, Z=z), 
\end{align*}
which is the probability measure resulting from relabeling $\Y0$ and $\Y1$ by swapping them.
By definition,
\begin{equation}\label{eq:qbar}
\begin{split}
    \textit{ATE}(\bar{\mathcal{Q}}^f) &:= \E_{\bar{\mathcal{Q}}^f}[\Y{1} - \Y{0}]
    = \E_{\mathcal{Q}^f}[\Y{0} - \Y{1}]
    = - \textit{ATE}(\mathcal{Q}^f).
\end{split}
\end{equation}
Let $\bar{\mathcal{P}}$ be the observed data law induced by $\bar{\mathcal{Q}}^f$ under consistency.
\Cref{eq:qbar} implies that if $L(\bar{\mathcal{P}})$ is a (sharp) lower bound for ATE$(\bar{\mathcal{Q}}^f)$, then $-L(\bar{\mathcal{P}})$ is a (sharp) upper bound for ATE$(\mathcal{Q}^f)$; that is,
$U(\mathcal{P})\equiv -L(\bar{\mathcal{P}})$.
The law $\bar{\mathcal{P}}$ can be easily constructed based on $\mathcal{P}$ as follows.
\[
    \bar{\mathcal{P}}(Y=y, D=d, Z=z) = \mathcal{P} (Y=y, D=1-d, Z=z).
\]
\end{proof}
Therefore, 
it is sufficient to characterize the lower bound functional $L(\cdot)$ for our purposes.
We call $\bar{\mathcal{P}}$ the conjugate observed law, and similar to $\mathcal{P}$, we define conditionals $\bar{\mathcal{P}}_z$ and use $\bar{p}_{ij,\mathbf{d}}$ as shorthand for its conditionals.


\paragraph{Sharp Testable Implications.}

It is known that the instrumental variable model has testable implications (see, e.g. \cite{pearl1995testability,kedagni2020generalized}).
In particular, the instrumental variable model imposes restrictions on the observed data law $\mathcal{P}$, and if these restrictions are violated, the IV model is falsified.
Here, we consider the problem of deriving the set of \emph{valid} and \emph{sharp} testable implications, in the sense that (i) if the true data-generating process follows the instrumental variable model, then $\mathcal{P}$ satisfies these implications, and (ii) if  $\mathcal{P}$ satisfies this set of testable implications, then there exists a full data law $\mathcal{Q}^f$ that follows the IV model and induces (marginalizes to) $\mathcal{P}$.
In other words, the IV model cannot be falsified.


\section{Linear Programming Formulation and the Dual}
\label{sec: lp formulation and dual}
In this section, we formalize the set of necessary and sufficient (linear) constraints that the full data law must satisfy under Assumptions \ref{asm: exclusion}, \ref{asm: independence} and \ref{asm:consistency}. Additionally, we review a unified linear programming approach to solving both classes of problems considered in this work.


\begin{restatable}{proposition}{prprelation}\label{prop:relation}
\label{prop: suffciency of constraints}
    Under Assumptions \ref{asm: exclusion}, \ref{asm: independence} and \ref{asm:consistency}, given the observed law $\mathcal{P}$, the marginal law of potential outcome variables $Q(\cdot)$ is sharply characterized as follows:


    \begin{align}
    \label{eq: relation between q and p}
    \begin{cases}
        p_{y 0, z} =
        \sum\limits_{j\in [n]}
        \sum\limits_{\substack{\mathbf{d} \in \{0,1\}^{\ell}, \\d_z=0}}
        q_{y j, \mathbf{d}}
        \quad 
        \forall y \in [n],\; z \in [\ell],
        \\
        p_{y 1, z} =
        \sum\limits_{i\in [n]}
        \sum\limits_{\substack{\mathbf{d} \in \{0,1\}^{\ell}, \\d_z=1}}
        q_{i y, \mathbf{d}}
        \quad 
        \forall y \in [n],\; z \in [\ell],
        \\
        q_{ij,\mathbf{d}} \ge 0 
        \quad \hspace{4em}\forall i,j \in [n],\; \mathbf{d} \in \{0,1\}^{\ell},
    \end{cases}
    \end{align}
    where $d_z$ is the $z$-th element of vector $\mathbf{d}$.
\end{restatable}
The proof of \Cref{prop: suffciency of constraints} is given in Appendix \ref{proof: sufficiency}.
\Cref{prop: suffciency of constraints} states that for any IV model, these equations hold (necessity), and for every pair $\mathcal{P},\mathcal{Q}$ that satisfy these equations, there exists a full data law $\mathcal{Q}^{f*}$ conforming to the IV model that marginalizes to $\mathcal{Q}$ and $\mathcal{P}$ (sufficiency).


Let $\vp ,\Bar{\vp} \in \R^{2\ell n}$ and $\vq \in \R^ {2^\ell n^2}$ be vector representations of $p_{yd,z}$ and $\Bar{p}_{yd,z}$  for all $y,d,z$ and $q_{ij,\mathbf{d}}$ for all $i,j,\mathbf{d}$, respectively. 
\Cref{eq: relation between q and p} can be expressed in matrix form:
\begin{equation}
\label{eq: matrix LP formulation}
\Mt \vq = \vp, \qquad \vq \ge 0,
\end{equation}
where $M$ is a binary matrix such that
\[
M_{(ij,\mathbf{d}),(yd,z)} = 1
\]
if $q_{ij,\mathbf{d}}$ appears in the equation corresponding to $p_{yd,z}$ in \Cref{eq: relation between q and p}, and is $0$ otherwise. 
In Appendix, we provide an example where $n=\ell=2$ and the corresponding matrix $M$ is given (see \Cref{App: example1}).

\subsection{Linear Programming and the Dual Problem}
Based on \Cref{prop: suffciency of constraints} and its matrix representation (Equation \ref{eq: matrix LP formulation}),
the set of all full data laws consistent with the observed distribution $\vp$ is given by the convex polytope
\begin{align}
\Gamma(\vp)
:=
\left\{
\vq \in \mathbb{R}^{2^\ell n^2} :
\Mt \vq = \vp, \;
\vq \ge 0
\right\}.
\end{align}
Below, we review how the problems we consider in this paper are cast into linear programming (LP) problems using the latter observation.

\paragraph{Partial Identification for ATE.}

The ATE is a linear functional of the full data law. 
In particular,
\[
\text{ATE}(\mathcal{Q}^f)
=
\vc^\top \vq,
\]
where the coefficient vector $\vc$ is defined component-wise as
\[
c_{ij,\mathbf{d}}
=
\gamma_j - \gamma_i,
\]
for all $i,j \in [n]$ and all $\mathbf{d} \in \{0,1\}^{\ell}$ (see Equation \ref{eq:ate}).

Therefore, the sharp identification region for the ATE is given by solutions to the linear programs
\begin{align}\label{eq: primal lower bound PI}
L(\mathcal{P})
&=
\min_{\vq}
\quad
\vc^\top \vq
\quad \text{s.t.}
\quad
\Mt \vq = \vp,
\quad
\vq \ge 0
\end{align}
and

\begin{align}\label{eq: primal upper bound PI}
U(\mathcal{P})
&=
\max_{\vq}
\quad
\vc^\top \vq
\quad \text{s.t.}
\quad
\Mt \vq = \vp,
\quad
\vq \ge 0 \nonumber\\
&=
-L(\Bar{\mathcal{P}}) \nonumber\\
&=
-\min_{\vq}
\quad
\vc^\top \vq
\quad \text{s.t.}
\quad
\Mt \vq = \Bar{\vp},
\quad
\vq \ge 0,
\end{align}
where $U,L$ are defined in \Cref{eq: U and L}.
Since $\Gamma(\vp)$ (the feasibility region) is a convex polytope, the resulting interval
\[
[L(\mathcal{P}), U(\mathcal{P})] = [L(\mathcal{P}), -L(\bar{\mathcal{P}})]
\]
is the sharp identification region for the ATE.
\red

\black
\paragraph{Testable Implications of IV Model.}
To test whether a given observed probability vector $\vp$ is compatible with the IV model, we seek necessary and sufficient conditions on $\vp$ such that existence of a vector $\vq$ satisfying the linear system in \Cref{eq: matrix LP formulation} is guaranteed.
In other words, we want to determine whether the set $\Gamma(\vp)$ is nonempty. 
This can be cast as deciding the feasibility of the following linear program:
\begin{align}
\label{eq: general primal feasibility test}
\min_{\vq}
\quad
0
\quad \text{s.t.}
\quad
\Mt \vq = \vp,
\quad
\vq \ge 0.
\end{align}

The IV model can be falsified if and only if the latter is not feasible.


\paragraph{Dual Formulation.}
The dual formulation provides an alternative characterization of both sharp identification regions and testable implications. 
By the strong duality theorem of linear programming, whenever the feasible set $\Gamma(\vp)$ is nonempty, the optimal values of the primal and dual problems coincide 
(see, e.g., \cite[Section 5.2.3]{boyd2004convex}).
The dual has two desirable advantages:
(i) solving the dual can be computationally more efficient as it is defined over parameters corresponding to the observed data with dimension $2\ell n$, whereas the primal is defined over parameters corresponding to the full data law with dimension $2^\ell n^2$;
(ii) more importantly, the dual admits a 
form where the constraints do not depend on the observed vector $\vp$.
As we shall see in \Cref{sec: extreme points}, the latter is essential for deriving analytical (closed-form) solutions to the LPs.



The dual LPs associated with Equations 
\eqref{eq: primal lower bound PI} and \eqref{eq: primal upper bound PI} are
\begin{align}
\label{eq: dual_lower_bound_formal}
\max_{\vv\in \R^{2\ell n}}
\quad
\vv^\top \vp
\quad
\text{s.t.}
\quad
M \vv \le \vc,
\end{align}
and 
\begin{align}
\label{eq: dual_upper_bound_formal}
-\max_{\vv \in \R^{2\ell n}}
\quad
\vv^\top \Bar{\vp}
\quad
\text{s.t.}
\quad
M \vv \le \vc,
\end{align}
respectively,
where $\vv$ is the vector of dual variables corresponding to the equality constraints of primal LPs.
Every dual-feasible vector $\vv$ in Equation \eqref{eq: dual_lower_bound_formal} or \eqref{eq: dual_upper_bound_formal} generates a valid bound which holds uniformly over all full data laws consistent with the observed distribution. The optimal dual solution yields the tightest such bound. 

In the testable implications setting,
a fundamental result from linear programming (a theorem of the alternative; see, e.g., \cite[Section 5.8]{boyd2004convex}) implies that the primal problem of \Cref{eq: general primal feasibility test} 
is 
feasible if and only if
\begin{align}
\label{eq: general dual feasibility test_formal}
\{
\max_{\vr\in \R^{2\ell n}}
\quad
\vp^\top \vr
\quad
\text{subject to}
\quad
M \vr \le 0
\}
\;\;\le\;\; 0.
\end{align}
In other words,
each vector $\vr$ satisfying $M \vr \le 0$ imposes the linear inequality $\vp^\top \vr \le 0$ that must be satisfied by the observed distribution $\vp$. 
These inequalities are necessary and sufficient for the existence of a full data law $\vq \ge 0$ satisfying $\Mt \vq = \vp$.




\subsection{Dual Solution via Extreme Points and Rays}\label{sec: extreme points}
In this section, we discuss how to derive closed-form solutions of the dual LPs.
We
Denote the feasible regions of Equations \eqref{eq: dual_lower_bound_formal} and \eqref{eq: dual_upper_bound_formal} by
\[
\mathcal H := \{ \vv \in \R^{2\ell n} : M\vv \le \vc \},
\]
and the corresponding feasible region of \Cref{eq: general dual feasibility test_formal} by
\[
\mathcal K := \{ \vr \in \R^{2\ell n} : M\vr \le 0 \}.
\]
Note that these sets do not depend on the observed data law; they remain invariant across every instance of the IV model, and so are their extreme points.
Closed-form solutions can therefore be obtained through finding the extreme points and rays of these feasible regions.
We provide the necessary formalism below.

Both $\mathcal{H}$ and $\mathcal{K}$ are polyhedra. 
However, since $M$ might not have full column rank, these polyhedra contain affine directions given by $\ker(M)$. 
Consequently, feasible points are not isolated in $\R^{2\ell n}$: if $\vv$ is feasible, then so is $\vv + \vs$ for any $\vs \in \ker(M)$.
To eliminate this intrinsic degeneracy, we work modulo $\ker(M)$.

\begin{definition}[$M$-equivalent and $M$-distinct]
Two vectors $\vv_1, \vv_2$ are $M$-equivalent if
\[
\vv_1 - \vv_2 \in \ker(M),
\]
or equivalently,
\(
M\vv_1 = M\vv_2.
\)
Two vectors are called $M$-distinct if they are not $M$-equivalent.
\end{definition}

Working in the quotient space $\mathbb{R}^{2\ell n} \setminus \ker(M)$ removes affine directions and restores a proper notion of extremality.

\begin{definition}[Vertex]
\label{def: vertex}
A vector $\vv\in\mathcal{H}$ is a vertex of $\mathcal H$ if
for any representation
\[
\vv = \lambda \vv_1 + (1-\lambda) \vv_2,
\qquad \lambda \in (0,1),
\quad \vv_1, \vv_2 \in \mathcal H,
\]
the vectors $\vv,\vv_1,\vv_2$ are $M$-equivalent.
\end{definition}

A vertex is a point in $\mathcal{H}$ that cannot be written as a nontrivial convex combination of two $M$-distinct feasible points.
The dual optimum will be attained at a vertex of $\mathcal H$ (see \Cref{sec: appendix preliminaries vertex} for a more detailed discussion). 
Therefore, sharp identification bounds can be obtained by evaluating $\vpt \vv$ and $\Bar{\vp}^{\top} \vv$ over the (finite) set of vertices.


The set $\mathcal K$ is a polyhedral cone. 
The feasibility of the primal problem is equivalent to
\[
\vpt \vr \le 0 \quad \text{for all } \vr \in \mathcal K.
\]
Hence, the extreme rays of $\mathcal{K}$ determine all necessary and sufficient IV inequalities.

\begin{definition}[Extreme ray] \label{def: extreme ray}
A nonzero vector $\vr \in \mathcal K$ is an extreme ray of $\mathcal K$
if for any decomposition
\[
\vr = \vr_1 + \vr_2, 
\quad \vr_1, \vr_2 \in \mathcal K,
\]
there exist $\alpha, \beta \ge 0$ such that
\begin{equation*}
M\vr_1=\alpha M\vr, \quad M\vr_2=\beta M\vr
\end{equation*}
or equivalently, 
\[
\vr_1 = \alpha \vr + \vs_1, 
\qquad
\vr_2 = \beta \vr + \vs_2
\]
for some $\vs_1, \vs_2 \in \ker(M)$.
\end{definition}

Each extreme ray $\vr$ generates the affine set
\[
\mathcal{E}_{\vr}
= \{ \alpha \vr + \vs : \alpha > 0,\; \vs \in \ker(M) \}.
\]
Since $\mathcal K$ is a polyhedral cone, it admits only finitely many $M$-distinct extreme rays \cite{Rockafellar1970} (see \Cref{sec: appendix preliminaries extreme rays} for a more detailed discussion).
Consequently, $\mathcal K$ can be written as the conic hull of finitely many extreme rays. 
It follows that the infinite family of inequalities $\vpt \vr \le 0$ for all $\vr \in \mathcal K$ reduces to a finite family indexed by the extreme rays.


The above discussion yields a unifying principle:

\begin{itemize}
    \item Sharp identification bounds are obtained by enumerating the vertices of $\mathcal H$.
    \item Sharp testable implications are obtained by enumerating the extreme rays of $\mathcal K$.
\end{itemize}
\black

We shall explicitly characterize the vertices of $\mathcal{H}$
and the extreme rays of $\mathcal{K}$ 
in Sections~\ref{sec: pi} and~\ref{sec: iv}, respectively, and present closed-form sharp bounds and IV inequalities accordingly.

\section{Partial Identification of ATE}
\label{sec: pi}

We first present the sharp ATE bounds for the case where $\ell=2$ (binary instrument) in \Cref{sec: l2ate}.
Later in \Cref{sec: llate}, we provide an exponential lower bound on the number of terms in the bounds for an arbitrary value of $\ell$ (multi-valued instrument).

\subsection{Binary instrument}
\label{sec: l2ate}

To compute the vertices of $\mathcal{H}$ explicitly in this case, we use the special structure of $\mathcal{H}$.
An extensive list of the useful properties of $\mathcal{H}$ along with their proof is available in Appendix \ref{sec: proof of pi}.
Using these properties, for every feasible point $\vv\in\mathcal{H}$, we build a corresponding binary vector $\vB$, which we call the signature of that feasible point.
See 
\Cref{def: Bx} in Appendix \ref{app: PI b properties} for the details of constructing a signature vector.
We next define a set $S$ of such signature vectors that, as we shall see, correspond to the $M$-distinct vertices of $\mathcal{H}$.


\begin{restatable}[Admissible signatures]{definition}{defineS}\label{def: S define}
    The set of admissible signatures, denoted by $S$, is defined as $S=S_1\cup S_2\cup S_3$, where each $S_i$ is a set of binary vectors $\vB\in\R^{n\times2\times2}$ s.t. 
    \begin{enumerate}
        \item $\vB\in S_1$ iff
            \begin{itemize}
                \item $\exists\:t\in[n-1]$ such that $\B_{i00}=\B_{i01}=1$ for all $i\ge t$ and $\B_{i00}\ne \B_{i01}$ for all $i< t$.
                \item For all $i\in [n]$, $\B_{i10}\ne\B_{i11}$.
                \item There exist $i,j\in [n]$ such that $\B_{i10}=\B_{j11}=1$.
            \end{itemize}
        \item $\vB\in S_2$ iff
            \begin{itemize}
                \item $\B_{(n-1)00}=\B_{(n-1)01}=\B_{010}=\B_{011}=1$.
                \item For all $0\le i<n-1$, $\B_{i00}\ne\B_{i01}$.
                \item For all $0< j\le n-1$, $\B_{j10}\ne\B_{j11}$.
            \end{itemize}
        \item $\vB\in S_3$ iff 
            \begin{itemize}
                \item $\exists\:t\in [n-1]$ such that $\B_{i10}=\B_{i11}=1$ for all $i\le t+1$ and $\B_{i10}\ne \B_{i11}$ for all $i>t+1$.
                \item For all $i\in [n]$, $\B_{i00}\ne\B_{i01}$.
                \item There exist $i, j\in [n]$ such that $\B_{i00}=\B_{j01}=1$.
            \end{itemize}
    \end{enumerate}
\end{restatable}

In Appendix \ref{sec: proof of pi}, we first show that given a set of $M$-distinct vertices of $\mathcal{H}$, every vertex induces a unique admissible signature $\vB\in S$.
Conversely, given an admissible signature, we can reconstruct the corresponding vertex, as outlined below.


          
          
          


\begin{restatable}[Vertex map]{definition}{defineU}\label{def: U define}
Given an admissible signature $\vB\in S$,
we define the vector $\vu(\vB) \in \R^{n \times 2 \times 2}$ as follows:
\begin{itemize}
    \item $u_{i00} = -\gamma_i - \alpha$ if $\B_{i00} = 1$, and
          $u_{i00} = \gamma_0$ otherwise;
          
    \item $u_{i10} = \gamma_i$ if $\B_{i10} = 1$, and
          $u_{i10} = -\gamma_{n-1} - \alpha$ otherwise;
          
    \item $u_{i01} = -\gamma_i$ if $\B_{i01} = 1$, and
          $u_{i01} = \gamma_0 + \alpha$ otherwise;
          
    \item $u_{i11} = \gamma_i + \alpha$ if $\B_{i11} = 1$, and
          $u_{i11} = -\gamma_{n-1}$ otherwise,
\end{itemize}
where 
\[\alpha=
\begin{cases}
     -\gamma_0 - \gamma_t \quad&\text{ if } \vB \in S_1,\\
     -\gamma_0 - \gamma_{n-1} \quad&\text{ if } \vB \in S_2,\\
     -\gamma_t - \gamma_{n-1} \quad&\text{ if } \vB \in S_3.
\end{cases}
\]
\end{restatable}

This following result establishes a one-to-one correspondence between the set of admissible signatures $S$ and the $M$-distinct vertices of $\mathcal{H}$.

\begin{restatable}{theorem}{ateVertices}\label{thm:main PI result}
 The set $\mathcal{V}=\{\vu(\vB) : \vB \in S\}$ is a set of $M$-distinct vertices of $\mathcal{H}$, and any other vertex of $\mathcal{H}$ is $M$-equivalent to a member of $\mathcal{V}$.
\end{restatable}

\begin{corollary}
    The number of $M$-distinct vertices of $\mathcal{H}$ is exactly
    \[
    5 \times 4^{n-1} - 2^{n+2} + 4,
    \]
    which is the number of admissible signatures.
\end{corollary}

Sharp ATE bounds can be expressed in terms of the $M$-distinct vertices of $\mathcal{H}$ as follows.

\begin{restatable}{theorem}{ateTerms}\label{thm: PI ATE bounds}
Under Assumptions \ref{asm: exclusion}, \ref{asm: independence}, and \ref{asm:consistency}, the ATE admits the following valid and sharp bounds:
    \begin{equation}
    \label{eq:atebounds}
        \max_{\vv \in \mathcal{V}} \vv^\top \vp
        \;\le\;
        ATE
        \;\le\;
        -\max_{\vv \in \mathcal{V}} \vv^\top \Bar{\vp}.
    \end{equation}
\end{restatable}

We presented the sharp bounds on ATE under random assignment (\Cref{asm: independence}).
However, our next result indicates that the ATE admits the same set of sharp bounds under the weaker assumption of joint independence.

\begin{restatable}{theorem}{prpjointindep}
\label{thm: jointindep}
    The ATE bounds of \Cref{eq:atebounds} are valid and sharp under Assumptions \ref{asm: exclusion}, \ref{asm: jointind} and \ref{asm:consistency_a}.
\end{restatable}
\Cref{thm: jointindep} does not require the variables $D^{(z)}$ to be defined.

For $n = 2$, the corresponding $S$, $\mathcal{V}$, and ATE bounds are provided in Appendix \ref{App: example2}.

\subsection{Multi-valued instrument}
\label{sec: llate}

To show that the number of linear terms in the ATE bounds grows exponentially in the multi-valued instrument case ($\ell>2$), 
we construct a family of $M$-distinct vertices whose cardinality grows as $\Omega(\ell^n)$.

\begin{restatable}{theorem}{propgeneralizedvertex}
\label{prop: generalized exponential vertices}
For any $a \in [\ell]$, and any
$\vs = (s_0,\dots,s_{n-1}) \in ([\ell]\setminus\{a\})^n$,
define $\vw(a,\vs)\in\R^{n\times 2\times\ell}$ as
\begin{align*}
w_{y1,j} &= 0, && \forall y,\ j \neq a, \\
w_{y1,a} &= \gamma_{y}-\gamma_{n-1}, && \forall y, \\
w_{y0,j} &= 0, && \forall y,\ j \notin \{a,s_y\}, \\
w_{y0,s_y} &= \gamma_{n-1}-\gamma_{y}, && \forall y, \\
w_{y0,a} &= \gamma_0-\gamma_{n-1}, && \forall y.
\end{align*}
Then $\left\{\vw(a,\vs): a\in[\ell],\ \vs \in ([\ell]\setminus\{a\})^n \text{ non-constant}\right\}$ is a set of $M$-distinct vertices of $\mathcal{H}$ (the feasible region of
Equation~\ref{eq: dual_lower_bound_formal}).
The number of these $M$-distinct vertices equals
\[
\ell\Big((\ell-1)^{\,n-1}-(\ell-1)\Big).
\]
\end{restatable}

\begin{restatable}[Exponential lower bound]{corollary}{thmExponentialATEBounds}%
\label{thm: exponential lower bound ATE}
Under Assumptions \ref{asm: exclusion}, \ref{asm: independence} and \ref{asm:consistency} (or alternatively, \ref{asm: exclusion}, \ref{asm: jointind} and \ref{asm:consistency_a}),
if a set of linear functionals of $\mathcal{P}$ is a sharp ATE bound, then it must contain at least
\[
\ell\Big((\ell-1)^{\,n-1}-(\ell-1)\Big)
\]
terms.  
\end{restatable}



\section{Testable Implications}
\label{sec: iv}
In a similar fashion to the previous section, we first present the sharp implications expressed as IV inequalities in a binary instrument case ($\ell=2$) in \Cref{sec: l2ineq}.
Later in \Cref{sec: llineq}, we provide an exponential lower bound on the number of sharp inequalities for an arbitrary value of $\ell$.

\subsection{Binary instrument}
\label{sec: l2ineq}
As mentioned in \Cref{sec: lp formulation and dual}, the implications of the IV model boil down to
\begin{equation}
\label{eq: polar_cone_condition_formal}
\vp^\top \vr \le 0
\quad
\text{for all } \vr \in \mathcal{K} .
\end{equation}
Since $\mathcal{K}$ is a polyhedral cone, it can be expressed as a cone combination of finitely many, e.g. $\omega$, extreme rays 
(see, e.g., \cite[Section 19, Part IV]{Rockafellar1970}):
\begin{equation}
\label{eq: cone representation}
    \mathcal K = \text{cone}(r_1,\dots,r_{\omega})
\end{equation}
\Cref{eq: polar_cone_condition_formal} 
is therefore equivalent to 
the finite family of inequalities
\begin{equation}
\label{eq: iv_inequalities_from_rays}
\vpt r_i \le 0 \qquad \forall i \in \{1,\dots,\omega\},
\end{equation}
where $\{r_1,\dots,r_\omega\}$ are $M$-distinct extreme rays of $\mathcal{K}$.
That is, the inequalities of \Cref{eq: iv_inequalities_from_rays} are sufficient to test the IV model.
In Appendix \ref{sec: proof of iv}, we characterize these extreme rays (see \Cref{thm: pattern of extreme rays}). 
We further show that each inequality of \Cref{eq: iv_inequalities_from_rays} is necessary, that is, for each $i\in\{1,\dots,\omega\}$, there exists a vector $\vp$ such that $\vp$ satisfies all inequalities but the one corresponding to $r_i$ (see \Cref{sec: proof of necessity and sufficiency of extreme rays}).
Therefore, to obtain an explicit set of necessary and sufficient (valid and sharp) IV inequalities, it suffices to characterize the $M$-distinct extreme rays $r_i$ of
$\mathcal K$.
The following result presents these sharp implications.

\begin{restatable}{theorem}{thminequalities}
\label{thm: iv inequalities}
    Let $\mathcal{P}$ be the observed probability distribution, and $p_{yd,z}$ defined as in \Cref{sec: problem setup}.
    The sharp testable implications of the IV model under Assumptions \ref{asm: exclusion}, \ref{asm: independence}, and \ref{asm:consistency} are the following inequalities:
    \begin{align}
    \label{eq: necessary and sufficient iv inequalities}
    \begin{cases}
        \sum\limits_{k\in [n]} -p_{k0,1} + \sum\limits_{k \in T} (p_{k1,0} -p_{k1,1})
        \leq 0 
         \qquad \forall T \subseteq [n-1], T \neq \varnothing \\
        \sum\limits_{k\in [n]} -p_{k1,0} + \sum\limits_{k \in T} (p_{k0,1} -p_{k0,0})
        \leq 0 
        \qquad \forall T \subset [n], T \neq \varnothing \\
        \sum\limits_{k\in [n]} -p_{k0,0} + \sum\limits_{k \in T} (p_{k1,1} -p_{k1,0})
        \leq 0 
         \qquad \forall T \subseteq [n-1], T \neq \varnothing 
    \end{cases}
    \end{align}
    
\end{restatable}
\begin{corollary}
    The number of sharp IV inequalities in a binary instrument setting is
    $2^{n+1} - 4$. 
\end{corollary}

\begin{remark}
For a binary instrument,
\cite{bonet2001instrumentality} derived $2^{n+1}$ IV inequalities, showing that they are necessary but not sufficient.
In fact, some of those inequalities were shown to be redundant within the same paper. 
More recently, \cite{kedagni2020generalized} obtained $2^{n^2} + 2n(4 + 2^n) + 2^{n+1}$ inequalities and proved their necessity, noting that this set may not be sufficient (and clearly contains redundancies). 
\Cref{thm: iv inequalities} provides a non-redundant, necessary and sufficient set of size $2^{n+1} - 4$.
\end{remark}

For $n = 2$, the corresponding extreme rays and IV inequalities are provided in Appendix \ref{App: example3}.

\subsection{Multi-valued instrument}
\label{sec: llineq}
 

We now show that the number of IV inequalities grows exponentially when the instrument $Z$ takes $\ell > 2$ values. 
To this end, it suffices to find a set of 
$M$-distinct extreme rays of $\mathcal{K}$ with exponentially increasing cardinality, and construct the corresponding IV inequalities.
The following result presents such a set of inequalities.

\begin{restatable}{theorem}{propgeneralizediv}
\label{prop: exponential inequalities general}
Suppose Assumptions \ref{asm: exclusion}, \ref{asm: independence} and \ref{asm:consistency} hold.
For any $y' \in [n-1]$, $j' \in [\ell]$, and  non-constant $(j_0,\dots,j_{n-1}) \in ([\ell]\setminus\{j'\})^n$,
the following inequality holds:
\begin{align}
\label{eq: exponential iv inequality}
p_{y'1,j'} 
\;\le\;
\sum_{j \in \{j_0,\dots,j_{n-1}\}} p_{y'1,j}
\;+\;
\sum_{y=0}^{n-1} p_{y0,j_y}.
\end{align}
Moreover, 
none of these inequalities are implied by the others (no redundancy),
and their total number is
\[
(n-1)\,\ell\,\bigl((\ell-1)^n - (\ell-1)\bigr).
\]
\end{restatable}

\begin{corollary}[Exponential lower bound]
\label{cor: exponential lower bound inequalities}
Any sharp set of linear IV (in $\mathcal{P}$) inequalities must contain at least
$
(n-1)\,\ell\,\bigl((\ell-1)^n - (\ell-1)\bigr)
$
inequalities.  
\end{corollary}

\section{Simulation}
\label{sec: sim}
As shown in \Cref{thm: PI ATE bounds} and \Cref{thm: iv inequalities}, both the number of linear terms in the ATE bounds and the number of testable implications grows exponentially in the outcome support $n$.
These numbers are provided in \Cref{tab:ate_horizontal} for $2\leq n\leq10$. 
These results establish an exponential-time lower bound on the time complexity of deriving closed-form bounds and testable implications.
Our approach of explicitly characterizing the vertices (extreme rays) of the dual feasible set allows for the derivation of the bounds (inequalities) in time linear in the number of terms appearing in the bound (inequalities).
We can therefore use these characterizations to derive the ATE bounds as well as testable implications efficiently, i.e., with time complexity matching the lower bound.

\begin{table}[ht]
\centering
\begin{tabular}{|c|c|c|}
\hline
$n$ & Terms in ATE Bound & IV Inequalities \\ \hline
2 & 8 & 4 \\ \hline
3 & 52 & 12 \\ \hline
4 & 260 & 28 \\ \hline
5 & 1156 & 60 \\ \hline
6 & 4868 & 124 \\ \hline
7 & 19972 & 252 \\ \hline
8 & 80900 & 508 \\ \hline
9 & 325636 & 1020 \\ \hline
\end{tabular}
\caption{Number of linear terms in the ATE bounds, and the number of testable implications. $n$ represents the size of the outcome support.}
\label{tab:ate_horizontal}
\end{table}

\begin{figure}[t]
    \centering
    \includegraphics[width=0.6\linewidth]{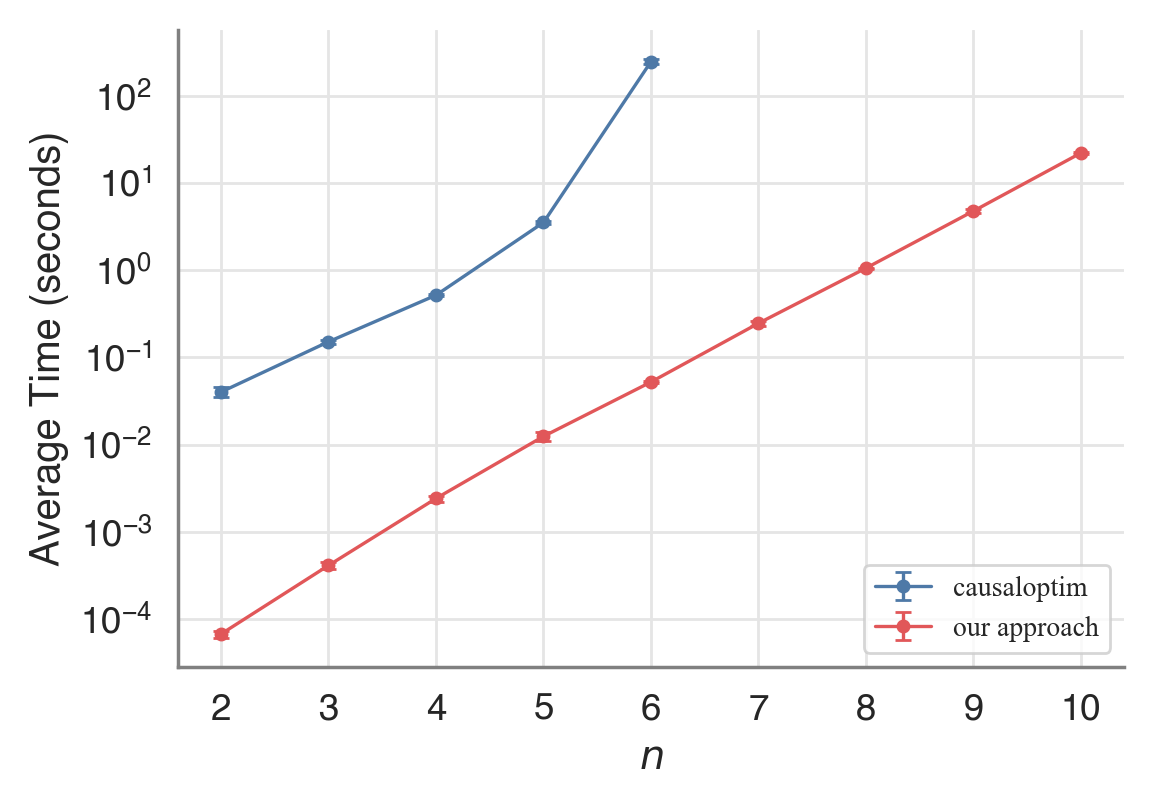}
    \caption{Average running time of our approach 
    vs causaloptim \citep{Sachs2020AGM} 
    for enumerating the vertices of the dual LP and producing the ATE bounds.}
    \label{fig:figure1}
\end{figure}

In order to assess whether this approach is advantageous to the existing ones,
we compared the running time of deriving sharp closed-form ATE bounds using our approach vs the approach of \citet{Sachs2020AGM}, who enumerate the vertices of the feasible set using off-the-shelf methods.
We used the causaloptim library in R \citep{Sachs2020AGM} for this comparison.
The code implementing our approach is available \href{https://github.com/ArefeBoushehrian/Analytical-Causal-Bounds-in-Instrumental-Variable-Models}{online}\footnote{https://github.com/ArefeBoushehrian/Analytical-Causal-Bounds-in-Instrumental-Variable-Models}.

\Cref{fig:figure1} illustrates the running times across different values of $n$ (the size of the outcome support).
As expected, our approach 
requires time exponential in $n$ (linear trend in the log-scale plot).
However, as shown in \Cref{fig:figure1}, the approach of \citet{Sachs2020AGM} 
requires time \emph{super-exponential} in $n$ to enumerate the vertices, albeit the number of vertices is only exponential.
With support sizes as small as $n=6$, we can already witness more than $1000$ times faster running times through our approach.


\section{CONCLUSION}
\label{sec: conclusion}
We studied the problem of deriving sharp analytical closed-form ATE bounds and testable implications of the discrete IV model.
We showed that the number of linear expressions in any set of sharp bounds for the ATE, as well as the number of sharp testable implications in such models grow exponentially in the outcome support.
This rules out any approach trying to derive sharp linear bounds (inequalities) in polynomial time.
On the positive side, we explicitly characterized the sharp ATE bounds and the sharp IV inequalities in the binary instrument setting.
As shown in our simulations, this approach makes deriving bounds and inequalities tractable for moderate sizes of outcome support, which were untractable with previously existing approaches.
We leave such sharp characterizations for multi-valued instruments as future work.

\section*{Aknowledgements}
Sina Akbari was supported by the Swiss National Science Foundation through a Mobility fellowship under grant P500PT\_230240.




\bibliographystyle{plainnat}
\bibliography{ref}

@book{balke1995probabilistic,
  title={Probabilistic counterfactuals: semantics, computation, and applications},
  author={Balke, Alexander Abraham},
  year={1995},
  publisher={University of California, Los Angeles}
}

@article{Balke1997BoundsOT,
  title={Bounds on Treatment Effects from Studies with Imperfect Compliance},
  author={Alexander Balke and Judea Pearl},
  journal={Journal of the American Statistical Association},
  year={1997},
  volume={92},
  pages={1171-1176},
  url={https://api.semanticscholar.org/CorpusID:18365761}
}

@book{Rockafellar1970,
url = {https://doi.org/10.1515/9781400873173},
title = {Convex Analysis},
author = {Ralph Tyrell Rockafellar},
publisher = {Princeton University Press},
address = {Princeton},
doi = {doi:10.1515/9781400873173},
isbn = {9781400873173},
year = {1970},
lastchecked = {2025-11-18}
}

@article{song2025categoricalinstrumentalvariablemodel,
  title={The categorical instrumental variable model: Characterization, partial identification, and statistical inference},
  author={Song, Yilin and Guo, F Richard and Chan, KC and Richardson, Thomas S},
  journal={arXiv preprint arXiv:2405.09510},
  year={2024}
}

@article{Sachs2020AGM,
  title={A General Method for Deriving Tight Symbolic Bounds on Causal Effects},
  author={Michael C. Sachs and Gustav Jonzon and Arvid Sj{\"o}lander and Erin E. Gabriel},
  journal={Journal of Computational and Graphical Statistics},
  year={2020},
  volume={32},
  pages={567 - 576},
  url={https://api.semanticscholar.org/CorpusID:244955008}
}

@article{kedagni2020generalized,
    author = {Kédagni, Désiré and Mourifié, Ismael},
    title = {Generalized instrumental inequalities: testing the instrumental variable independence assumption},
    journal = {Biometrika},
    volume = {107},
    number = {3},
    pages = {661-675},
    year = {2020},
    month = {02},
    issn = {0006-3444},
    doi = {10.1093/biomet/asaa003},
    url = {https://doi.org/10.1093/biomet/asaa003},
    eprint = {https://academic.oup.com/biomet/article-pdf/107/3/661/33658405/asaa003.pdf},
}

@inproceedings{pearl1995testability,
author = {Pearl, Judea},
title = {On the testability of causal models with latent and instrumental variables},
year = {1995},
isbn = {1558603859},
publisher = {Morgan Kaufmann Publishers Inc.},
address = {San Francisco, CA, USA},
booktitle = {Proceedings of the Eleventh Conference on Uncertainty in Artificial Intelligence},
pages = {435–443},
numpages = {9},
location = {Montr\'{e}al, Qu\'{e}, Canada},
series = {UAI'95}
}

@article{Richardson2014ACE,
title={ACE Bounds; SEMs with Equilibrium Conditions},
volume={29},
ISSN={0883-4237},
url={http://dx.doi.org/10.1214/14-STS485},
DOI={10.1214/14-sts485},
number={3},
journal={Statistical Science},
publisher={Institute of Mathematical Statistics},
author={Richardson, Thomas S. and Robins, James M.},
year={2014},
month=aug 
}

@article{robins1989analysis,
  title={The analysis of randomized and non-randomized AIDS treatment trials using a new approach to causal inference in longitudinal studies},
  author={Robins, James M},
  journal={Health service research methodology: a focus on AIDS},
  pages={113--159},
  year={1989},
  publisher={US Public Health Service}
}

@article{kitagawa2009identification,
  title={Identification region of the potential outcome distributions under instrument independence},
  author={Kitagawa, Toru},
  year={2009},
  publisher={Centre for Microdata Methods and Practice (Cemmap)}
}

@misc{angrist1995identification,
  title={Identification and estimation of local average treatment effects},
  author={Angrist, Joshua and Imbens, Guido},
  year={1995},
  publisher={National Bureau of Economic Research Cambridge, Mass., USA}
}

@misc{angrist1991sources,
  title={Sources of identifying information in evaluation models},
  author={Angrist, Joshua and Imbens, Guido},
  year={1991},
  publisher={National Bureau of Economic Research Cambridge, Mass., USA}
}

@article{chamberlain1986asymptotic,
  title={Asymptotic efficiency in semi-parametric models with censoring},
  author={Chamberlain, Gary},
  journal={journal of Econometrics},
  volume={32},
  number={2},
  pages={189--218},
  year={1986},
  publisher={Elsevier}
}

@misc{manski1998monotone,
  title={Monotone instrumental variables with an application to the returns to schooling},
  author={Manski, Charles F and Pepper, John V},
  year={1998},
  publisher={National Bureau of Economic Research Cambridge, Mass., USA}
}

@inproceedings{bonet2001instrumentality,
  title={Instrumentality tests revisited},
  author={Bonet, Blai},
  booktitle={Proceedings of the Seventeenth conference on Uncertainty in artificial intelligence},
  pages={48--55},
  year={2001}
}

@article{gunsilius2021nontestability,
  title={Nontestability of instrument validity under continuous treatments},
  author={Gunsilius, Florian F},
  journal={Biometrika},
  volume={108},
  number={4},
  pages={989--995},
  year={2021},
  publisher={Oxford University Press}
}

@article{beresteanu2012partial,
  title={Partial identification using random set theory},
  author={Beresteanu, Arie and Molchanov, Ilya and Molinari, Francesca},
  journal={Journal of Econometrics},
  volume={166},
  number={1},
  pages={17--32},
  year={2012},
  publisher={Elsevier}
}

@article{cheng2006bounds,
  title={Bounds on causal effects in three-arm trials with non-compliance},
  author={Cheng, Jing and Small, Dylan S},
  journal={Journal of the Royal Statistical Society Series B: Statistical Methodology},
  volume={68},
  number={5},
  pages={815--836},
  year={2006},
  publisher={Oxford University Press}
}

@book{boyd2004convex,
  title={Convex optimization},
  author={Boyd, Stephen and Vandenberghe, Lieven},
  year={2004},
  publisher={Cambridge university press}
}

@book{manski2003partial,
  title={Partial identification of probability distributions},
  author={Manski, Charles F},
  year={2003},
  publisher={Springer}
}

@article{manski1990nonparametric,
  title={Nonparametric bounds on treatment effects},
  author={Manski, Charles F},
  journal={The American Economic Review},
  volume={80},
  number={2},
  pages={319--323},
  year={1990},
  publisher={JSTOR}
}

@article{frechet1935generalisation,
  title={G{\'e}n{\'e}ralisation du th{\'e}oreme des probabilit{\'e}s totales},
  author={Fr{\'e}chet, Maurice},
  journal={Fundamenta mathematicae},
  volume={25},
  number={1},
  pages={379--387},
  year={1935},
  publisher={Polska Akademia Nauk. Instytut Matematyczny PAN}
}

@article{shu2025identification,
  title={Identification of Probabilities of Causation: A Complete Characterization},
  author={Shu, Xin and Wang, Shuai and Li, Ang},
  journal={arXiv preprint arXiv:2505.15274},
  year={2025}
}

@inproceedings{finkelstein2020deriving,
  title={Deriving bounds and inequality constraints using logical relations among counterfactuals},
  author={Finkelstein, Noam and Shpitser, Ilya},
  booktitle={Conference on uncertainty in artificial intelligence},
  pages={1348--1357},
  year={2020},
  organization={PMLR}
}

@article{duarte2024automated,
  title={An automated approach to causal inference in discrete settings},
  author={Duarte, Guilherme and Finkelstein, Noam and Knox, Dean and Mummolo, Jonathan and Shpitser, Ilya},
  journal={Journal of the American Statistical Association},
  volume={119},
  number={547},
  pages={1778--1793},
  year={2024},
  publisher={Taylor \& Francis}
}

\newpage
\appendix
\begin{center}
    \centering
    \bfseries
    \LARGE
    Appendices
\end{center}
\vspace{1em}
These appendices are organized as follows:
\begin{itemize}
\item The proof of \Cref{prop: suffciency of constraints} along with some preliminaries for the rest of our results is provided in \Cref{sec: appendix preliminaries}.

\item The proofs of Theorems \ref{thm:main PI result} and \ref{thm: PI ATE bounds} are provided in \Cref{sec: proof of pi}.

\item The proof of Theorem \ref{thm: iv inequalities} is provided in \Cref{sec: proof of iv}.

\item The proofs of Theorems \ref{cor: exponential lower bound inequalities} and \ref{prop: exponential inequalities general} are provided in \Cref{sec: proof of generalizations}.

\end{itemize}

\section{Preliminaries} \label{sec: appendix preliminaries}

\subsection{Proof of Proposition \ref{prop: suffciency of constraints}}

\prprelation*
\begin{proof} 
\label{proof: sufficiency}
    $(\Rightarrow)$
    Consider the first equation for a fixed $y$ and $z$.
    The right-hand-side summation is equal to $\mathcal{Q}(\Y{0}=y,\Dv{z}=0)$.
    Under \Cref{asm: independence}, the latter is also equal to $\mathcal{Q}(\Y{0}=y,\Dv{z}=0\mid Z=z)$, which, under consistency (\Cref{asm:consistency}) equals $\mathcal{P}(Y=y,D=0\mid Z=z)=p_{y0,z}$.
    The second set of equations hold similarly. The the inequalities $q_{ij,\mathbf{d}}\geq0$ hold because $\mathcal{Q}$ is a probability measure.
    
    $(\Leftarrow)$
    By Theorem~1 of \cite{song2025categoricalinstrumentalvariablemodel}, it suffices to show that the linear relations \Cref{eq: relation between q and p} imply the family of inequalities
    \begin{equation}
    \label{eq: song sufficient constraints}
    \mathcal{Q}\!\left(
    \Y0 \in \mathcal{V}^{(0)},\Y1 \in \mathcal{V}^{(1)}
    \right)
    \le
    \sum_{d=0}^{1}
    \mathcal{P}\!\left(Y \in \mathcal{V}^{(d)}, D=d \mid Z=z\right),
    \qquad \forall z\in[\ell],
    \end{equation}
    for every choice of nonempty sets
    $\mathcal{V}^{(0)},\mathcal{V}^{(1)}\subseteq\{\gamma_0,\dots,\gamma_{n-1}\}$.
    
    Fix an arbitrary $z\in[\ell]$.
    We show that \Cref{eq: song sufficient constraints} follows solely from \Cref{eq: relation between q and p} for this fixed $z$.
    
    \medskip
    \noindent Left-hand side:
    Expanding the definition of $\mathcal{Q}$ yields
    \begin{equation*}
    \mathcal{Q}\!\left(
    \Y0 \in \mathcal{V}^{(0)},\Y1 \in \mathcal{V}^{(1)}
    \right)
    =
    \sum_{\mathbf{d} \in\{0,1\}^{\ell}}
    \sum_{y_0\in\mathcal{V}^{(0)}}
    \sum_{y_1\in\mathcal{V}^{(1)}}
    q_{y_0 y_1, \mathbf{d}} .
    \end{equation*}
    
    \medskip
    \noindent Right-hand side:
    Using \eqref{eq: relation between q and p},
    \begin{equation*}
    \sum_{d=0}^{1}\sum_{y\in\mathcal{V}^{(d)}} p_{y d,z}
    =
    \sum_{y\in\mathcal{V}^{(0)}}
    \sum_{y_1}
    \sum\limits_{\substack{\mathbf{d} \in \{0,1\}^{\ell}, \\d_z=0}}
    q_{y y_1, \mathbf{d}}
    +
    \sum_{y\in\mathcal{V}^{(1)}}
    \sum_{y_0}
    \sum\limits_{\substack{\mathbf{d} \in \{0,1\}^{\ell}, \\d_z=1}}
    q_{y_0 y, \mathbf{d}}
    ,
    \end{equation*}
    
    Since all entries of $\vq$ are nonnegative, 
    Hence,
    \begin{equation*}
    \sum_{d=0}^{1}\sum_{y\in\mathcal{V}^{(d)}} p_{y d,z}
    \;\ge\;
    \sum_{\mathbf{d}\in\{0,1\}^{\ell}}
    \sum_{y_0\in\mathcal{V}^{(0)}}
    \sum_{y_{1}\in\mathcal{V}^{(1)}}
    q_{y_0 y_1, \mathbf{d}} .
    \end{equation*}
    
    The right-hand side equals the left-hand side of
    \eqref{eq: song sufficient constraints}, proving the inequality.
\end{proof}

\subsection{Vertices of \texorpdfstring{$\mathcal{H}$}{H}}
\label{sec: appendix preliminaries vertex}

We define the $M$-equivalence class of a vertex $\vv$ as
\[
    [\vv] = \{\, \vv + \vs \mid \vs \in \ker(M) \,\}.
\]
The \Cref{def: vertex} captures extremality after quotienting out the intrinsic affine directions induced by $\ker(M)$.

\begin{lemma}[Quotient characterization]
\label{lem: quotient_vertex_equiv}
Let $\pi : \R^{2\ell n} \to \R^{2\ell n} / \ker(M)$ denote the canonical projection, i.e., the map that identifies any two vectors $x,y \in \R^{2\ell n}$ such that $x - y \in \ker(M)$.
Then $\vv$ is a vertex of $\mathcal{H}$ in the sense of \Cref{def: vertex} if and only if $\pi(\vv)$ is an extreme point of the projected polyhedron $\pi(\mathcal{H})$.
\end{lemma}

\begin{proof}
($\Rightarrow$)
Assume $\vv$ is a vertex and suppose that
\[
\pi(\vv) = \lambda \pi(\vv_1) + (1-\lambda)\pi(\vv_2)
\]
for some $\lambda\in(0,1)$ and $\vv_1,\vv_2 \in\mathcal{H}$.
By linearity of $\pi$, this implies
\[
\pi\!\Big(\vv - \lambda \vv_1 - (1-\lambda)\vv_2\Big) = \mathbf{0},
\]
hence $\vv = \lambda \vv_1 + (1-\lambda)\vv_2 + \vs$ for some $\vs\in\ker(M)$.
Since $M\vs=0$, we also have $M(\vv_1+\vs)=M\vv_1, M(\vv_2+\vs)=M\vv_2$ and $\vv_1+\vs, \vv_2+\vs \in\mathcal H$ (because feasibility depends only on $M\cdot$).
Thus we can rewrite
\[
\vv = \lambda (\vv_1+\vs) + (1-\lambda)(\vv_2 + \vs)
\]
with both points in $\mathcal H$.
By \Cref{def: vertex}, it follows that
$M(\vv_1+\vs)=M\vv$ and $M(\vv_2 + \vs)=M\vv$, i.e. $\pi(\vv_1)=\pi(\vv)$ and $\pi(\vv_2)=\pi(\vv)$.
Therefore $\pi(\vv)$ is an extreme point of $\pi(\mathcal H)$.

($\Leftarrow$)
Conversely, assume that $\pi(\vv)$ is an extreme point of $\pi(\mathcal H)$ and write
\[
\vv = \lambda \vv_1 + (1-\lambda)\vv_2,
\qquad
\lambda\in(0,1),\quad \vv_1,\vv_2\in\mathcal H.
\]
Applying $\pi$ yields
\[
\pi(\vv) = \lambda \pi(\vv_1) + (1-\lambda)\pi(\vv_2).
\]
By extremality of $\pi(\vv)$, we get $\pi(\vv_1)=\pi(\vv_2)=\pi(\vv)$, which is equivalent to
$M\vv_1=M\vv_2=M\vv$.
This is exactly the condition in \Cref{def: vertex}.
\end{proof}

\begin{lemma}\label{lem: ps is zero}
    For any observed law $\mathcal{P}$ conforming to the IV and $\vs \in \ker(M)$,
    $$\vpt \vs = 0.$$
\end{lemma}

\begin{proof}
    Since $M^\top \vq = \vp$, 
    \[
        \vp^\top \vs = (M^\top \vq)^\top \vs 
        = \vq^\top M \vs\overset{(a)}{=} \vq^\top 0 = 0,
    \]
    where in (a) we used the fact that $\vs \in \ker(M)$, we have $M\vs = 0$. 
\end{proof}



Consequently, for any linear objective $\vp^\top \vv$, it suffices to consider \emph{one}
representative per vertex class $[\vv]$.

By the fundamental theorem of linear programming, whenever the primal is feasible and bounded,
the dual optimum over $\mathcal H$ is attained at an extreme point of $\pi(\mathcal H)$,
equivalently at a vertex $M$-equivalence class of $\vv$ of $\mathcal H$ in the sense of
\Cref{def: vertex}. Therefore, sharp identification bounds can be obtained by
evaluating $\vp^\top \vv$ over a finite set containing one representative from each vertex class.

\subsection{Extreme Rays of \texorpdfstring{$\mathcal{K}$}{K}}
\label{sec: appendix preliminaries extreme rays}

The set $\mathcal K := \{ \vr \in \mathbb{R}^{2\ell n} : M\vr \le \boldsymbol{0}\}$ is a polyhedral cone.
Moreover, by standard LP duality, feasibility of the primal system is equivalent to
\[
\vp^\top \vr \le 0 \quad \text{for all } \vr \in \mathcal K.
\]
Hence, sharp testable implications (i.e., a necessary and sufficient finite family of IV inequalities) can be obtained by
characterizing the extreme structure of $\mathcal K$.

\paragraph{\texorpdfstring{$M$}{M}-Equivalence under Scaling and Kernel Shifts.}
Because $\ker(M)\neq\{\boldsymbol 0\}$ (see \Cref{rem: s}), the cone $\mathcal K$ is generally not pointed and contains a lineality space.
In particular, $\ker(M)\subseteq \mathcal K$ and $\ker(M)\subseteq -\mathcal K$.
Accordingly, we work with rays modulo positive scaling and kernel shifts.

\begin{definition}[Extreme ray class]
\label{def: extreme_ray_class}
For nonzero $\vr,\vr' \in \mathcal K$, we write $\vr \sim \vr'$ if there exists $\lambda>0$ such that
\[
\lambda \vr - \vr' \in \ker(M).
\]
We call the $M$-equivalence class
\[
\mathcal{E}_{\vr} := \{ \lambda \vr + \vs : \lambda>0,\ \vs\in\ker(M)\}
\]
the \emph{ray class} generated by $\vr$.
\end{definition}

\begin{definition}[$M$-Distinct extreme rays]
\label{def: distinct of extreme rays}
Two extreme rays $\vr_1, \vr_2 \notin \ker(M)$ are called $M$-distinct if
\[
\mathcal{E}_{\vr_1} \neq \mathcal{E}_{\vr_2}.
\]
\end{definition}

\begin{lemma}
\label{lem: uniqueness of extreme rays}
For any two $M$-distinct extreme rays $\vr_1$ and $\vr_2$, we have
\[
\mathcal{E}_{\vr_1} \cap \mathcal{E}_{\vr_2} = \varnothing.
\]
\end{lemma}

\begin{proof}
Suppose, by contradiction, that there exists $\vr' \in \mathcal{E}_{\vr_1} \cap \mathcal{E}_{\vr_2}$.
Then there exist $\alpha_1,\alpha_2>0$ and $\vs_1,\vs_2\in\ker(M)$ such that
\[
\vr'=\alpha_1\vr_1+\vs_1
\qquad\text{and}\qquad
\vr'=\alpha_2\vr_2+\vs_2.
\]
Rearranging yields
\[
\vr_1
= \frac{\alpha_2}{\alpha_1}\,\vr_2
+ \frac{\alpha_2}{\alpha_1}(\vs_2-\vs_1),
\]
which implies $\vr_1 \in \mathcal{E}_{\vr_2}$.
By symmetry, $\vr_2 \in \mathcal{E}_{\vr_1}$, hence
$\mathcal{E}_{\vr_1}=\mathcal{E}_{\vr_2}$, a contradiction.
\end{proof}

The following lemma formalizes the fact that extremality is a property of the ray class and coincides with extremality
in the quotient cone.

\begin{lemma}[Quotient characterization of extreme rays]
\label{lem: quotient_extreme_rays}
Let $\pi:\R^{2\ell n}\to \R^{2\ell n}\setminus \ker(M)$ be the canonical projection and let $ \widehat{\mathcal K}:=\pi(\mathcal K)$.
Then $\vr\in\mathcal K\setminus\ker(M)$ is an extreme ray in the sense of \Cref{def: extreme ray}
if and only if $\pi(\vr)$ generates an extreme ray of the pointed polyhedral cone $ \widehat{\mathcal K}$.
Moreover, if $\vr\sim\vr'$, then $\pi(\vr)$ and $\pi(\vr')$ generate the same ray in $ \widehat{\mathcal K}$.
\end{lemma}

\begin{proof}
First, note that $\pi(\vr)=\pi(\vr')$ holds if and only if $\vr-\vr'\in\ker(M)$, hence $\vr\sim\vr'$ implies that
$\pi(\vr)$ and $\pi(\vr')$ lie on the same ray in $ \widehat{\mathcal K}$.
Now suppose that $\vr$ is not extreme in the sense of \Cref{def: extreme ray}.
Then there exist $\vr_1,\vr_2\in\mathcal K$ such that $\vr=\vr_1+\vr_2$ and neither summand belongs
to the ray class $\mathcal{E}_{\vr}$, i.e., $\vr_{t}\notin \{\alpha\vr+\vs:\alpha > 0,\ \vs\in\ker(M)\}$.
Applying $\pi$ yields $\pi(\vr)=\pi(\vr_1)+\pi(\vr_2)$ with $\pi(\vr_1),\pi(\vr_2)\in \widehat{\mathcal K}$,
and neither $\pi(\vr_{t})$ lies on the ray generated by $\pi(\vr)$.
Hence $\pi(\vr)$ does not generate an extreme ray of $ \widehat{\mathcal K}$.

Conversely, if $\pi(\vr)$ does not generate an extreme ray of $ \widehat{\mathcal K}$, then we can write
$\pi(\vr)= \widehat r_1+ \widehat r_2$ with $ \widehat r_1, \widehat r_2\in \widehat{\mathcal K}$ such that neither lies on the ray of $\pi(\vr)$.
Choose preimages $\vr_1,\vr_2\in\mathcal K$ with $\pi(\vr_{t})= \widehat r_{t}$.
Then $\pi(\vr-\vr_1-\vr_2)=0$, so $\vr=\vr_1+\vr_2+\vs$ for some $\vs\in\ker(M)$.
Since $\vs\in\mathcal K$ and $\mathcal K$ is a cone, we may absorb $\vs$ into, say, $\vr_2$ and obtain a decomposition
$\vr=\tilde{\vr}_{1}+\tilde{\vr}_{2}$ with $\tilde{\vr}_{1},\tilde{\vr}_{2}\in\mathcal K$ whose projections
do not lie on the ray of $\pi(\vr)$; equivalently, neither summand belongs to the ray class $\mathcal{E}_{\vr}$.
Thus $\vr$ is not extreme in the sense of \Cref{def: extreme ray}.
\end{proof}

\begin{corollary}[Finite generation and reduction to finitely many inequalities]
\label{cor: finite_extreme_rays_reduction}
The quotient cone $ \widehat{\mathcal K}=\pi(\mathcal K)$ is a pointed polyhedral cone and therefore admits only finitely many
extreme rays. Consequently, there exist $\vr_1,\dots,\vr_\omega\in\mathcal K$ such that
\[
\mathcal K = \ker(M) + \text{cone}(\vr_1,\dots,\vr_\omega),
\]
and every element of $\mathcal K$ can be written as
$\vs + \sum_{i=1}^\omega \lambda_i \vr_i$ for some $\vs\in\ker(M)$ and $\lambda_i\ge0$.

In particular, if $\vpt \vs = 0$ for all $\vs\in\ker(M)$,
then the family of inequalities
\[
\vpt \vr \le 0 \quad \forall \vr\in\mathcal K
\]
is equivalent to the finite family
\[
\vpt \vr_i \le 0
\qquad \forall i\in\{1,\dots,\omega\}.
\]
\end{corollary}

\begin{proof}
Since $\mathcal K$ is polyhedral, $ \widehat{\mathcal K}=\pi(\mathcal K)$ is also polyhedral. By construction it is pointed,
hence it has finitely many extreme rays; see, e.g., \cite[Section~19]{Rockafellar1970}.
Choosing one representative $\vr_i$ in each extreme-ray class and lifting back to $\mathcal K$ yields
$\mathcal K=\ker(M)+\text{cone}(\vr_1,\dots,\vr_\omega)$.

Finally, if $\vr=\vs+\sum_i \lambda_i \vr_i$, then
\[
\vp^\top \vr = \vp^\top \vs + \sum_i \lambda_i \vp^\top \vr_i
= \sum_i \lambda_i \vp^\top \vr_i,
\]
so $\vp^\top\vr\le0$ for all $\vr\in\mathcal K$ holds if and only if
$\vp^\top\vr_i\le0$ for all $i$.
\end{proof}

\paragraph{Lineality Space.}
Since $\ker(M)$ is a linear subspace, there exists finite nonzero vectors $\vs_1, \cdots \vs_t$
such that
\[
\ker(M) = \mathrm{span}(\vs_1, \cdots, \vs_t) = \text{cone}(\vs_1, \cdots, \vs_t) + \text{cone}(-\vs_1, \cdots, -\vs_t).
\]
Consequently,
\[
\mathcal K
=
\text{cone}(\vs_1, \cdots, \vs_t, -\vs_1, \cdots, -\vs_t, \vr_1, \dots, \vr_\omega),
\]
where $\vr_1,\dots,\vr_\omega$ are representatives of extreme-ray classes
in $\mathcal K \setminus \ker(M)$.
Moreover, note that $\vs_1, \cdots, \vs_t$ and $-\vs_1, \cdots, -\vs_t$ satisfy \Cref{def: extreme ray}, thus, they are extreme-rays
in $\mathcal K \cap \ker(M)$.

Therefore, if $\vpt \vs_i \leq 0$ and $- \vpt \vs_i \leq 0$ are satisfied for all $i \leq t$, then we have $\vpt \vs = 0$ for all $\vs\in\ker(M)$.
As a result, based on \Cref{cor: finite_extreme_rays_reduction}, 
\[
\vpt \vr \le 0 \quad \forall \vr\in\mathcal K
\]
is equivalent to the finite family
\begin{align*}
\vpt \vr_i &\le 0
\qquad \forall i\in\{1,\dots,\omega\} \\
\vpt \vs_i &\le 0
\qquad \forall i\in\{1,\dots,t\} \\
-\vpt \vs_i &\le 0
\qquad \forall i\in\{1,\dots,t\}.
\end{align*}


Consequently, sharp testable implications can be obtained by enumerating
one representative from each extreme-ray class along with $\vs_1, \cdots, \vs_t$ and $-\vs_1, \cdots, -\vs_t$ from the extreme ray class $\mathcal{E}_{0}$ which denote the $\ker(M)$.

\section{Vertex Enumeration}
\label{sec: proof of pi}

In this section, we provide the proofs and notations for \Cref{sec: pi}. First, in \Cref{app: PI definitions}, we introduce the basic definitions and notation required for the problem. In \Cref{app: PI basic properties}, we present fundamental lemmas that characterize the vertices and the structure of the dual LP matrix. In \Cref{app: PI b properties}, we define the signature binary number associated with a vertex $\vv$ and prove lemmas describing its structure. Next, in \Cref{app: PI bijection}, we establish the existence of a bijection between these signatures and the vertices. Finally, in \Cref{app: PI proof of main theorem}, we prove the main results, \Cref{thm:main PI result} and \Cref{thm: PI ATE bounds}, using the previously established lemmas.

\subsection{Definitions and Notations}\label{app: PI definitions}

In this subsection, we introduce the notation, definitions, and preliminary results required for the subsequent proofs.

\begin{definition}[Active Constraint Matrix]\label{def: vertex matrix}
For a point $\vv$, let $\Mv$ denote the submatrix of $M$ and 
$\vcv$ the corresponding subvector of $\vc$ associated with active (tight) constraints at $\vv$, i.e.,
$\Mv \vv = \vcv.$
We call $\Mv$ the \emph{active constraint matrix} at $\vv$.
\end{definition}

\begin{definition}\label{def: submatrixx}
Let $M(\alpha,\beta)$ denote the submatrix of $M$ obtained by restricting to
the rows indexed by $\alpha$ and the columns indexed by $\beta$.
Here, a row index $\alpha$ is of the form $y_0y_1d_0d_1$, and a column index
$\beta$ is of the form $ydz$, where each symbol may also be the wildcard $*$,
indicating that it ranges over all possible values.

By $M_1 \in M_2$, we mean that the row indices appearing in $M_1$ are a
subset of the row indices appearing in $M_2$, and the column indices appearing
in $M_1$ are a subset of the column indices appearing in $M_2$.

Also, by $M(\alpha_1 \cup \alpha_2, \beta_1\cup \beta_2)$ we mean the submatrix obtained by
restricting to all rows indexed by either $\alpha_1$ or $\alpha_2$, and to the
columns indexed by either $\beta_1$ or $\beta_2$. 

The same notation applies to $\Mv$. An example for this definition is provided in \cref{App: submatrix}.
\end{definition}

Based on the LP formulation, each variable $q_{y_0y_1,d_0d_1}$ is constrained by exactly
two $p_{yd,z}$s. More precisely,
\begin{itemize}
    \item $q_{y_0y_1,00}$ is constrained by $p_{y_0 0,0}$ and $p_{y_0 0,1}$,
    \item $q_{y_0y_1,01}$ is constrained by $p_{y_0 0,0}$ and $p_{y_1 1,1}$,
    \item $q_{y_0y_1,10}$ is constrained by $p_{y_0 0,1}$ and $p_{y_1 1,0}$,
    \item $q_{y_0y_1,11}$ is constrained by $p_{y_1 1,0}$ and $p_{y_1 1,1}$.
\end{itemize}
According to these relations between $\vp$ and $\vq$, the matrix $M$ will have a specific pattern outlined in the following remark.

\begin{remark}\label{rem: Q structure}
For $i,j\in[n]$, all entries equal to $1$ in $M$ occur in the following
positions and any other entry is eqaul to $0$:
\begin{itemize}
    \item $M(ij00,i00)$ and $M(ij00,i01)$,
    \item $M(ij01,i00)$ and $M(ij01,j11)$,
    \item $M(ij10,i01)$ and $M(ij10,j10)$,
    \item $M(ij11,j10)$ and $M(ij11,j11)$.
\end{itemize}
\end{remark}

\subsection{Structural Properties of Vertices}\label{app: PI basic properties}

In this subsection, we establish some fundamental structural properties of the constraint matrix $M$ and the active constraint matrix $\Mv$ associated with a vertex $\vv$. 

\begin{lemma}\label{lem: rank is 4n-1}
The rank of matrix $M$ is $4n-1$.
\end{lemma}

\begin{proof}
By \Cref{rem: Q structure}, every row of $M$ contains exactly one entry equal
to $1$ in $M(****,**0)$ and exactly one entry equal to $1$ in $M(****,**1)$.
Therefore, the sum of all columns of $M(****,**0)$ minus the sum of all columns of
$M(****,**1)$ gives a nontrivial linear combination of columns equal to zero.
Hence, $M$ cannot have full rank, and
\[
\rank(M)\le 4n-1.
\]

Now suppose that there exists a nontrivial linear combination of columns of $M$
that is equal to zero.
Then at least one column must have a nonzero coefficient.
Without loss of generality, assume that the coefficient assigned to the column
$M(****,i00)$ is $c_{i 0 0} = s\neq 0$.

For any $j\in[n]$, the columns $M(****,i00)$ and $M(****,j11)$ share a $1$ in some row of
$M$.
Since each row of $M$ contains exactly two entries equal to $1$, the
corresponding coefficient must satisfy
\[
c_{j11}=-s
\]
in order for the linear combination to vanish on that row.
Applying the same argument to all such rows implies that
\[
c_{j11}=-s \quad \text{for all } j\in[n].
\]

Next, since each column $M(****,j11)$ shares $1$’s with the column $M(****,j00)$,
the same reasoning shows that
\[
c_{j00}=s \quad \text{for all } j\in[n].
\]

Similarly, for any $i\in[n]$, the columns $M(****,i00)$ and $M(****,i01)$ share $n$
entries equal to $1$ in the rows indexed by $M(i*00,i0*)$.
Thus, to cancel these contributions we must have
\[
c_{i01}=-s \quad \text{for all } i\in[n].
\]
By the same argument, since $M(****,i01)$ and $M(****,i10)$ share $1$’s, we obtain
\[
c_{i10}=s \quad \text{for all } i\in[n].
\]

Therefore, in any nontrivial linear dependence among the columns of $M$, all
columns appear with coefficient either $s$ or $-s$, and this dependence is
unique up to scaling.
Consequently, the null space of $M$ is one-dimensional, and hence
\[
\rank(M)=4n-1.
\]
\end{proof}
\begin{remark}\label{rem: s}
As shown in the proof of \Cref{lem: rank is 4n-1}, the null space of $M$ is
one-dimensional.
Moreover, any vector $\vs=(s_{ydz})_{y\in[n],\,d,z\in\{0,1\}}$ in $\ker(M)$
must satisfy
\[
s_{yd0}=s,\qquad
s_{yd1}=-s
\quad\text{for all } y\in[n],d\in\{0,1\},
\]
for some scalar $s\in\mathbb{R}$. 
\end{remark}

\begin{lemma}\label{lem: vert iff rank is 4n-1}
Let $\vv$ be a feasible point satisfying $M\vv\le \vc$.
Then $\vv$ is a vertex if and only if
\[
\rank(\Mv)=4n-1.
\]
\end{lemma}

\begin{proof}
We prove both directions.

\medskip
\noindent
\textbf{($\Leftarrow$) If $\rank(\Mv)<4n-1$, then $\vv$ is not a vertex.}

Suppose $\rank(\Mv)<4n-1$.
Since 
$\rank(M)=4n-1$, 
we have
\[
\dim ( \ker(\Mv) ) > \dim (\ker(M)).
\]
Hence, there exists a nonzero vector $\vt \in \ker(\Mv)$ such that
$\vt \notin \ker(M)$.

Since $\vt\in\ker(\Mv)$, we have
\[
\Mv(\vv+\epsilon\vt)=\Mv\vv=\vcv
\]
for all $\epsilon\in\mathbb{R}$.
For sufficiently small $\epsilon>0$, the inequalities corresponding to
non-active constraints remain strict, because they are strict at $\vv$.
Therefore,
\[
M(\vv+\epsilon\vt)\le \vc
\quad\text{and}\quad
M(\vv-\epsilon\vt)\le \vc
\]
for sufficiently small $\epsilon$.

Thus both $\vv+\epsilon\vt$ and $\vv-\epsilon\vt$ are feasible.
Moreover,
\[
\vv=\tfrac12(\vv+\epsilon\vt)+\tfrac12(\vv-\epsilon\vt),
\]
and since $\vt\notin\ker(M)$, we have
\[
M(\vv+\epsilon\vt)\neq M\vv,
\]
so the two feasible points are $M$-distinct.

Hence $\vv$ can be written as a non-trivial convex combination of two
$M$-distinct feasible points, and therefore $\vv$ is not a vertex by
\Cref{def: vertex}.

\medskip
\noindent
\textbf{($\Rightarrow$) If $\rank(\Mv)=4n-1$, then $\vv$ is a vertex.}

Now suppose $\rank(\Mv)=4n-1$.
Assume
\[
\vv=\lambda\vvo+(1-\lambda)\vvt,
\qquad
\lambda\in(0,1),
\]
for feasible points $\vvo,\vvt$.

Since $\vv$ is active on $\Mv$, we have
\[
\Mv\vv=\vcv.
\]
Applying $\Mv$ to the convex combination gives
\[
\Mv\vv
=\lambda\Mv\vvo+(1-\lambda)\Mv\vvt.
\]
Because $\vvo$ and $\vvt$ are feasible,
\[
\Mv\vvo\le \vcv,
\qquad
\Mv\vvt\le \vcv.
\]
Since the convex combination equals $\vcv$, it follows that
\[
\Mv\vvo=\vcv
\quad\text{and}\quad
\Mv\vvt=\vcv.
\]
Therefore,
\[
\Mv(\vv-\vvo)=\mathbf{0},
\qquad
\Mv(\vv-\vvt)=\mathbf{0}.
\]

Because $\rank(\Mv)=4n-1=\rank(M)$, and the fact the $\Mv$ is a submatrix of $M$ we have
\[
\ker(\Mv)=\ker(M).
\]
Hence,
\[
\vv-\vvo\in\ker(M),
\qquad
\vv-\vvt\in\ker(M).
\]

By \Cref{def: vertex}, any two feasible points differing by a vector
in $\ker(M)$ represent the same extreme point.
Thus $\vv$ cannot be written as a non-trivial convex combination of two
$M$-distinct feasible points, and therefore $\vv$ is a vertex.
\end{proof}

From now on, we assume that $\vv$ represents a vertex. This means that
$M\vv \le \vc$ and $\Mv\vv = \vcv$, and $\rank(\Mv)=4n-1$ (by \Cref{lem: rank is 4n-1}).

\begin{remark}\label{rem: uniqness of Mv}
    Since $\rank(M)=\rank(\Mv)=4n-1$, the system $\Mv \vv = \vc$ has a unique solution in the affine
    space $\vv + \ker(M)$. Therefore, $\Mv$ uniquely determines the
    $M$-equivalence class $[\vv]$.

    In particular, if two active constraint matrices share the same
    row basis (that is, they consist of the same linearly independent
    rows), then they determine the same solution $\vv$ and thus the same
    $M$-equivalence class $[\vv]$. Consequently, the two active constraint
    matrices must coincide.
\end{remark}

\begin{lemma} \label{lem: first and last n^2 rows}
    If $M(ij00,***)\in \Mv$ then $j=0$. Also if $M(ij11,***) \in \Mv$ then $i=n-1$.
\end{lemma}
\begin{proof}
    Suppose that $M(ij00,***)\in \Mv$ for some $j>0$, then
    \[
        M(ij00,***)\vv = \gamma_j-\gamma_i > \gamma_0 - \gamma_i .
    \]
    On the other hand, based on the structure of $M$ we know that
    $M(ij00,***) = M(i000,***)$, so
    \[
        M(i000,***)\vv > \gamma_0 - \gamma_i ,
    \]
    which contradicts the fact that $\vv$ is a vertex. Hence $j=0$.

    Now suppose that $M(ij11,***)\in \Mv$ for some $i<n-1$. Then
    \[
        M(ij11,***)\vv = \gamma_j-\gamma_i < \gamma_j - \gamma_{n-1}.
    \]
    Again, by the structure of $M$ we have
    $M(ij11,***) = M((n-1)j11,***)$, and therefore
    \[
        M((n-1)j11,***)\vv < \gamma_j - \gamma_{n-1},
    \]
    which contradicts the fact that $\vv$ is a vertex. Hence $i=n-1$.
\end{proof}

\begin{lemma}\label{lem: one 1 in each column}
For any vertex $\vv$, every column of $\Mv$ contains at least one entry equal to $1$.
\end{lemma}
\begin{proof}
We know that the null space of $\Mv$ is one-dimensional and is spanned by some non-zero vector $\vs$ (see \Cref{rem: s}).
Suppose, by contradiction, that there exists a column indexed by $(y,d,z)$ in $\Mv$ that contains no entries equal to $1$.
Then the corresponding standard basis vector ${\ve}_{ydz}$ satisfies
\[
\Mv \ve_{ydz} = \mathbf{0}.
\]
Hence $\ve_{ydz}$ lies in the null space of $\Mv$.

Since $\ve_{ydz}$ is linearly independent from $\vs$, the null space of $\Mv$ has dimension at least $2$.
This contradicts the fact that $\vv$ is a vertex based on \Cref{lem: vert iff rank is 4n-1}. Hence every column of $\Mv$ must contain at least one entry equal to $1$.
\end{proof}

\subsection{Signature of a Vertex}\label{app: PI b properties}
In this section, we define the signature of a vertex and prove
its fundamental properties. These results will help us show that the signature is
unique for each vertex and that it can be used to reconstruct the
vertex.

\begin{definition}[Signature of a Point]\label{def: Bx}
We define signature $\vBu \in \{0,1\}^{4n}$ for a point $\vu$ as the column-wise logical OR of the rows of
$\Mu(**01\cup **10,***)$.
That is, for each index $(y,d,z)$, the entry $\Bu_{ydz}$ equals $1$ if there exists at least one entry equal to $1$ in the column $\Mu(**01\cup**10,ydz)$, and equals $0$ otherwise.
\end{definition}

\begin{lemma}\label{lem: first and last n rows based on B}
For any $y,z$, if $\Bv_{y0z}=0$, then $M(y000,***)\in \Mv$.
Similarly, if $\Bv_{y1z}=0$, then $M((n-1)y11,***)\in \Mv$.
\end{lemma}

\begin{proof}
By \Cref{lem: one 1 in each column}, every column of $\Mv$ must contain at least one entry equal to $1$.
Fix indices $y,z$ and suppose that $\Bv_{ydz}=0$.
By definition of $\vBv$, this means that there is no entry equal to $1$ in $\Mv(**01\cup**10,ydz)$.
Consequently, the required $1$ in column $\Mv(****,ydz)$ must lie in $\Mv(**00\cup **11,ydz)$.

First, assume that $d=0$.
By the structure of $M$ (see \Cref{rem: Q structure}), the only rows of $\Mv(**00\cup **11,y0z)$ that can contain a $1$ are the rows $M(yj00,***)$.
Hence, at least one such row must belong to $\Mv$.
By \Cref{lem: first and last n^2 rows}, this is only possible when $j=0$, and therefore $M(y000,***)\in \Mv$.

The case $d=1$ follows by a similar argument: if $\Bv_{y1z}=0$, then the $1$ in column $(y1z)$ must lie in $\Mv(**11,***)$, and by the structure of $M$ the only possible rows are $M(ij11,***)$.
Applying \Cref{lem: first and last n^2 rows} yields $i=n-1$, and hence $M((n-1)y11,***)\in \Mv$.
\end{proof}

\begin{lemma}\label{lem: full 1 columns}
 For any $i,j\in[n]$, $\Bv_{i00}=\Bv_{j11}=1$ if and only if $M(ij01,***)\in \Mv$.
Similarly,  for any $i,j\in[n]$, $\Bv_{i01}=\Bv_{j10}=1$ if and only if $M(ij10,***)\in \Mv$.
\end{lemma}
\begin{proof}
    The “$\Leftarrow$” direction in both statements follows immediately from
\Cref{def: Bx}, since the presence of a row in $\Mv$ forces a $1$ in the
corresponding columns.

    We prove the “$\Rightarrow$” directions.
    
\medskip
\noindent\textbf{First case.}
Assume $\Bv_{i00}=\Bv_{j11}=1$.
By \Cref{def: Bx} and the fact that $\Mv(**10,*00\cup*11)$ contains no entries equal to $1$,
there must exist at least one row in $\Mv(**01,***)$ containing required $1$ entries.

Suppose, by contradiction, that $M(ij01,***)\notin \Mv$.
Then there exist indices $a\neq i$ and $b\neq j$ such that
\[
M(aj01,***),\; M(ib01,***) \in \Mv.
\]
(If no such $a$ existed, $M(**01,j11)$ would contain no $1$(see \Cref{rem: Q structure}), contradicting
$\Bv_{j11}=1$; the same argument applies to $b$).

Since these rows belong to $\Mv$, we have
\[
M(aj01,***)\vv = \gamma_j - \gamma_a,
\qquad
M(ib01,***)\vv = \gamma_b - \gamma_i,
\]
and therefore
\[
\bigl(M(aj01,***) + M(ib01,***)\bigr)\vv
= \gamma_j - \gamma_i + \gamma_b - \gamma_a.
\]

On the other hand, since $M(ij01,***)\notin \Mv$, we obtain
\[
M(ij01,***)\vv < \gamma_j - \gamma_i,
\qquad
M(ab01,***)\vv \le \gamma_b - \gamma_a,
\]
and therefore,
\[
\bigl(M(ij01,***) + M(ab01,***)\bigr)\vv
< \gamma_j - \gamma_i + \gamma_b - \gamma_a = \bigl(M(aj01,***) + M(ib01,***)\bigr)\vv.
\]

By the structure of $M$ (see \Cref{rem: Q structure}),
\[
M(aj01,***) + M(ib01,***) = M(ij01,***) + M(ab01,***),
\]
which yields a contradiction.

\medskip
\noindent\textbf{Second case.}
Assume $\Bv_{i01}=\Bv_{j10}=1$ and suppose $M(ij10,***)\notin \Mv$.
Since $\Mv(**01,*01\cup *10)$ contains no $1$’s, the required $1$’s must lie in $\Mv(**10,***)$.
Thus, there exist indices $a\neq i$ and $b\neq j$ such that
\[
M(aj10,***),\; M(ib10,***) \in \Mv.
\]

As before, using the structure of $M$ we have
\[
M(aj10,***) + M(ib10,***) = M(ij10,***) + M(ab10,***),
\]
and evaluating at $\vv$ gives
\begin{align*}
\gamma_j - \gamma_i + \gamma_b - \gamma_a
&= (M(aj10,***) + M(ib10,***))\vv
\\&= (M(ij10,***) + M(ab10,***))\vv
\\&< \gamma_j - \gamma_i + \gamma_b - \gamma_a,
\end{align*}
a contradiction.

Therefore, $M(ij10,***)\in \Mv$.
\end{proof}

\begin{lemma}\label{lem: no double zeros in respective columns}
For any $i\in[n]$, we have
\[
\Bv_{i00}+\Bv_{i01}>0
\quad\text{and}\quad
\Bv_{i10}+\Bv_{i11}>0.
\]
Equivalently, at least one of $\Bv_{i00},\Bv_{i01}$ and at least one of
$\Bv_{i10},\Bv_{i11}$ must be equal to $1$.
\end{lemma}

\begin{proof}
We prove the first statement; the second follows by an analogous argument.

Suppose, by contradiction, that $\Bv_{i00}=\Bv_{i01}=0$.
By \Cref{lem: first and last n rows based on B}, this implies that
$M(i000,***)\in\Mv$.

Moreover, $M(i000,***)$ is the only row of $\Mv$ that can contain a $1$ in
either column $\Mv(****,i00)$ or $\Mv(****,i01)$.
Indeed, since $\Bv_{i00}=\Bv_{i01}=0$, no row of $\Mv(**01\cup **10,***)$ contains a $1$
in these columns.
By the structure of $M$, the only remaining cells of $M$ that contain a $1$
in columns $\Mv(****,i00)$ or $\Mv(****,i01)$ are $M(**00,i*00\cup i*01)$.
By \Cref{lem: first and last n^2 rows}, among these rows only $M(i000,***)$ can
belong to $\Mv$.

Now consider the vector
\[
\mathbf{t} = \ve_{i00} - \ve_{i01}.
\]
Since the only row of $\Mv$ with nonzero entries in columns $\Mv(****,i00)$ or
$\Mv(****,i01)$ is $M(i000,***)$, and this row contains a $1$ in both columns,
we have $\Mv \mathbf{t} = \mathbf{0}$.
Thus, $\mathbf{t}$ lies in the null space of $\Mv$.

Together with the vector $\vs$, which spans the null space of $\Mv$ (see \cref{rem: s}), this implies that the null space of $\Mv$ is at least two-dimensional

Consequently, $\Mv$ cannot have rank $4n-1$, contradicting the assumption
that $\vv$ is a vertex.

Therefore, $\Bv_{i00}+\Bv_{i01}>0$.

The proof for $\Bv_{i10}+\Bv_{i11}>0$ is identical.
\end{proof}

\begin{lemma}\label{lem: the only possible equility of duals}
We have
\[
\Bv_{i00}=\Bv_{i01}=\Bv_{j10}=\Bv_{j11}=1
\]
if and only if $i=n-1$, $j=0$, and
\[
M((n-1)000,***),\; M((n-1)011,***) \in \Mv.
\]
\end{lemma}

\begin{proof}
We first prove the “$\Rightarrow$” direction.

Suppose 
\[
\Bv_{i00}=\Bv_{i01}=\Bv_{j10}=\Bv_{j11}=1.
\]
By \Cref{lem: full 1 columns}, this implies that
\[
M(ij01,***),\; M(ij10,***) \in \Mv.
\]
Hence
\[
M(ij01,***)\vv = \gamma_j - \gamma_i,
\qquad
M(ij10,***)\vv = \gamma_j - \gamma_i,
\]
and
\[
(M(ij01,***) + M(ij10,***))\vv = 2(\gamma_j - \gamma_i).
\]

By the structure of $M$ (see \Cref{rem: Q structure}), we have 
\[
M(ij01,***) + M(ij10,***) = M(i000,***) + M((n-1)j11,***).
\]
Therefore,
\[
2(\gamma_j - \gamma_i)
= (M(i000,***) + M((n-1)j11,***))\vv
\le (\gamma_0 - \gamma_i) + (\gamma_j - \gamma_{n-1}).
\]

Rearranging yields
\[
\gamma_j - \gamma_i \le \gamma_0 - \gamma_{n-1}.
\]
Since $\gamma_0 - \gamma_{n-1}$ is the maximum possible difference,
equality can occur only if $j=0$ and $i=n-1$.
In this case, equality in the above inequality forces
\[
M((n-1)000,***),\; M((n-1)011,***) \in \Mv,
\]
because
\[
(M((n-1)000,***) + M((n-1)011,***))\vv
= 2(\gamma_0 - \gamma_{n-1}).
\]

\medskip
We now prove the “$\Leftarrow$” direction.

Suppose $i=n-1$, $j=0$, and
\[
M((n-1)000,***),\; M((n-1)011,***) \in \Mv.
\]
By the structure of $M$, we have
\[
M((n-1)001,***) + M((n-1)010,***)
= M((n-1)000,***) + M((n-1)011,***).
\]
Evaluating at $\vv$ gives
\[
(M((n-1)001,***) + M((n-1)010,***))\vv
= 2(\gamma_0 - \gamma_{n-1}).
\]

Since each of
$M((n-1)001,***)\vv, M((n-1)010,***)\vv$
is bounded above by $\gamma_0 - \gamma_{n-1}$,
both must attain this value.
Thus,
\[
M((n-1)001,***),\; M((n-1)010,***) \in \Mv,
\]
which implies
\[
\Bv_{(n-1)00}=\Bv_{(n-1)01}=\Bv_{010}=\Bv_{011}=1.
\]
\end{proof}
\begin{remark}\label{rem: only possible equality of duals}
If $\Bv_{i00}=\Bv_{i01}$ for some $i<n-1$, then for all $j\in [n]$ we have $\Bv_{j10}\ne\Bv_{j11}$.
Similarly, if $\Bv_{j10}=\Bv_{j11}$ for some $j>0$, then for all $i\in [n]$ we have $\Bv_{i00}\ne\Bv_{i01}$.
\end{remark}
\begin{proof}
    This follows directly from \Cref{lem: the only possible equility of duals}, together with \Cref{lem: one 1 in each column}.
\end{proof}

\begin{lemma} \label{lem: at least one intersection}
    There exists at least one $i\in[n]$ and $d\in \{0,1\}$ such that 
    \[
    \Bv_{id0}=\Bv_{id1}=1.
    \]
\end{lemma}

\begin{proof}
Suppose, by contradiction, that for every $i\in[n]$ and every $d\in\{0,1\}$,
we have
\[
\Bv_{id0}\ne\Bv_{id1}.
\]
In particular, for every $i$,
\[
\Bv_{i00}+\Bv_{i01}=1
\quad\text{and}\quad
\Bv_{i10}+\Bv_{i11}=1,
\]
by \Cref{lem: no double zeros in respective columns}.

For every $i$,
since exactly one of $\Bv_{i00},\Bv_{i01}$ is equal to $0$,
\Cref{lem: first and last n rows based on B} implies that
\[
M(i000,***)\in \Mv,
\]
and, in particular for $i=n-1$,
\begin{equation}
    \label{eq:plem13a}
    M((n-1)000,***)\in \Mv.
\end{equation}
Similarly, since exactly one of $\Bv_{i10},\Bv_{i11}$ is equal to $0$,
again by \Cref{lem: first and last n rows based on B} we obtain
\[
M((n-1)i11,***)\in \Mv,
\]
and, for the choice of $i=0$,
\begin{equation}
    \label{eq:plem13b}
    M((n-1)011,***)\in \Mv.
\end{equation}
Equations \eqref{eq:plem13a} and \eqref{eq:plem13b} together with \Cref{lem: the only possible equility of duals} imply
\[
\Bv_{(n-1)00}=\Bv_{(n-1)01}=\Bv_{010}=\Bv_{011}=1,
\]
which contradicts the assumption that
$\Bv_{id0}+\Bv_{id1}=1$ for all $i,d$.

Therefore, there must exist at least one $i\in[n]$ and $d\in\{0,1\}$ such that
\[
\Bv_{id0}=\Bv_{id1}=1.
\]
\end{proof}

\begin{lemma}\label{lem: rank of Q(01|10) part}
Let $a_{dz}$ denote the number of indices $i$ such that $\Bv_{idz}=1$.
Then
\[
\rank(\Mv(**01,***)) \le a_{00}+a_{11}-1
\quad\text{and}\quad
\rank(\Mv(**10,***)) \le a_{01}+a_{10}-1.
\]
\end{lemma}

\begin{proof}
We prove the bound for $\Mv(**01,***)$; the proof for $\Mv(**10,***)$ is identical.

By definition of $\vBv$, the submatrix $\Mv(**01,***)$ has nonzero entries only in
columns indexed by $(i00)$ and $(i11)$ with $\Bv_{i00}=1$ or $\Bv_{i11}=1$ (as all entries in $\Mv(**10,*00\cup *11)$ are equal to $0$).
Hence, $\Mv(**01,***)$ has at most $a_{00}+a_{11}$ nonzero columns.

Moreover, by the structure of $M$, every row of $\Mv(**01,***)$ contains exactly
one entry equal to $1$ in a column of type $(i00)$ and exactly one entry equal
to $1$ in a column of type $(j11)$.
Consequently, the sum of all columns of type $(i00)$ is equal to the sum of all
columns of type $(j11)$.
This yields a nontrivial linear dependence among the columns of $\Mv(**01,***)$.

Therefore, the column space of $\Mv(**01,***)$ has dimension at most
$a_{00}+a_{11}-1$, and hence
\[
\rank(\Mv(**01,***)) \le a_{00}+a_{11}-1.
\]

The same argument, applied to columns of type $(i01)$ and $(i10)$, gives
\[
\rank(\Mv(**10,***)) \le a_{01}+a_{10}-1.
\]
\end{proof}

\begin{lemma}\label{lem: at least one one in Bs}
For any $d,z\in\{0,1\}$, there exists an index $i\in[n]$ such that
\[
\Bv_{idz}=1.
\]
\end{lemma}

\begin{proof}
It suffices to prove the statement for $(d,z)=(0,0)$; the other cases follow
by symmetry.

Suppose, by contradiction, that
\[
\Bv_{i00}=0 \quad \text{for all } i\in[n].
\]
By \Cref{lem: full 1 columns}, this implies that no row of the form
$M(i'j'01,***)$ can belong to $\Mv$.
Consequently, $\Mv(**01,***)$ is empty and for every $j\in[n]$,
\[
\Bv_{j11}=0.
\]

We now bound the rank of $\Mv$.
First, submatrix $\Mv(**01,***)$ is
empty and contributes zero to the rank.
By \Cref{lem: rank of Q(01|10) part}, the submatrix $\Mv(**10,***)$ has rank at most
\[
a_{01}+a_{10}-1 \le 2n-1.
\]

Next, by \Cref{lem: first and last n^2 rows}, at most $2n$ rows from
$M(**00\cup **11,***)$ can belong to $\Mv$.
Therefore, the total rank of $\Mv$ is at most
\[
(2n-1) + 2n = 4n-1.
\]
Equality can occur only if \emph{all} such rows belong to $\Mv$.

In particular, this forces
\[
M((n-1)000,***),\; M((n-1)011,***) \in \Mv.
\]
By \Cref{lem: the only possible equility of duals}, this implies
\[
\Bv_{(n-1)00}=\Bv_{(n-1)01}=\Bv_{010}=\Bv_{011}=1,
\]
which contradicts the assumption that $\Bv_{i00}=0$ for all $i$.

Hence, there exists at least one $i\in[n]$ such that $\Bv_{i00}=1$.
The same argument applies to any choice of $(d,z)$.
\end{proof}

\begin{lemma}\label{lem: after intersection structure}
Suppose that $\Bv_{i00}=\Bv_{i01}=1$ for some $i\in[n]$.
Then for every $j$ with $i<j\le n-1$, we have
\[
\Bv_{j00}=\Bv_{j01}=1
\quad\text{and}\quad
M(j000,***)\notin \Mv.
\]
Similarly, if $\Bv_{i10}=\Bv_{i11}=1$ for some $i\in[n]$, then for every $j$ with
$0\le j<i$,
\[
\Bv_{j10}=\Bv_{j11}=1
\quad\text{and}\quad
M((n-1)j11,***)\notin \Mv.
\]
\end{lemma}

\begin{proof}
We prove the first statement; the second follows by symmetry.

Suppose $\Bv_{i00}=\Bv_{i01}=1$ for some $i$.
If $i=n-1$, there is nothing to prove, so assume $i<n-1$.

By \Cref{rem: only possible equality of duals}, it is impossible that
$\Bv_{t10}=\Bv_{t11}=1$ for any $t$.
Hence, for every $t\in[n]$, by \Cref{lem: no double zeros in respective columns} exactly one of $\Bv_{t10},\Bv_{t11}$ is equal to $0$.
By \Cref{lem: first and last n rows based on B}, this implies
\[
M((n-1)t11,***)\in \Mv
\quad\text{for all } t\in[n].
\]

By \Cref{lem: at least one one in Bs}, there exist indices $a,b$ such that
\[
\Bv_{a10}=1
\quad\text{and}\quad
\Bv_{b11}=1.
\]
Since $\Bv_{i01}=1$ and $\Bv_{a10}=1$, and $\Bv_{i00}=1$ and $\Bv_{b11}=1$,
\Cref{lem: full 1 columns} implies
\[
M(ia10,***),\; M(ib01,***) \in \Mv,
\]
which yields
\[
M(ia10,***)\vv = \gamma_a-\gamma_i,
\qquad
M(ib01,***)\vv = \gamma_b-\gamma_i.
\]
Adding the sides gives
\begin{equation}
v_{i00}+v_{i01} + v_{a10}+v_{b11} = \gamma_a+\gamma_b-2\gamma_i.
\label{eq: X1}
\end{equation}

Also we have the inequality
\[
M(i000,***)\vv = v_{i00}+v_{i01} \le \gamma_0-\gamma_i,
\]
which implies
\begin{equation}
v_{a10}+v_{b11}\ge \gamma_a+\gamma_b-\gamma_i-\gamma_0.
\label{eq: X2}
\end{equation}

Now fix any $j$ with $i<j\le n-1$.
Suppose, by contradiction, that $\Bv_{j00}=0$ or $\Bv_{j01}=0$.
Then by \Cref{lem: first and last n rows based on B},
\[
M(j000,***)\in \Mv,
\]
so
\begin{equation}
v_{j00}+v_{j01}=\gamma_0-\gamma_j.
\label{eq: X3}
\end{equation}

Using inequalities for the rows $M(ja10,***)$ and $M(jb01,***)$, we obtain
\[
v_{j01}+v_{a10}\le \gamma_a-\gamma_j,
\qquad
v_{j00}+v_{b11}\le \gamma_b-\gamma_j.
\]
Adding and using \eqref{eq: X3} gives
\begin{equation}
v_{a10}+v_{b11}\le \gamma_a+\gamma_b-\gamma_j-\gamma_0.
\label{eq: X4}
\end{equation}

Combining \eqref{eq: X2} and \eqref{eq: X4} yields
\[
\gamma_a+\gamma_b-\gamma_i-\gamma_0
\le v_{a10}+v_{b11}
\le \gamma_a+\gamma_b-\gamma_j-\gamma_0,
\]
which is impossible since $\gamma_j>\gamma_i$.
Hence,
\[
\Bv_{j00}=\Bv_{j01}=1.
\]

Finally, if $M(j000, ***)\in \Mv$ while $\Bv_{j00}=\Bv_{j01}=1$, then equality holds in
\eqref{eq: X4}, which again contradicts \eqref{eq: X2} because $\gamma_j>\gamma_i$.
Therefore,
\[
M(j000, ***)\notin \Mv.
\]

This completes the proof of the first statement.
The second statement follows from a symmetric argument.




\end{proof}

\begin{lemma}\label{lem: B intersenct affect on first and last n rows}
Suppose $0\le i<n-1$ is the smallest index such that
\[
\Bv_{i00}=\Bv_{i01}=1.
\]
Then $M(i000,***)\in\Mv$.
Similarly, if $0<j\le n-1$ is the largest index such that
\[
\Bv_{j10}=\Bv_{j11}=1,
\]
then $M((n-1)j11,***)\in\Mv$.
\end{lemma}
\begin{proof}
First, we prove the first statement.

Let $i<n-1$ be the smallest index such that $\Bv_{i00}=\Bv_{i01}=1$.
By \Cref{rem: only possible equality of duals}, this implies that
\begin{equation}
\Bv_{y10}+\Bv_{y11}=1 \quad \text{for all } y\in[n].
\label{eq: Y1}
\end{equation}

By minimality of $i$ and \Cref{lem: after intersection structure}, we have
\begin{equation}
\Bv_{j00}+\Bv_{j01}=
\begin{cases}
1, & 0\le j<i,\\
2, & i\le j\le n-1.
\end{cases}
\label{eq: Y2}
\end{equation}
Let $a_{dz}$ denote the number of indices $k$ with $\Bv_{k dz}=1$.
From \eqref{eq: Y2} we obtain
\[
a_{00}+a_{01}=i+2(n-i)=2n-i,
\qquad
a_{10}+a_{11}=n
\]
by \eqref{eq: Y1}.
Hence, by \Cref{lem: rank of Q(01|10) part},
\begin{equation}
\rank\bigl(\Mv(**01\cup **10,***)\bigr)
\le (a_{00}+a_{01})+(a_{10}+a_{11})-2
= 3n-i-2.
\label{eq: Y3}
\end{equation}
By \Cref{lem: after intersection structure}, no row $M(j000,***)$ with $j>i$
can belong to $\Mv$.

Suppose, by contradiction, that $M(i000,***)\notin\Mv$.
Then the submatrix $\Mv(**00\cup **11,***)$ contains at most $n+i$ rows and therefore
has rank at most $n+i$.
Combining with \eqref{eq: Y3}, we obtain
\[
\rank(\Mv)\le (3n-i-2)+(n+i)=4n-2,
\]
which contradicts the fact that $\vv$ is a vertex and hence
$\rank(\Mv)=4n-1$.

Therefore, $M(i000,***)\in\Mv$.

The proof of the second statement is symmetric. Let $0<j\le n-1$ be the largest index such that
$\Bv_{j10}=\Bv_{j11}=1$.
Since $j>0$, by \Cref{rem: only possible equality of duals} we cannot have
$\Bv_{t00}=\Bv_{t01}=1$ for any $t$, and hence
\begin{equation}
\Bv_{t00}+\Bv_{t01}=1 \quad \text{for all } t\in[n].
\label{eq: Y4}
\end{equation}
By maximality of $j$ and \Cref{lem: after intersection structure}, we have
\begin{equation}
\Bv_{t10}+\Bv_{t11}=
\begin{cases}
2, & t\le j,\\
1, & t>j.
\end{cases}
\label{eq: Y5}
\end{equation}
Let $a_{dz}$ be as above. Then \eqref{eq: Y4} gives $a_{00}+a_{01}=n$, and \eqref{eq: Y5} yields
\[
a_{10}+a_{11}=2(j+1)+(n-(j+1))=n+j+1.
\]
Therefore, by \Cref{lem: rank of Q(01|10) part},
\begin{equation}
\rank\bigl(\Mv(**01\cup **10,***)\bigr)\le (a_{00}+a_{01})+(a_{10}+a_{11})-2
= (n)+(n+j+1)-2 = 2n+j-1.
\label{eq: Y6}
\end{equation}

By \Cref{lem: after intersection structure}, no row $M((n-1)t11,***)$ with $t<j$
can belong to $\Mv$.

If $M((n-1)j11,***)\notin \Mv$, then $\Mv(**00\cup **11,***)$ contains at most
$n+(n-(j+1))=2n-j-1$ rows, hence has rank at most $2n-j-1$.
Combining with \eqref{eq: Y6} gives
\[
\rank(\Mv)\le (2n+j-1)+(2n-j-1)=4n-2,
\]
contradicting $\rank(\Mv)=4n-1$ at a vertex. Thus $M((n-1)j11,***)\in \Mv$.

\end{proof}

\begin{lemma}\label{lem: double intersect}
The index $i=n-1$ is the smallest index such that
\[
\Bv_{i00}=\Bv_{i01}=1
\]
if and only if the index $j=0$ is the largest index such that
\[
\Bv_{j10}=\Bv_{j11}=1.
\]
\end{lemma}
\begin{proof}
It suffices to prove one direction, since the statement is symmetric.

Assume that $i=n-1$ is the smallest index such that
$\Bv_{i00}=\Bv_{i01}=1$.
By \Cref{rem: only possible equality of duals}, this implies that there
is no index $j>0$ for which $\Bv_{j10}=\Bv_{j11}=1$.

Suppose, by contradiction, that $\Bv_{010}\neq 1$ or $\Bv_{011}\neq 1$.
Then $\Bv_{010}+\Bv_{011}=1$.
For all $j\in[n]$, we therefore have $\Bv_{j10}+\Bv_{j11}=1$.
By \Cref{lem: rank of Q(01|10) part}, this implies
\[
\rank(\Mv(**01\cup **10,***)) \le 2n-1.
\]

Moreover, since there are at most $2n$ rows in $\Mv(**00\cup **11,***)$, we have
\[
\rank(\Mv(**00 \cup **11,***)) \le 2n.
\]
Consequently,
\[
\rank(\Mv) \le 4n-1,
\]
and equality can hold only if all such rows belong to $\Mv$.
In particular, this forces
\[
M((n-1)000,***),\; M((n-1)011,***) \in \Mv,
\]
which implies $\Bv_{010}=\Bv_{011}=1$ by \Cref{lem: the only possible equility of duals}.
This contradicts the assumption.

Therefore, $\Bv_{010}=\Bv_{011}=1$, and hence $j=0$ is the largest index such
that $\Bv_{j10}=\Bv_{j11}=1$.
\end{proof}

\subsection{Bijection Between Vertices and Their Signature}\label{app: PI bijection}

Here we show that there is a bijection between vertices and their signatures. Consequently, the vertices
can be enumerated via their signatures.

\begin{proposition}\label{pro: M and B unique} 
The active constraint matrix $\Mv$ uniquely determines $\vBv$, and conversely,
$\vBv$ uniquely determines $\Mv$.
\end{proposition}

\begin{proof}
One can construct $\vBv$ from $\Mv$ by \Cref{def: Bx}.

Now suppose that we are given $\vBv$ for some vertex $\vv$ and wish to recover
$\Mv$.
First, by \Cref{lem: full 1 columns}, the submatrix $\Mv(**01\cup **10,***)$ is uniquely
determined by $\vBv$.

It remains to determine which of the $2n$ rows in $M(**00\cup **11,***)$ belong to $\Mv$.
By \Cref{lem: at least one intersection}, there exists an index $y\in[n]$ and
$d\in\{0,1\}$ such that
\[
\Bv_{yd0}=\Bv_{yd1}=1.
\]

We first consider the case $d=0$ and let $i$ be the smallest index such that
$\Bv_{i00}=\Bv_{i01}=1$.

If $i=n-1$, then by \Cref{lem: double intersect} we have
$\Bv_{010}=\Bv_{011}=1$, and by \Cref{lem: the only possible equility of duals}
it follows that
\[
M((n-1)000,***),\; M((n-1)011,***) \in \Mv.
\]
Moreover, for any $0\le i'<n-1$ and $0<j'\le n-1$, we have
$\Bv_{i'00}+\Bv_{i'01}=1$ and $\Bv_{j'10}+\Bv_{j'11}=1$
(by \Cref{rem: only possible equality of duals}),
so exactly one entry in each pair is zero.
By \Cref{lem: first and last n rows based on B}, this implies
\[
M(i'000,***),\; M((n-1)j'11,***) \in \Mv.
\]
Hence, $\Mv$ is uniquely determined by $\vBv$ in this case.

Now suppose that $i<n-1$.
By \Cref{rem: only possible equality of duals}, for every $j\in[n]$ we have
$\Bv_{j10}+\Bv_{j11}=1$, and hence exactly one of them is zero.
By \Cref{lem: first and last n rows based on B}, this implies
\[
M((n-1)j11,***)\in\Mv \quad \text{for all } j\in[n].
\]
Furthermore, by \Cref{lem: after intersection structure}, for all $i'>i$ we have
$M(i'000,***)\notin\Mv$.
By minimality of $i$, for all $i'<i$ we have $\Bv_{i'00}+\Bv_{i'01}=1$, and hence
$M(i'000,***)\in\Mv$.
Finally, by \Cref{lem: B intersenct affect on first and last n rows}, we also
have $M(i000,***)\in\Mv$.
Thus, $\Mv$ is uniquely determined by $\vBv$ in this case as well.

We now consider the case $d=1$.
Let $j$ be the largest index such that $\Bv_{j10}=\Bv_{j11}=1$.

If $j=0$, then by \Cref{lem: double intersect} we reduce to the first case
considered above, and $\Mv$ can again be uniquely constructed from $\vBv$.

Finally, suppose that $j>0$.
By \Cref{rem: only possible equality of duals}, for every $i\in[n]$ we have
$\Bv_{i00}+\Bv_{i01}=1$, and hence exactly one of them is zero.
By \Cref{lem: first and last n rows based on B}, this implies
\[
M(i000,***)\in\Mv \quad \text{for all } i\in[n].
\]
Moreover, by \Cref{lem: after intersection structure}, for all $j'<j$ we have
$M((n-1)j'11,***)\notin\Mv$, while by maximality of $j$ and
\Cref{lem: B intersenct affect on first and last n rows} we have
$M((n-1)j11,***)\in\Mv$.
For $j'>j$, we have $\Bv_{j'10}+\Bv_{j'11}=1$, and hence
$M((n-1)j'11,***)\in\Mv$ by \Cref{lem: first and last n rows based on B}.

Therefore, in all cases, $\Mv$ is uniquely determined by $\vBv$.
This completes the proof.
\end{proof}

\begin{remark}
Since $[\vv]$ and $\Mv$ uniquely determine each other, it follows that
$[\vv]$ uniquely determines $\vBv$, and conversely $\vBv$ uniquely determines
$[\vv]$.
\end{remark}

We introduce the following set and claim that it precisely consists of
all  signatures associated with vertices.

\defineS*

\begin{proposition}\label{pro: Bv is in S}
$\vBv\in S$ for any vertex $\vv$.
\end{proposition}
\begin{proof}
Suppose that some $0\le t<n-1$ is the smallest index such that $\Bv_{t00}=\Bv_{t01}$.
Then \Cref{lem: after intersection structure}, \Cref{rem: only possible equality of duals}, and \Cref{lem: at least one one in Bs} imply that the first, second, and third conditions of $S_1$ hold, respectively.
Thus $\vBv\in S_1$.

Suppose that some $0< t\le n-1$ is the largest index such that $\Bv_{t10}=\Bv_{t11}$.
Then \Cref{lem: after intersection structure}, \Cref{rem: only possible equality of duals}, and \Cref{lem: at least one one in Bs} imply that the first, second, and third conditions of $S_3$ hold, respectively.
Thus $\vBv\in S_3$.

Finally, suppose that $\Bv_{i00}\ne\Bv_{i01}$ for all $0\le i<n-1$ and that $\Bv_{j10}\ne\Bv_{j11}$ for all $0< j\le n-1$.
By \Cref{lem: at least one intersection}, at least one of $\Bv_{010}=\Bv_{011}$ or $\Bv_{(n-1)00}=\Bv_{(n-1)01}$ must hold.
If one holds, then \Cref{lem: double intersect} implies the other, and hence the first condition of $S_2$ is satisfied.
Moreover, \Cref{lem: the only possible equility of duals} implies that the remaining two conditions of $S_2$ also hold.
Therefore $\vBv\in S_2$.
\end{proof}

\begin{definition}\label{def: valid}
A vector $\vB \in \{0,1\}^{4n}$ is called \emph{valid} if there exists a vertex $\vv$ such that $\vB = \vBv$.
\end{definition}

\defineU*

\begin{proposition}\label{pro: B in S is valid}
    Any $\vB\in S$ is valid and its corresponding vertex is $\vu(\vB)$.
\end{proposition}

\begin{proof}
In this proof, for every $\vB \in S$, we prove that vector $\vu(\vB)$ satisfies
all the required equalities and strict inequalities.
In particular, we prove that
\[
M\vu(\vB) \le \vc
\quad\text{and}\quad
\vB = \vB^{\vu(\vB)},
\]
which results the feasibility of $\vu(\vB)$ and its uniqueness with respect to $\vB$.
Finally, we verify that $\vu(\vB)$ is a vertex by showing that
\[
\rank(M_{\vu(\vB)}) = 4n-1,
\]
in view of \Cref{lem: vert iff rank is 4n-1}. During this proof, $\vu$ is used as shorthand for $\vu(\vB)$.

\medskip
\noindent
\textbf{Relevant constraints.}
Similar to \Cref{lem: first and last n^2 rows}, it is enough to check the inequality for rows
\[
M(i000,***) \quad\text{and}\quad M((n-1)j11,***)
\]
among the rows of $M(**00\cup **11,***)$, because all other rows correspond to strictly weaker inequalities.
Therefore, it suffices to verify the following inequalities:
\begin{itemize}
    \item $M(i000,***)\vu=u_{i00}+u_{i01}\le \gamma_0-\gamma_i$;
    \item $M((n-1)j11,***)\vu=u_{j10}+u_{j11}\le \gamma_j-\gamma_{n-1}$;
    \item $M(ij01,***)\vu=u_{i00}+u_{j11}\le \gamma_j-\gamma_i$;
    \item $M(ij10,***)\vu=u_{i01}+u_{j10}\le \gamma_j-\gamma_i$,
\end{itemize}
for all relevant indices.

\medskip
\noindent
\textbf{Equality structure.}
If $\B_{i00}=\B_{j11}=1$, then
\[
u_{i00}+u_{j11}=\gamma_j-\gamma_i,
\]
so the inequality corresponding to $M(ij01,***)$ holds with equality, independently
of the value of $\alpha$.
If $\B_{i01}=\B_{j10}=1$, then
\[
u_{i01}+u_{j10}=\gamma_j-\gamma_i,
\]
so the inequality corresponding to $M(ij10,***)$ also holds with equality.

To establish strict inequality whenever at least one of the corresponding
$\vB$-entries is zero, we verify that each component $u_{ydz}$ strictly
increases when $\B_{ydz}$ takes the value $1$ instead of $0$.
Since equality holds when both corresponding $\vB$-entries are equal to $1$,
this implies that the sum is strictly smaller than $\gamma_j-\gamma_i$
whenever at least one of them is zero.

\medskip
\noindent
\textbf{Case $\vB\in S_1$.}
Let $0\le t<n-1$ be the smallest index such that $\B_{t00}=\B_{t01}=1$.
Set
\[
\alpha=-\gamma_0-\gamma_t.
\]

\begin{itemize}
    \item For $i<t$, $\B_{i00}\neq\B_{i01}$, and hence
          $u_{i00}+u_{i01}=\gamma_0-\gamma_i$.
          For $i\ge t$, $\B_{i00}=\B_{i01}=1$, and
          \[
          u_{i00}+u_{i01}
          =-2\gamma_i-\alpha
          =\gamma_0-\gamma_i+(\gamma_t-\gamma_i).
          \]
          Thus equality holds if and only if $i=t$, and the inequality is strict
          for all $i>t$ (as $\gamma_t-\gamma_i<0$).
          Therefore,
          \[
          M(i000,***)\in\Mu \iff i\le t.
          \]

    \item For all $j\in[n]$, $\B_{j10}\neq\B_{j11}$.
          Consequently,
          \[
          u_{j10}+u_{j11}=\gamma_j-\gamma_{n-1},
          \]
          so equality holds for all $j$, and
          \[
          M((n-1)j11, ***)\in\Mu \quad \text{for all } j\in[n].
          \]

    \item If $\B_{i00}=0$, which occurs only for $i<t$, then
          \[
          -\gamma_i-\alpha>\gamma_0,
          \]
          which shows that $u_{i00}$ is strictly larger when $\B_{i00}=1$.
          If $\B_{j11}=0$, then
          \[
          \gamma_j+\alpha>-\gamma_{n-1},
          \]
          which shows that $u_{j11}$ is strictly smaller when $\B_{j11}=0$.
          Therefore,
          \[
          M(ij01,***)\in\Mu \iff \B_{i00}=\B_{j11}=1.
          \]

    \item If $\B_{i01}=0$, then
          \[
          -\gamma_i>\gamma_0+\alpha,
          \]
          so $u_{i01}$ is strictly larger when $\B_{i01}=1$.
          If $\B_{j10}=0$, then
          \[
          \gamma_j>-\gamma_{n-1}-\alpha,
          \]
          so $u_{j10}$ is strictly smaller when $\B_{j10}=0$.
          Hence,
          \[
          M(ij10,***)\in\Mu \iff \B_{i01}=\B_{j10}=1.
          \]
\end{itemize}

Thus $M\vu\le\vc$ and $\vBu=\vB$ for $\vB\in S_1$.

\medskip
\noindent
\textbf{Case $\vB\in S_2$.}
Set $\alpha=-\gamma_0-\gamma_{n-1}$.

\begin{itemize}
    \item For $i<n-1$, $\B_{i00}\neq\B_{i01}$, so
          $u_{i00}+u_{i01}=\gamma_0-\gamma_i$.
          For $i=n-1$, $\B_{(n-1)00}=\B_{(n-1)01}=1$, and hence
          $u_{(n-1)00}+u_{(n-1)01}=\gamma_0-\gamma_{n-1}$.
          Therefore,
          \[
          M(i000,***)\in\Mu \quad \text{for all } i\in[n].
          \]

    \item For $j>0$, $\B_{j10}\neq\B_{j11}$, which gives equality.
          For $j=0$, $\B_{010}=\B_{011}=1$, which also gives equality.
          Thus,
          \[
          M((n-1)j11, ***)\in\Mu \quad \text{for all } j\in[n].
          \]

    \item If $\B_{i00}=0$, then $i<n-1$ and
          \[
          -\gamma_i-\alpha>\gamma_0,
          \]
          showing that $u_{i00}$ is strictly larger when $\B_{i00}=1$.
          If $\B_{j11}=0$, then $j>0$ and
          \[
          \gamma_j+\alpha>-\gamma_{n-1},
          \]
          showing that $u_{j11}$ is strictly smaller when $\B_{j11}=0$.
          Hence,
          \[
          M(ij01,***)\in\Mu \iff \B_{i00}=\B_{j11}=1.
          \]

    \item If $\B_{i01}=0$, then $i<n-1$ and
          \[
          -\gamma_i>\gamma_0+\alpha,
          \]
          showing that $u_{i01}$ is strictly larger when $\B_{i01}=1$.
          If $\B_{j10}=0$, then $j>0$ and
          \[
          \gamma_j>-\gamma_{n-1}-\alpha,
          \]
          showing that $u_{j10}$ is strictly smaller when $\B_{j10}=0$.
          Hence,
          \[
          M(ij10,***)\in\Mu \iff \B_{i01}=\B_{j10}=1.
          \]
\end{itemize}

Thus $M\vu\le\vc$ and $\vBu=\vB$ for $\vB\in S_2$.

\medskip
\noindent
\textbf{Case $\vB\in S_3$.}
Let $0<t\le n-1$ be the largest index such that $\B_{t10}=\B_{t11}=1$.
Set
\[
\alpha=-\gamma_t-\gamma_{n-1}.
\]

\begin{itemize}
    \item For all $i\in[n]$, $\B_{i00}\neq\B_{i01}$, and hence
          $u_{i00}+u_{i01}=\gamma_0-\gamma_i$.
          Therefore,
          \[
          M(i000,***)\in\Mu \quad \text{for all } i\in[n].
          \]

    \item For $j>t$, $\B_{j10}\neq\B_{j11}$ and equality holds.
          For $j\le t$, $\B_{j10}=\B_{j11}=1$, so
          \[
          u_{j10}+u_{j11}
          =\gamma_j-\gamma_{n-1}+(\gamma_j-\gamma_t).
          \]
          Equality holds only when $j=t$.
          Hence,
          \[
          M((n-1)j11, ***)\in\Mu \iff j\le t.
          \]

    \item If $\B_{i00}=0$, then
          \[
          -\gamma_i-\alpha>\gamma_0,
          \]
          which shows that $u_{i00}$ is strictly larger when $\B_{i00}=1$.
          If $\B_{j11}=0$, then $j>t$ and
          \[
          \gamma_j+\alpha>-\gamma_{n-1},
          \]
          which shows that $u_{j11}$ is strictly smaller when $\B_{j11}=0$.
          Therefore,
          \[
          M(ij01,***)\in\Mu \iff \B_{i00}=\B_{j11}=1.
          \]

    \item If $\B_{i01}=0$, then
          \[
          -\gamma_i>\gamma_0+\alpha,
          \]
          showing that $u_{i01}$ is strictly larger when $\B_{i01}=1$.
          If $\B_{j10}=0$, then $j>t$ and
          \[
          \gamma_j>-\gamma_{n-1}-\alpha,
          \]
          showing that $u_{j10}$ is strictly smaller when $\B_{j10}=0$.
          Hence,
          \[
          M(ij10,***)\in\Mu \iff \B_{i01}=\B_{j10}=1.
          \]
\end{itemize}

Thus $M\vu\le\vc$ and $\vBu=\vB$ for $\vB\in S_3$.

\medskip
\noindent
It remains to prove that $\rank(\Mu)=4n-1$ in all cases.

To do so, we show that $\ker(\Mu)$ is one-dimensional.

Suppose that $\vs\in\ker(\Mu)$.
Then
\[
\Mu\vs=0.
\]

\medskip
\noindent
\textbf{Structure coming from $M(ij01,***)$.}

In all cases we have
\[
M(ij01,***)\in\Mu \iff \B_{i00}=\B_{j11}=1.
\]
Hence, whenever $\B_{i00}=\B_{j11}=1$, we obtain
\[
M(ij01,***)\vs
= s_{i00}+s_{j11}
=0.
\]

By the structure of $\vB$, there exists some index $i_0$ such that
$\B_{i_000}=1$.
Let
\[
s_{i_000}=a.
\]
Then for every $j$ such that $\B_{j11}=1$, the relation
$s_{i_000}+s_{j11}=0$ implies
\[
s_{j11}=-a.
\]
Substituting this back into $s_{i00}+s_{j11}=0$ for any $i$ with
$\B_{i00}=1$ gives
\[
s_{i00}=a.
\]

Therefore,
\[
s_{i00}=a \quad\text{if }\B_{i00}=1,
\qquad
s_{j11}=-a \quad\text{if }\B_{j11}=1.
\]

\medskip
\noindent
\textbf{Structure coming from $M(ij10,***)$.}

Similarly, in all cases we have
\[
M(ij10,***)\in\Mu \iff \B_{i01}=\B_{j10}=1.
\]
Hence, whenever $\B_{i01}=\B_{j10}=1$, we obtain
\[
M(ij10,***)\vs
= s_{i01}+s_{j10}
=0.
\]

Let
\[
s_{i_001}=a'
\]
for some $i_0$ with $\B_{i_001}=1$.
Then for every $j$ such that $\B_{j10}=1$ we have
\[
s_{j10}=-a',
\]
and consequently for every $i$ such that $\B_{i01}=1$,
\[
s_{i01}=a'.
\]

Thus,
\[
s_{i01}=a' \quad\text{if }\B_{i01}=1,
\qquad
s_{j10}=-a' \quad\text{if }\B_{j10}=1.
\]

\medskip
\noindent
\textbf{Determination of the remaining entries.}

If $\B_{y0z}=0$ for some $y\in[n]$ and $z \in \{0,1\}$, then by construction
\[
M(y000,***)\in\Mu.
\]
Hence
\[
M(y000,***)\vs
= s_{y00}+s_{y01}
=0.
\]
Since by the structure of $\vB$ we know that
\[
\B_{y00}+\B_{y01}>0,
\]
at least one of $s_{y00}$ or $s_{y01}$ is already determined
(either equal to $a$ or $a'$).
The equation $s_{y00}+s_{y01}=0$ therefore uniquely determines
the other entry.

Similarly, if $\B_{y1z}=0$, then
\[
M((n-1)y11,***)\in\Mu,
\]
and therefore
\[
s_{y10}+s_{y11}=0,
\]
which again uniquely determines the remaining component.

Consequently, every component of $\vs$ is determined in terms of
$a$ and $a'$.

\medskip
\noindent
\textbf{Case $\vB\in S_1$.}

In this case, we know that
\[
M(t000,***)\in\Mu
\]
for the index $t$ defined in $S_1$.
Hence
\[
s_{t00}+s_{t01}=0.
\]
Since $\B_{t00}=\B_{t01}=1$, we have
\[
s_{t00}=a,
\qquad
s_{t01}=a'.
\]
Thus
\[
a+a'=0,
\quad\text{so}\quad
a=-a'.
\]

Therefore all components of $\vs$ are determined by a single
parameter $a$, and the null space is one-dimensional.

\medskip
\noindent
\textbf{Case $\vB\in S_2$.}

Here we know that
\[
M((n-1)000,***)\in\Mu,
\]
so
\[
s_{(n-1)00}+s_{(n-1)01}=0.
\]
Since $\B_{(n-1)00}=\B_{(n-1)01}=1$, we have
\[
s_{(n-1)00}=a,
\qquad
s_{(n-1)01}=a',
\]
which implies
\[
a+a'=0.
\]
Again $a=-a'$, so $\vs$ depends on a single parameter.

\medskip
\noindent
\textbf{Case $\vB\in S_3$.}

In this case,
\[
M((n-1)t11,***)\in\Mu,
\]
for the index $t$ defined in $S_3$. So
\[
s_{t10}+s_{t11}=0.
\]
Since $\B_{t10}=\B_{t11}=1$, we have
\[
s_{t10}=-a',
\qquad
s_{t11}=-a.
\]
Thus
\[
-a-a'=0,
\quad\text{so again}\quad
a=-a'.
\]

\medskip
\noindent
In all three cases, $\vs$ is uniquely determined up to a scalar multiple.
Hence,
\[
\dim (\ker(\Mu))=1,
\]
and therefore
\[
\rank(\Mu)=4n-1.
\]

Finally, we have proven that for each $\vB \in S$ there exists a unique
active constraint matrix $\Mu$ and a unique vector $\vu$ such that
$\vu$ is a vertex.
Hence, every $\vB \in S$ is valid, and this validity is unique.
\end{proof}

\begin{remark}\label{rem: order of creating v}
    If $\vB$ is valid, we can recreate its corresponding vertex in $\mathcal{O}(n)$, which indicates the length of the vertex.
\end{remark}
\begin{proof}
Given $\vB$, we can explicitly construct the corresponding vertex $\vu$
using \Cref{def: U define}.
This construction requires only $\mathcal{O}(n)$ time.

Indeed, once $\vB$ is given, the only remaining parameter to determine
is the value of $\alpha$.
To compute this value, it suffices to find either
\begin{itemize}
    \item the smallest index $t$ such that $\B_{t00}=\B_{t01}$, or
    \item the largest index $t$ such that $\B_{t10}=\B_{t11}$.
\end{itemize}
Both quantities can be obtained by a single linear scan of $\vB$,
and therefore can be computed in $\mathcal{O}(n)$ time.
\end{proof}

\subsection{Proof of the Theorems \ref{thm:main PI result}, \ref{thm: PI ATE bounds}, and \ref{thm: jointindep}}\label{app: PI proof of main theorem}
\ateVertices*
\begin{proof}
By \cref{pro: Bv is in S}, every vertex has a signature in $S$. 
Therefore, any vertex is $M$-equivalent to a member of $\mathcal{V}$.
Conversely, by \cref{pro: B in S is valid}, every $\vB \in S$ has an associated vertex, and this vertex is unique. Indeed, no two $M$-distinct vectors produce the same signature $\vB$.
Thus, set $\mathcal{V}$ includes $M$-distinct vertices.
As a result, there is a bijection between the elements of $S$ and vertices and the proof is completed.
\end{proof}

\begin{proposition}
\label{prop: sharpness of vertices}
For every vertex $\vv$ of the dual feasible set, there exists a distribution $\mathcal{P}$ conforming to the IV model (equivalently, a corresponding $\mathcal{Q}$ that the pair $\mathcal{P}, \mathcal{Q}$ satisfy equations in \ref{eq: relation between q and p}) such that $\vv$ uniquely attains the maximum (or, in the minimization case, the minimum) value of the dual objective.
\end{proposition}

Let $\vv$ be a vertex of the feasible set of \Cref{eq: dual_lower_bound_formal}.
By the definition of a vertex (see \Cref{lem: vert iff rank is 4n-1}), there exists a collection of inequalities in \Cref{eq: dual_lower_bound_formal} that are tight at $\vv$ and whose corresponding equalities admit a unique vertex solution, namely $\vv$ itself
(see \Cref{rem: uniqness of Mv}).
Since \Cref{eq: dual_lower_bound_formal} is the dual of
\Cref{eq: primal lower bound PI}, each dual inequality corresponds to a primal variable
$q_{kl,ij}$ for some $i,j \in \{0,1\}$ and $k,l \in [n]$.

\begin{proof} 
\label{proof: sharpness}
    Fix a vertex $\vv$, and let $u$ denote the number of inequalities in \Cref{eq: dual_lower_bound_formal} that are tight at $\vv$.
    By construction of the dual via the Lagrangian, each such inequality is associated with a primal variable $q_{\y0\y1,\dv0\dv1}$.
    
    We define a distribution $\mathcal{Q}$ as follows:
    for each inequality that holds with equality at $\vv$, set the corresponding $q_{\y0\y1,\dv0\dv1} = \frac{1}{u}$, and set all remaining probabilities in the support of $\mathcal{Q}$ to zero.
    This defines a valid distribution supported exactly on the tight inequalities at $\vv$.
    
    Next, rewrite the dual objective in \Cref{eq: dual_lower_bound_formal} in terms of
    $q_{\y0\y1,\dv0\dv1}$:
    \begin{align*}
        \sum_{d,z \in \{0,1\},\, y \in [n]} p_{yd,z} \, v_{yd,z}
        =&
        \sum_{j,l,y} q_{yl,0j} \, v_{y0,0}
        + \sum_{j,k,y} q_{ky,1j} \, v_{y1,0} \\
        &+ \sum_{i,l,y} q_{yl,i0} \, v_{y0,1}
        + \sum_{i,k,y} q_{ky,i1} \, v_{y1,1},
    \end{align*}
    which can be equivalently expressed as
    \begin{gather}
    \label{eq: equivalent dual objective value}
        \sum_{k,l} q_{kl,00} \bigl( v_{k0,0} + v_{k0,1} \bigr)
        + \sum_{k,l} q_{kl,01} \bigl( v_{k0,0} + v_{l1,1} \bigr) \\
        + \sum_{k,l} q_{kl,10} \bigl( v_{l1,0} + v_{k0,1} \bigr)
        + \sum_{k,l} q_{kl,11} \bigl( v_{l1,0} + v_{l1,1} \bigr).
        \nonumber
    \end{gather}

    Each term in \Cref{eq: equivalent dual objective value} is a convex combination of expressions that are individually bounded above (or below) by the inequalities defining the dual feasible set.
    Since all nonzero coefficients $q_{\y0\y1,\dv0\dv1}$ equal $\frac{1}{u} > 0$, the dual objective is maximized (or minimized) if and only if every inequality corresponding to a nonzero $q_{\y0\y1,\dv0\dv1}$ is tight.
    
    By construction of $M$, this occurs precisely when the set of inequalities that are tight at $\vv$ hold with equality.
    In other words, $\vu$ can maximize the dual objective if and only if $\Mv \subseteq \Mu$.
    Since $\vv$ is a vertex, by \Cref{rem: uniqness of Mv}, $\Mu$ has the same basis as $\Mv$ and must be equal to it.
    Thus, the system of equalities represented by $\Mv$ admits a unique vertex solution, namely $\vv$ itself.
    Therefore, $\vv$ uniquely attains the extreme value of the dual objective, completing the proof.
\end{proof}

\ateTerms*
\begin{proof}
    We have enumerated all $M$-distinct vertices of the dual feasible set,
    which is sufficient. Moreover, by \Cref{prop: sharpness of vertices},
    each vertex provides a unique bound. Hence, every vertex contributes
    a necessary term to the bounds.

    Therefore, we obtain
    \[
        \max_{\vv \in \mathcal{V}} \vv^\top \vp = L(\mathcal{P})
        \;\le\;
        ATE
        \;\le\;
        -L(\Bar{\mathcal{P}}) = -\max_{\vv \in \mathcal{V}} \vv^\top \Bar{\vp}.
    \]
    based on Equations \ref{eq: primal lower bound PI} and \ref{eq: primal upper bound PI}.
\end{proof}

\prpjointindep*
\begin{proof}
    As established in Theorem 1 of \cite{song2025categoricalinstrumentalvariablemodel}, the set of admissible pairs $(\mathcal{P}, \mathcal{Q})$ under Assumptions \ref{asm: exclusion}, \ref{asm: independence}, and \ref{asm:consistency} (denoted by $\mathcal{M}_1$ in \cite{song2025categoricalinstrumentalvariablemodel}) is identical to the set of admissible pairs under Assumptions \ref{asm: exclusion}, \ref{asm: jointind}, and \ref{asm:consistency_a} (denoted by $\mathcal{M}_2$ in \cite{song2025categoricalinstrumentalvariablemodel}).

    More precisely, for every pair $(\mathcal{P}, \mathcal{Q})$ for which there exists a full data law $\mathcal{Q}^{f*}$ satisfying Assumptions \ref{asm: exclusion}, \ref{asm: independence}, and \ref{asm:consistency}, there exists a full data law $\mathcal{Q}^{f'}$ satisfying Assumptions \ref{asm: exclusion}, \ref{asm: jointind}, and \ref{asm:consistency_a} that induces the same observed distribution $\mathcal{P}$, and vice versa.
    
    Consequently, the sharp identification region for the ATE under the two sets of assumptions is identical. 
    In particular, the bounds in \eqref{eq:atebounds}, obtained via \eqref{eq: primal lower bound PI} and \eqref{eq: primal upper bound PI}, remain valid and sharp under Assumptions \ref{asm: exclusion}, \ref{asm: jointind}, and \ref{asm:consistency_a}.
\end{proof}

\section{Extreme Ray Enumeration}
\label{sec: proof of iv}

Note that \Cref{eq: general dual feasibility test_formal} 
can be also written as follows:
\begin{align}
    0 \geq
    \max \quad 
    &\Sigma_{d,z\in \{0,1\}, y \in [n] } p_{yd, z} x_{yd,z}
    \\
    \text{s.t.} \quad 
    & x_{k 0, 0} + x_{k 0, 1} \leq 0 \quad \forall k \in [n]
    \nonumber \\
    & x_{k 0, 0} + x_{l 1, 1} \leq 0 \quad \forall k, l \in [n]
    \nonumber \\
    & x_{l 1, 0} + x_{k 0, 1} \leq 0 \quad \forall k, l \in [n]
    \nonumber \\
    & x_{l 1, 0} + x_{l 1, 1} \leq 0 \quad \forall l \in [n]
    \nonumber 
\end{align}

In the definition below, we introduce a set of vectors $\mathcal{R}$.
We shall shortly prove that $\mathcal{R}$ is exactly the set of vectors that can generate the cone $\mathcal{K} := \{ \vx \in \R^{2\ell n} : M\vx \le \mathbf{0} \}$.
\begin{definition}
\label{def: pattern of extreme rays}
    Let $\mathcal{R}$ be the set of all vectors $\vr = \big( r_{ki,j}\big)_{k\in[n],i,j\in\{0,1\}}$ that belong to one of the following families:
    
    \medskip
    \noindent\textbf{(I) The all-$\pm 1$ rays.}
    One vector corresponding to each $s\in\{-1,1\}$, defined as
    \begin{equation*}
    r_{k i,1}=s,\qquad r_{k i,0}=-s,
    \qquad \forall\, k\in[n],\ i\in\{0,1\}.
    \end{equation*}
    \medskip
    \noindent\textbf{(II) Single-entry perturbations.}
    One vector corresponding to every $(k',i',j')\in [n]\times\{0,1\}\times\{0,1\}$ with $(k',i',j')\neq (n-1,1,1)$, defined as
    \begin{align*}
    \begin{cases}
        r_{k'i',j'}=-1,\\
        r_{k i,j}=0 \qquad \forall(k,i,j)\neq (k',i',j')
    \end{cases}
    \end{align*}
    \medskip
    \noindent\textbf{(III) One special coordinate set to zero.}
    \begin{align*}
    \begin{cases}
        r_{k1,0}=-1,\ \ r_{k0,0}=-1 \qquad& \forall k\in[n], \\
    r_{k1,1}=1\qquad& \forall k\in[n-1], \nonumber \\
    r_{k0,1}=1\qquad& \forall k\in[n], \nonumber \\
    r_{(n-1)1,1}=0 .
    \end{cases}
    \end{align*}
    \medskip
    \noindent\textbf{(IV) Rays indexed by a nonzero binary vector $\mathbf{s}\in\{0,1\}^{n-1}$.}
    One vector corresponding to every
    $\mathbf{s}\in\{0,1\}^{n-1}\setminus\{\mathbf{0}\}$, defined as
    \begin{align*}
    \begin{cases}r_{k0,1}=-1,\ \ r_{k0,0}=0 \quad &\forall k\in[n], \\
    r_{k1,0}=s_k\quad &\forall k\in[n-1],\nonumber \\
    r_{k1,1}=-r_{k1,0}\quad &\forall k\in[n],\nonumber \\
    r_{(n-1)1,0}=0.
    \end{cases}
    \end{align*}
    \medskip
    \noindent\textbf{(V) Rays indexed by a nontrivial binary vector $\mathbf{s}\in\{0,1\}^n$.}
    One vector corresponding to every $\mathbf{s}\in\{0,1\}^n\setminus\{\mathbf{0},\mathbf{1}\}$, defined as
    \begin{align*}
    \begin{cases}r_{k1,1}=0,\ \ r_{k1,0}=-1 \quad &\forall k\in[n], \\
    r_{k0,1}=s_k\quad &\forall k\in[n], \\
    r_{k0,0}=-r_{k0,1}\quad &\forall k\in[n].
    \end{cases}
    \end{align*}
    \medskip
    \noindent\textbf{(VI) Another family indexed by $\mathbf{s}\in\{0,1\}^{n-1}\setminus\{\mathbf{0}\}$.}
    One vector corresponding to every $\mathbf{s}\in\{0,1\}^{n-1}\setminus\{\mathbf{0}\}$, defined as
    \begin{align*}
    \begin{cases}r_{k0,1}=0,\ \ r_{k0,0}=-1 \quad &\forall k\in[n], \\
    r_{k1,1}=s_k\quad &\forall k\in [n-1],\\
    r_{k1,0}=-r_{k1,1}\quad &\forall k\in[n], \\
    r_{(n-1)1,1}=0.
    \end{cases}
    \end{align*}
\end{definition}
\begin{restatable}{proposition}{thmpattern}
\label{thm: pattern of extreme rays}
    The cone $\mathcal{K}$ is generated by cone combination of vectors $\vr\in \mathcal{R}$.
\end{restatable}
In order to prove \Cref{thm: pattern of extreme rays}, we characterize the extreme rays.
To do so, we will utilize a few preliminary results presented below.

\begin{definition}[Active Constraint Matrix]\label{def: vertex matrix IV version}
For a feasible point $\vr$ in $\mathcal{K}$, the active constraint matrix at $\vr$, denoted by $\Mr$, is the maximal (row-) submatrix of $M$ such that
\[
\Mr \vr = \mathbf{0}.
\]
\end{definition}

\begin{definition}[Positive Homogeneity]
A set defined by a system of equalities is positively homogeneous if it is closed under nonnegative scaling, i.e., if $\vr$ satisfies the defining equalities, then $\lambda \vr$ also satisfies them for all $\lambda \ge 0$.
\end{definition}

\begin{lemma}
\label{lem: rank of rays are 4n-2}
Let $\vr \neq \mathbf 0$ be a feasible point of the cone $\mathcal{K}:=\{\vx: M\vx \le \mathbf{0}\}$.
Then $\vr$ is an extreme ray of $\mathcal K$ if and only if the active contain matrix at least $4n-2$ linearly independent constraints, i.e.,
\begin{equation*}
    \rank(\Mr) \geq 4n-2
\end{equation*}
\end{lemma}
\begin{proof}
$(\Rightarrow)$
Fix an extreme ray $\vr$.
Suppose by contradiction that $\rank(\Mr)<4n-2$.  
Then the null space of $\Mr$ has dimension strictly larger than $1+\dim(\ker(M))$. 
Hence there exists a nonzero vector
\begin{equation*}
\vt \in \ker(\Mr)\setminus \mathrm{span}\!\left(\vr,\ker(M)\right).
\end{equation*}
Since $\vt\in\ker(\Mr)$, all inequalities that are tight at $\vr$ remain tight along the direction $\vt$. 
Therefore, for sufficiently small
$\epsilon>0$,
\begin{equation*}
M(\vr+\epsilon\vt)\le\mathbf{0},
\qquad
M(\vr-\epsilon\vt)\le\mathbf{0},
\end{equation*}
so both $\vr+\epsilon\vt$ and $\vr-\epsilon\vt$ belong to $\mathcal K$.
Moreover,
\begin{equation*}
\vr=\tfrac12(\vr+\epsilon\vt)+\tfrac12(\vr-\epsilon\vt).
\end{equation*}
Because $\vt\notin \mathrm{span}(\vr,\ker(M))$, at least one of the vectors $\vr\pm\epsilon\vt$ is not contained in $\mathrm{span}(\vr,\ker(M))$, contradicting \Cref{def: extreme ray}.
Hence $\rank(\Mr)\ge 4n-2$.

$(\Leftarrow)$
Suppose $\rank(\Mr)\ge 4n-2$. Then
\begin{equation*}
\dim(\ker(\Mr))\le 1+\dim(\ker(M)).
\end{equation*}
By construction, $\vr\in\ker(\Mr)$ and $\ker(M)\subseteq\ker(\Mr)$.

Let $\vr_1,\vr_2\in\mathcal K$ such that
\begin{equation*}
\vr=\vr_1+\vr_2 .
\end{equation*}
Applying $\Mr$ gives
\begin{equation*}
\mathbf{0}=\Mr\vr=\Mr\vr_1+\Mr\vr_2 .
\end{equation*}
Since $\vr_1,\vr_2\in\mathcal K$, we have $\Mr\vr_1\le\mathbf{0}$ and $\Mr\vr_2\le\mathbf{0}$, hence
\begin{equation*}
\Mr\vr_1=\Mr\vr_2=\mathbf{0},
\end{equation*}
which implies $\vr_1,\vr_2\in\ker(\Mr)$.

If $\vr\in\ker(M)$, then $\Mr=M$ by definition, and therefore $\vr_1,\vr_2\in\ker(M)$, satisfying \Cref{def: extreme ray}.
Otherwise, $\vr\notin\ker(M)$. 
Since $\ker(M)\subseteq\ker(\Mr)$ and
\begin{equation*}
\dim(\ker(\Mr))\le 1+\dim(\ker(M)),
\end{equation*}
we must have
\begin{equation*}
\ker(\Mr)=\mathrm{span}\!\left(\vr,\ker(M)\right).
\end{equation*}
Thus any feasible decomposition of $\vr$ lies in $\mathrm{span}\!\left(\vr,\ker(M)\right)$, which proves that $\vr$ is an extreme ray.
\end{proof}


According to \Cref{lem: rank is 4n-1} and \Cref{lem: rank of rays are 4n-2}, we have two sets of extreme rays, one set for those with $\rank(\Mr) = 4n-1$ and one set with $\rank(\Mr) = 4n-2$.

\begin{remark}[Invariant directions and trivial rays]
\label{rem: trivial ray}
Based on the structure of the cone $\mathcal{K} := \{\vx : M \vx \le \mathbf{0}\}$, if $\rank(\Mr) = \rank(M) = 4n-1$ holds for a vector $\vr$, then $M\vr = \mathbf{0}$.
Therefore, only the following vectors can satisfy $\rank(\Mr) = 4n-1$:
\begin{align*}
    \{ \vr_a: \ \ a\in\reals\},
    \end{align*}
    where $\vr_a$ is defined as
    \begin{align*} 
    r_{ki,0} = a,
    r_{ki,1} = -a,
    \forall k \in [n],\ i\in\{0,1\}.
\end{align*}
On the other hand, each vector $\vr_a$ is a scalar multiple of any other vectors of this kind.
As a result, all vectors $\vr$ that satisfy $\rank(\Mr) = 4n-1$, produce two extreme rays.
Convenient choices for these extreme rays are the vectors $\vr^{\mathrm{inv}_1}$ and $\vr^{\mathrm{inv}_2}$ defined by
\begin{align}
\label{eq: trivial ray}
    \begin{cases}
        r^{\mathrm{inv}_1}_{ki,0} = -1, & \forall k \in [n],\ i \in \{0,1\}, \\
        r^{\mathrm{inv}_1}_{ki,1} = \phantom{-}1, & \forall k \in [n],\ i \in \{0,1\}.
    \end{cases} \nonumber \\
    \begin{cases}
        r^{\mathrm{inv}_2}_{ki,0} = \phantom{-}1, & \forall k \in [n],\ i \in \{0,1\}, \\
        r^{\mathrm{inv}_2}_{ki,1} = -1, & \forall k \in [n],\ i \in \{0,1\}.
    \end{cases}
\end{align}
Moreover, since both inequalities $\vpt \vr^{\mathrm{inv}_1} \leq 0$ and
\[
-\vpt \vr^{\mathrm{inv}_1} = \vpt \vr^{\mathrm{inv}_2} \leq 0
\]
are required to hold, it follows that $\vpt \vr^{\mathrm{inv}_1} = 0$.
As $\vr^{\mathrm{inv}_1}$ and $\vr^{\mathrm{inv}_2}$ span $\ker(M)$, we have
\[
\vpt \vs = 0, \qquad \forall \vs \in \ker(M),
\]
whenever $\vp$ satisfies \Cref{eq: polar_cone_condition_formal}.
\end{remark}

The two extreme rays identified in \Cref{rem: trivial ray} belong to the family \textbf{(I)} in \Cref{def: pattern of extreme rays}.

\begin{lemma}[Shift invariance]
\label{lem: invariency of M_x}
    Let $\vx$ be feasible for \Cref{eq: general dual feasibility test_formal}, i.e., $M \vx \le \mathbf{0}$.
    Fix any $a \in \mathbb R$, and define $\vy$ by
    \begin{align*}
    y_{ki,0} = x_{ki,0} + a,
    \qquad
    y_{ki,1} = x_{ki,1} - a,
    \qquad
    \forall k \in [n],\ i\in\{0,1\}.
    \end{align*}
    Then $\vy$ is also feasible: $M \vy \le \mathbf{0}$.
    Moreover, $\Mx = \My$ and $\vpt \vx = \vpt \vy$; in particular, $\vx$ and $\vy$ induce the same constraints on $\vp$.
\end{lemma}
\begin{proof}
Each inequality in $M \vx \le \mathbf{0}$ involves a sum of one $(\cdot,0)$-coordinate and one $(\cdot,1)$-coordinate.
Under the transformation above, these sums remain unchanged, hence feasibility is preserved.
The equalities $\Mx=\My$ and $\vpt \vx=\vpt \vy$ follow from the same cancellation (see \Cref{rem: trivial ray}). 
\end{proof}

\begin{remark}[Normalization]
\label{rem: fix last variable to zero}
By \Cref{lem: invariency of M_x}, for the set of extreme rays with $\rank(\Mr) = 4n-2$, we may impose the normalization $r_{(n-1)1,1}=0$ without loss of generality.
\end{remark}

\subsection{Characterizing the Extreme Rays} 
\label{sec: instructing pattern of extreme rays}

Under \Cref{rem: fix last variable to zero}, the vector $\vr$ has $4n-1$ free (unfixed) coordinates.
Throughout the remainder of this appendix, the term \emph{coordinate} refers to these $4n-1$ unfixed entries, excluding $r_{(n-1)1,1}$.
By \Cref{lem: rank of rays are 4n-2}, any nonzero, nontrivial extreme ray $\vr$ of $\mathcal K$—excluding the case $\rank(\Mr)=4n-1$ already treated in \Cref{rem: trivial ray}—satisfies exactly $4n-2$ linearly independent tight equalities.
Consequently, the solution space of the linear equation system $\Mr \vr = \mathbf{0}$ is one–dimensional.

Equivalently, $\vr$ is unique up to scaling.
In particular, at most one coordinate can remain absent from the active equalities.
We therefore classify all extreme rays according to the number of unconstrained coordinates.

\subsubsection{Case 1: A Single Unconstrained Coordinate}
\label{sec: case 1}

Suppose that exactly one coordinate does not appear in any equality.
Then all remaining $4n-2$ variables are uniquely determined by the linear equation system.
Since the system is homogeneous and closed under positive homogeneity, this unique solution should also satisfy positive homogeneity (if $\vr$ is the unique solution then $\vr = \lambda \vr$ $\forall \lambda \geq 0$) therefore, this solution have to set all constrained variables equal to zero.

Thus $\vr = c \mathbf{e}_i$, where $c$ is constant and $\mathbf{e}_i$ is the standard basis vector corresponding to the unconstrained coordinate.
Feasibility of \Cref{eq: general dual feasibility test_formal} requires $c<0$, as a result, a convenient choice for these extreme rays are
\begin{equation*}
\vr = - \mathbf{e}_i.
\end{equation*}
This case contributes exactly $4n-1$ extreme rays belonging to the family \textbf{(II)} in \Cref{def: pattern of extreme rays}.

In what follows, we assume that every variable appears in at least one equality of our linear equation system.

\begin{lemma}
\label{clm: patterns of extreme rays}
Every $M$-equivalence class $\mathcal{E}_\vr$ (other than those in Case~\ref{sec: case 1} and \Cref{eq: trivial ray}) has a member, e.g. $\vr$, that satisfies
\begin{equation*}
\vr \in \{-1,0,1\}^{4n-1}.
\end{equation*}
Moreover, $\vr$ must follow exactly one of the two sign patterns below:
\begin{align}
\label{eq: different extreme ray cases}
\textbf{Pattern A:}\quad
\begin{cases}
r_{ki,0} \in \{0,-1\},\\
r_{ki,1} \in \{0,+1\},
\end{cases}
\qquad
\textbf{Pattern B:}\quad
\begin{cases}
r_{k0,0} = 0,\\
r_{k1,0} \in \{0,+1\},\\
r_{ki,1} \in \{0,-1\},
\end{cases}
\end{align}
for all $k \in [n]$ and $i\in\{0,1\}$.
\end{lemma}
Note that, based on the feasibility of $\vr$, in \emph{Pattern B}, if $r_{ki,1} = 0$ for all possible $k$ and $i$, then $r_{k1,0}$ have to be zero for all $k$ which yields to a trivial ray $\boldsymbol{0}$ that is a member of trivial extreme ray  \Cref{eq: trivial ray}.
Excluding that, these two patterns are disjoint.

Before presenting the formal proof, we provide a high-level outline.
\begin{proof}[Proof outline]
Let $\vr$ be an extreme ray not covered by Case~\ref{sec: case 1}.
The defining equalities form a homogeneous linear system with $4n-1$ variables and rank $4n-2$.
Hence, the solution space is one–dimensional and can be parameterized as
\begin{equation*}
r_{ki,j} = a_{ki,j} w + b_{k,i,j},
\end{equation*}
for some scalar $w$.

Homogeneity of the system implies that if $\vr$ is feasible, then $\lambda \vr$ is feasible for all $\lambda\ge0$.
This forces each nonzero coordinate of $\vr$ to have fixed magnitude, which we normalize to $1$.
Therefore $r_{ki,j}\in\{-1,0,1\}$.

Each equality has the form $r_{k_0i_0,0} + r_{k_1i_1,1} = 0$, which implies that all nonzero $r_{ki,0}$ share the same sign and all nonzero $r_{ki,1}$ share the opposite sign.
Using \Cref{rem: fix last variable to zero}, which enforces $r_{k0,0}\le0$, only the two sign patterns in \eqref{eq: different extreme ray cases} are possible.
\end{proof}
Using the intuition given in the outline, we continue on with the full version of the proof.
\begin{proof}[Full version of the proof]
    Consider an extreme ray $\vr$ and the system of linear equalities obtained from the inequalities of \Cref{eq: general dual feasibility test_formal} that are active at $\vr$.
    By Assumption~\ref{rem: fix last variable to zero} and the feasibility of $\vr$, we have
    \begin{equation*}
    r_{(n-1)1,1} + r_{k0,0} \le 0 \Rightarrow r_{k0,0} \le 0 \qquad \forall k \in [n].
    \end{equation*}
    
    As mentioned before, the solution space of this homogeneous linear system $M\vr \leq \mathbf{0}$ is one--dimensional.
    Consequently, there exists a scalar parameter $w \in \reals$ and coefficients $a_{ki,j}, b_{ki,j} \in \reals$ such that every solution of the system, and in particular each the extreme ray $\vr$, can be written as
    \begin{equation}
    \label{eq: linear system answer form}
    r_{ki,j} = a_{ki,j} w + b_{ki,j},
    \qquad
    \forall k \in [n],\ i,j \in \{0,1\}.
    \end{equation}
    
    Now consider any equality of the form
    \begin{equation*}
    r_{k_0i_0,0} + r_{k_1i_1,1} = 0
    \end{equation*}
    appearing in the system.
    Substituting \Cref{eq: linear system answer form} yields
    \begin{equation*}
    (a_{k_0i_0,0} + a_{k_1i_1,1})\, w
    +
    (b_{k_0i_0,0} + b_{k_1i_1,1}) = 0
    \qquad \forall w \in \reals.
    \end{equation*}
    Since this equality must hold for all $w$, we obtain
    \begin{equation*}
    a_{k_0i_0,0} + a_{k_1,i_1,1} = 0,
    \qquad
    b_{k_0i_0,0} + b_{k_1i_1,1} = 0.
    \end{equation*}

    To show that each nonzero coordinate has the same absolute value, we use the following crucial property of extreme rays: positive homogeneity.
    Therefore, for any $\lambda \ge 0$, there must exist a scalar $w_{\lambda}$ such that
    \begin{equation*}
    \lambda r_{ki,j} = a_{ki,j} w_{\lambda} + b_{ki,j}
    \qquad \forall k,i,j.
    \end{equation*}
    Substituting \Cref{eq: linear system answer form} into the left-hand side gives
    \begin{equation*}
    \lambda a_{ki,j} w + \lambda b_{ki,j}
    = a_{ki,j} w_{\lambda} + b_{ki,j}.
    \end{equation*}
    For indices $(k,i,j)$ such that $a_{ki,j} \neq 0$, this implies
    \begin{equation}
    \label{eq: relation for any coordinate}
    w_{\lambda} = \lambda w + \frac{(\lambda - 1)b_{ki,j}}{a_{ki,j}}.
    \end{equation}
    Since \Cref{eq: relation for any coordinate} holds for any indices that $a_{ki,j} \neq 0$, there exists a constant $c$ independent of $(k,i,j)$ such that
    \begin{equation*}
    \frac{b_{ki,j}}{a_{ki,j}} = c
    \end{equation*}
    whenever $a_{ki,j} \neq 0$.
    
    On the other hand, if $a_{ki,j} = 0$, then the above homogeneity condition reduces to
    \begin{equation*}
    b_{ki,j} = \lambda b_{ki,j} \qquad \forall \lambda \ge 0,
    \end{equation*}
    which forces $b_{ki,j} = 0$.

    We now justify the normalization of the nonzero coordinates of $\vr$.
    Define a graph $G$ whose vertices correspond to the coordinates $(k,i,j)$ of $\vr$, and where an edge connects two vertices $(k_0,i_0,0)$ and $(k_1,i_1,1) \neq ((n-1),1,1)$ whenever the equality
    \begin{equation*}
    r_{k_0,i_0,0} + r_{k_1,i_1,1} = 0
    \end{equation*}
    appears in the active constraint system.
    Also, we have some self loops on vertices $(k_0,i_0,0)$ whenever we have 
    \begin{equation*}
    r_{k_0,i_0,0} + r_{(n-1),1,1} = 0 \Rightarrow r_{k_0,i_0,0} = 0
    \end{equation*}
    in the active constraint system.

    Each connected component of $G$ corresponds to a subset of variables that are linearly linked through the equalities.
    Therefore, we can separate the whole linear system as several disjoint linear equations, each corresponding to a connected component of $G$.
    Since the whole linear system has rank $4n-2$ over $4n-1$ variables, exactly one connected component of $G$ (and its corresponding linear equation) contains one degree of freedom, while all others correspond to uniquely determined variables (since their linear equation system is full ranked, it has a unique solution, which is therefore a zero vector).
    Equivalently, for all vertices outside this specific component we have $a_{ki,j}=0$ and hence $r_{ki,j}=0$.

    Within the unique component containing the free variable, the equalities force all variables to be affine functions of the same scalar $w$.
    Note that there is no edge of type $r_{k_0,i_0,0} = 0$ in this connected component, due to the fact that each variable has a path with the node corresponding $r_{k_0,i_0,0}$ must have a value equal to zero and uniquely determined.
    Moreover, along any edge of $G$, the equality $r_{k_0i_0,0} + r_{k_1i_1,1} = 0$ implies that the corresponding coordinates must have the same absolute value.
    By connectivity, it follows that all nonzero coordinates in this component have equal absolute value.

    Since $\vr$ is a ray, it is defined only up to positive scaling.
    Without loss of generality (by rescaling $w$), we may normalize the ray so that
    \begin{equation*}
    \lvert a_{ki,j} w + b_{ki,j} \rvert = 1
    \qquad \text{whenever } a_{ki,j} \neq 0.
    \end{equation*}
    Consequently, each coordinate of $\vr$ satisfies
    \begin{equation*}
    r_{ki,j} \in \{-1, 0, 1\}.
    \end{equation*}
    
    We now analyze the sign structure of $\vr$.
    Each equality in the system has the form
    \begin{equation*}
    r_{k_0i_0,0} + r_{k_1i_1,1} = 0,
    \end{equation*}
    which implies that whenever both variables are nonzero, they must have opposite signs.
    Since we are excluding Case~\ref{sec: case 1}, every variable appears in at least one equality.
    As a result, due to the fact that all nonzero coordinates present in one connected component, so, they all have paths to each other, all variables $r_{ki,0}$ share the same sign, and all nonzero variables $r_{ki,1}$ share the opposite sign (we treat zero as sign--neutral).
    
    Finally, from \Cref{rem: fix last variable to zero} we have $r_{k0,0} \le 0$ for all $k \in [n]$.
    This restricts the possible sign configurations to exactly two cases:
    \begin{itemize}
    \item $r_{k0,0} = 0$ for all $k$, with $r_{k1,0} \ge 0$ and $r_{ki,1} \le 0$;
    \item $r_{ki,0} \le 0$ and $r_{ki,1} \ge 0$.
    \end{itemize}
    These correspond precisely to the two patterns stated in \Cref{eq: different extreme ray cases}, which completes the proof.
\end{proof}

Using \Cref{clm: patterns of extreme rays}, we now enumerate all extreme rays by determining which coordinates can be zero under each admissible sign pattern while satisfying the rank condition of \Cref{lem: rank of rays are 4n-2}.

\subsubsection{Case 2: Pattern B}

Suppose Pattern~B holds, namely
\begin{equation*}
r_{k0,0} = 0, \qquad
r_{k1,0} \in \{0,1\}, \qquad
r_{ki,1} \in \{0,-1\},
\quad \forall k \in [n],\ i\in\{0,1\}.
\end{equation*}

We distinguish two subcases.

\paragraph{Subcase 2.1: $\forall k,i:\ r_{ki,0}=0$.}
In this subcase, the only potentially nonzero variables are $\{r_{k0,1},r_{k1,1}\}_{k\in[n]}$.
Let
\begin{equation*}
A:=\{k\in[n]: r_{k0,1}=0\},\qquad 
B:=\{k\in[n]: r_{k1,1}=0\},
\end{equation*}
and denote $a:=\vert A \vert$ and $b:=\vert B \vert$.

We prove that the rank of the active constraint matrix in this subcase equals
\begin{equation*}
(2n-1)+a+b,
\end{equation*}
by constructing an independent family of this cardinality whose span contains all tight equalities.

\medskip
\noindent\textbf{Step 1: a spanning family.}
Define the following sets of equalities:
Let $k' \in A$ and $l' \in B$ be arbitrary fixed values.
\begin{align*}
\mathcal{L}_1 &:= \{\, r_{k0,0}+r_{k0,1}=0 \ :\ k\in A \,\}, \\
\mathcal{L}_2 &:=  \{\, r_{k'0,0}+r_{l1,1}=0 \ :\ l\in B \,\}
               \ \cup\ 
               \{\, r_{k0,0}+r_{l'1,1}=0:\ k\in[n]\setminus\{l'\} \,\}, \\
\mathcal{L}_3 &:= \{\, r_{l1,0}+r_{l1,1}=0 \ :\ l\in B \,\}, \\
\mathcal{L}_4 &:= \{\, r_{l1,0}+r_{k'0,1}=0 \ :\ l\notin B \,\}.
\end{align*}
\begin{remark}
Note that choosing $k'\in A$ and $l'\in B$ implicitly assumes that both sets are nonempty.
We now justify this assumption.

First, $B\neq\varnothing$ by normalization, since $r_{(n-1)1,1}=0$ implies $(n-1)\in B$.

Suppose next that $A=\varnothing$. Then $r_{k0,1}\neq 0$ for all $k\in[n]$, and under pattern~$B$ we have $r_{k0,1}=-1$ for every $k\in[n]$.
In this case, the families $\mathcal L_1$ and $\mathcal L_4$ are empty and the only tight equalities are of types $r_{k0,0}+r_{l1,1}=0$ and $r_{l1,0}+r_{l1,1}=0$.

Consequently, the number of independent equalities is at most
\[
(n-1+b)+b \le 2n-1,
\]
which is strictly smaller than the extreme ray rank $4n-2$.
Hence $A$ must be nonempty.
We may therefore fix arbitrary indices $k'\in A$ and $l'\in B$.
\end{remark}

Let
\begin{equation*}
\mathcal{L}:=\mathcal{L}_1\cup\mathcal{L}_2\cup\mathcal{L}_3\cup\mathcal{L}_4.
\end{equation*}
By construction,
\begin{equation*}
 \vert \mathcal{L} \vert =a+(n-1+b)+b+(n-b)=2n-1+a+b.
\end{equation*}

\medskip
\noindent\textbf{Step 2: spanning property.}
We show that any remaining tight equality lies in $\mathrm{span}(\mathcal{L})$.
The equalities of type $r_{k0,0}+r_{k0,1}=0$ and $r_{l1,0}+r_{l1,1}=0$ are already contained in $\mathcal{L}_1$ and $\mathcal{L}_3$.
It is easy to see that any equality in $\{\, r_{k0,0}+r_{l1,1}=0 \ :\ k\in[n],\, l\in B \,\}$ lies in $\mathrm{span}(\mathcal{L}_2)$.
It remains to show that $\{\, r_{k0,1}+r_{l1,0}=0 \ :\ l\in[n],\, k\in A \,\}$ lies in $\mathrm{span}(\mathcal{L})$.
Fix $l\in B$ and $k\in A$.
Then
\begin{equation*}
(r_{l1,0}+r_{l1,1}=0) - (r_{k0,0}+r_{l1,1}=0)
\quad\Rightarrow\quad
r_{l1,0}-r_{k0,0}=0.
\end{equation*}
Adding $(r_{k0,0}+r_{k0,1}=0)$ yields
\begin{equation*}
r_{l1,0}+r_{k0,1}=0.
\end{equation*}
Now, fix $l\notin B$ and $k\in A\setminus\{k'\}$.
Then
\begin{equation*}
    (r_{l1,0}+r_{k'0,1}=0) - (r_{k'0,0}+r_{k'0,1}=0)
    \quad\Rightarrow\quad
    r_{l1,0}-r_{k'0,0}=0.
\end{equation*}
Adding $(r_{k'0,0}+r_{l'1,1}=0) - (r_{k0,0}+r_{l'1,1}=0)$ yields
\begin{equation*}
    r_{l1,0}-r_{k0,0}=0.
\end{equation*}
Adding $(r_{k0,0}+r_{k0,1}=0)$ yields
\begin{equation*}
r_{l1,0}+r_{k0,1}=0.
\end{equation*}
Hence all equalities of type $r_{l1,0}+r_{k0,1}=0$ with $k\in A$ are linear combinations of equalities in $\mathcal{L}_1\cup\mathcal{L}_2\cup\mathcal{L}_3$.

\medskip
\noindent\textbf{Step 3: linear independence.}
Order the variables as
\begin{equation*}
\{r_{k0,1}:k\in A\},\quad
\{r_{k0,0}:k\in[n]\},\quad
\{r_{l1,1}:l\in B\},\quad
\{r_{l1,0}:l\in B\},\quad
\{r_{l1,0}:l\notin B\}.
\end{equation*}
Each family $\mathcal{L}_1,\dots,\mathcal{L}_4$ introduces a new pivot variable not used earlier:
\begin{itemize}
\item $\mathcal{L}_1$ pivots on $r_{k0,1}$ ($k\in A$),
\item $\mathcal{L}_2$ pivots on $r_{k0,0}$ ($k\in[n]\setminus\{l'\} $) and $r_{l1,1}$ ($l\in B$),
\item $\mathcal{L}_3$ pivots on $r_{l1,0}$ ($l\in B$),
\item $\mathcal{L}_4$ pivots on $r_{l1,0}$ ($l\notin B$).
\end{itemize}
The coefficient matrix therefore contains a block upper-triangular submatrix with nonzero diagonal, implying that $\mathcal{L}$ is linearly independent.

\medskip
\noindent\textbf{Conclusion.}
The active constraint matrix has rank
\begin{equation*}
2n-1+a+b.
\end{equation*}
Imposing the extreme ray rank condition $2n-1+a+b=4n-2$ yields $\{a,b\}=\{n-1,n\}$, and the configuration coincides with Case~\ref{sec: case 1}.

\paragraph{Subcase 2.2: $\exists k_0:\ r_{k_01,0}=1$.}

In this subcase, the equality constraints imply
\begin{equation*}
r_{l0,1}=-1 \quad \forall\, l\in[n],
\qquad
r_{k_01,1}=-1.
\end{equation*}
Moreover, by \Cref{rem: fix last variable to zero},
\begin{equation*}
r_{(n-1)1,0}=0.
\end{equation*}

Let
\begin{equation*}
A:=\{k\in[n]:r_{k1,0}=1\},\qquad 
B:=\{k\in[n]:r_{k1,1}=0\},
\end{equation*}
and denote $a:= \vert A \vert $ and $b:= \vert B \vert $.
Since $r_{k1,1} \leq -r_{k1,0}$, the above constraints imply $A\cap B=\varnothing$ and therefore
\begin{equation*}
a+b\le n.
\end{equation*}

\medskip
\noindent\textbf{Rank count.}
We now count the number of independent tight equalities.

\begin{itemize}
\item Equalities of type $r_{k0,0}+r_{k0,1}=0$ do not appear in this subcase.

\item Equalities of type $r_{k0,0}+r_{l1,1}=0$ contribute $n+b-1$ independent constraints.

\item Equalities of type $r_{l1,0}+r_{k0,1}=0$ contribute $n+a-1$ independent constraints.

\item Equalities of type $r_{k1,0}+r_{k1,1}=0$ contribute $a+b$ independent constraints.
\end{itemize}

Hence the total number of independent equalities equals
\begin{equation*}
(n+b-1)+(n+a-1)+(a+b)
=2n-2+2a+2b.
\end{equation*}

Imposing the extreme ray rank condition $2n-2+2a+2b=4n-2$ yields
\begin{equation*}
a+b=n.
\end{equation*}
\begin{remark}
    Note that $A$ cannot be empty. 
    Indeed, if $A=\varnothing$, then no equalities of type $r_{l1,0}+r_{k0,1}=0$ appear. 
    In this case the number of independent equalities is at most
    \begin{equation*}
    (n+b-1)+b \le 3n-1,
    \end{equation*}
    which is strictly smaller than the extreme ray rank $4n-2$.
    Hence $A\neq\varnothing$.
\end{remark}

\medskip
\noindent\textbf{Enumeration of rays.}
Thus $A$ and $B$ form a partition of $[n]$ with $(n-1)\notin A$ because
$r_{(n-1)1,0}=0$.  
Hence $A\subseteq[n-1]$ is nonempty, and every nonempty subset $A\subseteq[n-1]$ yields an extreme ray.

Equivalently, for any binary vector $\vs\in\{0,1\}^{n-1}$ that is not identically zero, define
\begin{equation*}
\begin{cases}
r_{k0,1}=-1, & \forall k\in[n],\\
r_{k0,0}=0, & \forall k\in[n],\\
r_{(n-1)1,0}=0,\\
r_{k1,0}=s_k, & \forall k\in[n-1],\\
r_{k1,1}=-r_{k1,0}, & \forall k\in[n].
\end{cases}
\end{equation*}
This produces $2^{\,n-1}-1$ $M$-distinct extreme rays belonging to the family \textbf{(IV)} in \Cref{def: pattern of extreme rays}.

\subsubsection{Case 3: Pattern A}

Suppose Pattern~A holds, namely
\begin{equation*}
r_{ki,0} \in \{0,-1\}, \qquad
r_{ki,1} \in \{0,+1\},
\quad \forall k\in[n],\ i\in\{0,1\}.
\end{equation*}

We consider four subcases.

\paragraph{Subcase 3.1: $\forall k,i:\ r_{ki,1}=0$.}
In this subcase, the only potentially nonzero variables are $\{r_{k0,0},r_{k1,0}\}_{k\in[n]}$.
Define
\begin{equation*}
A:=\{k\in[n]:r_{k0,0}=0\},\qquad
B:=\{k\in[n]:r_{k1,0}=0\},
\end{equation*}
and denote $a:= \vert A \vert $ and $b:= \vert B \vert $.

\medskip
\noindent\textbf{Rank count.}
We count the number of independent tight equalities.

\begin{itemize}
\item Equalities of type $r_{k0,0}+r_{k0,1}=0$ contribute $a$ independent constraints.

\item Equalities of type $r_{l1,0}+r_{k0,1}=0$ contribute $n+a-1$ independent constraints.

\item Equalities of type $r_{k1,0}+r_{k1,1}=0$ contribute $b$ independent constraints.

\item Equalities of type $r_{k0,0}+r_{l1,1}=0$ contribute $n-a$ independent constraints (Similar to Subcase 2.1).
\end{itemize}

Hence the total number of independent equalities equals
\begin{equation*}
a+(n+a-1)+(n-a)+b
=2n-1+a+b.
\end{equation*}

Imposing the extreme ray rank condition $2n-1+a+b=4n-2$
yields $\{a,b\}=\{n-1,n\}$.
Therefore this configuration coincides with Case~\ref{sec: case 1}.

\paragraph{Subcase 3.2: $\exists k_0,l_0:\ r_{k_00,1}=1,\ r_{l_01,1}=1$.}
In this case, the equality constraints force
\begin{equation*}
r_{l0,0}=r_{l1,0}=-1 \qquad \forall l\in[n].
\end{equation*}

Define
\begin{equation*}
A:=\{k\in[n]:r_{k0,1}=1\},\qquad
B:=\{k\in[n]:r_{k1,1}=1\},
\end{equation*}
and denote $a:= \vert A \vert $ and $b:= \vert B \vert $.
By \Cref{rem: fix last variable to zero}, $(n-1)\notin B$, hence $b\le n-1$.

\medskip
\noindent\textbf{Rank count.}
As in the previous subcases, we count the number of independent equalities:

\begin{itemize}
\item Equalities of type $r_{k0,0}+r_{k0,1}=0$ contribute $a$ constraints.

\item Equalities of type $r_{l1,0}+r_{k0,1}=0$ contribute $n+a-1$ constraints.

\item Equalities of type $r_{k0,0}+r_{l1,1}=0$ contribute $n-a$ constraints.

\item Equalities of type $r_{k1,0}+r_{k1,1}=0$ contribute $b$ constraints.
\end{itemize}

Hence the total number of independent equalities equals
\begin{equation*}
2n-1+a+b.
\end{equation*}

Imposing the extreme ray rank condition $2n-1+a+b=4n-2$
yields
\begin{equation*}
a=n,\qquad b=n-1.
\end{equation*}

\medskip
\noindent\textbf{Resulting extreme ray.}
These values uniquely determine the configuration, giving the extreme ray
\begin{equation*}
\begin{cases}
r_{k1,0}=-1, & \forall k \in [n],\\
r_{k0,0}=-1, & \forall k \in [n],\\
r_{(n-1)1,1}=0,\\
r_{k1,1}=1, & \forall k\in [n-1],\\
r_{k0,1}=1, & \forall k\in [n].
\end{cases}
\end{equation*}
which belongs to the family \textbf{(III)} in \Cref{def: pattern of extreme rays}.

\paragraph{Subcase 3.3: $\exists k_0:\ r_{k_00,1}=1,\ \forall l:\ r_{l1,1}=0$.}
In this subcase, the equality constraints imply
\begin{equation*}
r_{l1,0}=-1 \quad \forall l\in[n],
\qquad
r_{k_00,0}=-1 .
\end{equation*}

Define
\begin{equation*}
A:=\{k\in[n]:r_{k0,1}=1\},\qquad
B:=\{k\in[n]:r_{k0,0}=0\},
\end{equation*}
and denote $a:= \vert A \vert $ and $b:= \vert B \vert $.

\medskip
\noindent\textbf{Rank count.}
We now count the number of independent tight equalities.

\begin{itemize}
\item Equalities of type $r_{k1,0}+r_{k1,1}=0$ do not appear.

\item Equalities of type $r_{k0,0}+r_{l1,1}=0$ contribute $n+b-1$ independent constraints.

\item Equalities of type $r_{l1,0}+r_{k0,1}=0$ contribute $n+a-1$ independent constraints.

\item Equalities of type $r_{k0,0}+r_{k0,1}=0$ (for indices where both variables appear in tight constraints) contribute $a+b$ independent constraints.
\end{itemize}

\begin{remark}
If $B=\varnothing$, then no equalities of type $r_{k0,0}+r_{l1,1}=0$ appear.  
In this case the number of independent equalities is at most
\begin{equation*}
(n+a-1)+a=n+2a-1,
\end{equation*}
which is strictly smaller than the extreme ray rank $4n-2$.
Hence $B\neq\varnothing$.
\end{remark}

Hence the total number of independent equalities equals
\begin{equation*}
(n+b-1)+(n+a-1)+(a+b)=2n-2+2a+2b.
\end{equation*}

Imposing the extreme ray rank condition
\begin{equation*}
2n-2+2a+2b=4n-2
\end{equation*}
yields
\begin{equation*}
a+b=n.
\end{equation*}

\medskip
\noindent\textbf{Enumeration of rays.}
Thus, the sets $A$ and $B$ form a nontrivial partition of $[n]$.
Equivalently, for any binary vector $\vs\in\{0,1\}^n$ that is neither identically zero nor identically one, define
\begin{equation*}
\begin{cases}
r_{k1,1}=0, & \forall k \in [n],\\
r_{k1,0}=-1, & \forall k \in [n],\\
r_{k0,1}=s_k, & \forall k \in [n],\\
r_{k0,0}=-r_{k0,1}, & \forall k \in [n].
\end{cases}
\end{equation*}
This produces $2^n-2$ $M$-distinct extreme rays belonging to the family \textbf{(V)} in \Cref{def: pattern of extreme rays}.

\paragraph{Subcase 3.4: $\exists k_0:\ r_{k_01,1}=1,\ \forall l:\ r_{l0,1}=0$.}
In this subcase, the equality constraints imply
\begin{equation*}
r_{l0,0}=-1 \quad \forall l\in[n],
\qquad
r_{k_01,0}=-1 .
\end{equation*}

Define
\begin{equation*}
A:=\{k\in[n]:r_{k1,1}=1\},\qquad
B:=\{k\in[n]:r_{k1,0}=0\},
\end{equation*}
and denote $a:= \vert A \vert $ and $b:= \vert B \vert $.

\medskip
\noindent\textbf{Rank count.}
We count the number of independent tight equalities.

\begin{itemize}
\item Equalities of type $r_{k0,0}+r_{k0,1}=0$ do not appear.

\item Equalities of type $r_{k1,0}+r_{k1,1}=0$ contribute $a+b$ independent constraints.

\item Equalities of type $r_{l1,0}+r_{k0,1}=0$ contribute $n+b-1$ independent constraints.

\item Equalities of type $r_{k0,0}+r_{l1,1}=0$ contribute $n+a-1$ independent constraints.
\end{itemize}

\begin{remark}
    If $B=\varnothing$, then no equalities of type $r_{l1,0}+r_{k0,1}=0$ appear.  
    In this case the number of independent equalities is at most
    \begin{equation*}
    (n+a-1)+a=n+2a-1,
    \end{equation*}
    which is strictly smaller than the extreme ray rank $4n-2$.
    Hence $B\neq\varnothing$.
\end{remark}

Hence the total number of independent equalities equals
\begin{equation*}
(a+b)+(n+a-1)+(n+b-1)=2n-2+2a+2b.
\end{equation*}

Imposing the extreme ray rank condition
\begin{equation*}
2n-2+2a+2b=4n-2
\end{equation*}
yields
\begin{equation*}
a+b=n.
\end{equation*}

\medskip
\noindent\textbf{Enumeration of rays.}
By \Cref{rem: fix last variable to zero}, $(n-1)\notin A$.
Thus, $A\subseteq[n-1]$ is a nonempty subset, and every such subset produces an extreme ray.

Equivalently, for any nonzero binary vector $\vs\in\{0,1\}^{\,n-1}$, define
\begin{equation*}
\begin{cases}
r_{k0,1}=0, & \forall k \in [n],\\
r_{k0,0}=-1, & \forall k \in [n],\\
r_{(n-1)1,1}=0,\\
r_{k1,1}=s_k, & \forall k \in [n-1],\\
r_{k1,0}=-r_{k1,1}, & \forall k \in [n].
\end{cases}
\end{equation*}
This produces $2^{\,n-1}-1$ $M$-distinct extreme rays belonging to the family \textbf{(VI)} in \Cref{def: pattern of extreme rays}.
\medskip

The above case analysis exhausts all feasible configurations and characterizes all extreme rays of the cone.
This completes the proof of \Cref{thm: pattern of extreme rays}.

\subsection{Number of Extreme Rays}

Summing the contributions from all cases,
\begin{equation*}
2 + (4n-1) + (2^{n-1}-1) + 1 + (2^n-2) + (2^{n-1}-1)
= 2^{n+1} + 4n - 2,
\end{equation*}
which proves the following Corollary.
\begin{restatable}{corollary}{crlnumber}
\label{crl: number of extreme rays}
    According to \Cref{thm: pattern of extreme rays}, 
    the cone $\mathcal{K}$ is generated by $2^{n+1} + 4n - 2$ $M$-distinct vectors $r_i$.
\end{restatable}

\subsection{Extracting the IV Inequalities}
\label{sec: proof of necessity and sufficiency of extreme rays}
 \thminequalities*
\begin{proof}
Based on the fact that $\mathcal{K}$ can be written as a cone combination of $r_i$s, we can write the sufficient inequalities as $\vpt r_i \leq 0$.
As $\vp$ is the observed probability law, the following inequalities are obviously satisfied, hence, there is no need to check them.
\begin{align}
\label{eq: obvious inequalities}
\begin{cases}
     &\hspace{-1em}\sum\limits_{k\in [n]} (p_{k0,1} + p_{k1,1}) - \sum\limits_{k\in [n]} (p_{k0,0} + p_{k1,0})
    \leq 0 \\
    &\hspace{-1em}\sum\limits_{k\in [n]} (p_{k0,0} + p_{k1,0}) - \sum\limits_{k\in [n]} (p_{k0,1} + p_{k1,1})
    \leq 0 \\
    &\hspace{-1em}-p_{k0,0} \leq 0 \quad \forall k\in [n] \\
    &\hspace{-1em}-p_{k0,1} \leq 0 \quad \forall k\in [n] \\
    &\hspace{-1em}-p_{k1,0} \leq 0 \quad \forall k\in [n] \\
    &\hspace{-1em}-p_{k1,1} \leq 0 \quad \forall k\in [n-1] \\
    &\hspace{-1em}\sum\limits_{k\in [n-1]} (p_{k1,1} + p_{k0,1}) + p_{(n-1)0,1} 
    \leq \sum\limits_{k\in [n]} (p_{k1,0} + p_{k0,0}) 
\end{cases}
\end{align}
Therefore, the given inequalities in \Cref{thm: iv inequalities} are sufficient to test the validity of the observed probability vector $\vp$ satisfying our problem assumptions and setup.
     
For necessity part, we need to show that each extreme ray presented in \Cref{thm: pattern of extreme rays} cannot be derived as cone combination of the others and moreover, for each inequality in \Cref{eq: necessary and sufficient iv inequalities}, there exists a vector $\vp$ such that it satisfies all inequalities except that specific one.

Fix an extreme ray representative $\vr\in\mathcal{R}$ (where $\mathcal{R}$ is the family in \Cref{def: pattern of extreme rays}), and suppose, by contradiction, that
\begin{equation*}
\vr=\sum_{t=1}^m \lambda_t \vr_t,\qquad \lambda_t\ge 0,
\end{equation*}
where each $\vr_t\in\mathcal{R}$ are such that $\mathcal{E}_{\vr_t}\neq \mathcal{E}_{\vr}$ whenever $\vr_t \notin \ker(M)$.

Let $\mathcal I:=\{t:\lambda_t>0\}$ and pick any $t_0\in\mathcal I$. 
Define
\begin{equation*}
\vr_1:=\lambda_{t_0}\vr_{t_0},\qquad 
\vr_2:=\sum_{t\in\mathcal I\setminus\{t_0\}}\lambda_t \vr_t.
\end{equation*}
Then $\vr_1,\vr_2\in\mathcal K$ and $\vr=\vr_1+\vr_2$.

Since $\vr$ is an extreme ray, by \Cref{def: extreme ray} 
we have
\begin{equation*}
\mathcal{E}_{\vr_1}=\mathcal{E}_{\vr_2}=\mathcal{E}_{\vr}.
\end{equation*}
Hence there exist $\alpha>0$ and $\vs\in\ker(M)$ such that
\begin{equation*}
\vr=\alpha\,\vr_{t_0}+\vs .
\end{equation*}
We distinguish two cases.

\medskip
\noindent\textbf{Case 1: $\vr\notin\ker(M)$.}
This implies that $\vr_{t_0} \notin \ker(M)$.
By the normalization in \Cref{rem: fix last variable to zero}, $r_{(n-1)1,1}=0$.  
Since $\vs\in\ker(M)$, we have $\vs=\boldsymbol{0}$ if and only if $s_{(n-1)1,1}=0$, which implies $\vs=\boldsymbol{0}$ and therefore
\begin{equation*}
\vr=\alpha\,\vr_{t_0} .
\end{equation*}
All vectors in $\mathcal R$ are unique rays; hence $\vr=\vr_{t_0}$.

\medskip
\noindent\textbf{Case 2: $\vr\in\ker(M)$.}
Multiplying the above equality by $M$ yields
\begin{equation*}
\mathbf{0}=M\vr=\alpha\,M\vr_{t_0},
\end{equation*}
so $\vr_{t_0}\in\ker(M)$.  
Since both $\vr$ and $\vr_{t_0}$ belong to $\ker(M) = \text{cone}(\vs_1, \cdots, \vs_t, -\vs_1, \cdots, -\vs_t)$, combining the fact that $\vs_1, \cdots, \vs_t$ are independent bases of $\ker(M)$ with positiveness of $\alpha$ implies that $\vr$ and $\vr_{t_0}$ must coincide. 
Hence
\begin{equation*}
\vr=\vr_{t_0} .
\end{equation*}


Therefore, $\vr$ cannot be expressed as a conic combination of representatives of the other elements in $\mathcal{R}$. Since the conic hull of these representatives is a closed convex set, the strong separating hyperplane theorem guarantees the existence of a vector $\vp$ such that
\[
\vp^\top \vr > 0
\quad \text{and} \quad
\vp^\top \vx \le 0
\]
for all $\vx$ in the conic hull generated by representatives of the remaining elements of $\mathcal{R}$.
Moreover, since $\vr$ does not correspond to an extreme ray inducing the inequalities in \Cref{eq: obvious inequalities}, the vector $\vp$ satisfies these inequalities.
In particular, $\vp \neq \mathbf{0}$, and by normalization, the vector
\[
\vp^* = \frac{2\vp}{\|\vp\|_1}
\]
defines a valid observed probability vector such that
\[
(\vp^*)^\top \vr > 0
\quad \text{and} \quad
(\vp^*)^\top \vx \le 0
\]
for all $\vx$ in the conic hull generated by representatives of the other elements of $\mathcal{R}$.

As a result, for each inequality in \Cref{eq: necessary and sufficient iv inequalities}, there exists an observed probability vector $\vp$ that violates that inequality while satisfying all the others.
Consequently, each inequality $\vpt \vr\le 0$ associated with $\vr\in\mathcal{R}$ is \emph{necessary}: removing it would enlarge the feasible set of observed probability vectors $\vp$ beyond the set induced by $\mathcal{K}$.

Combining necessity with the sufficiency argument (since $\mathcal K=\text{cone}(\mathcal R)$), we conclude that the family of inequalities in \Cref{thm: iv inequalities} is necessary and sufficient.
\end{proof}

\section{Proofs of Generalized Versions}
\label{sec: proof of generalizations}

Below, we provide generalized versions of our previous Lemmas.
Specifically, we no longer restrict our setting to $\ell=2$; the instrument $Z$ can take an arbitrary number $\ell$ of values.

\begin{lemma}\label{lem: vert iff rank is rank M}
Let $\vv$ be a feasible point satisfying $M\vv\le \vc$.
Then $\vv$ is a vertex if and only if
\[
\rank(\Mv)=2n\ell - \dim(\ker(M)) = \rank(M).
\]
\end{lemma}

\begin{proof}
We prove both directions.

\medskip
\noindent
\textbf{($\Leftarrow$) If $\rank(\Mv) < 2n\ell - \dim(\ker(M))$, then $\vv$ is not a vertex.}

Suppose $\rank(\Mv)< 2n\ell - \dim(\ker(M))$.
Note that $\rank(M) = 2n\ell - \dim(\ker(M))$.
Because $\rank(\Mv) < \rank(M)$, we have
\[
\dim \ker(\Mv) > \dim \ker(M).
\]
Hence, there exists a nonzero vector $\vt \in \ker(\Mv)$ such that
$\vt \notin \ker(M)$.

Since $\vt\in\ker(\Mv)$, we have
\[
\Mv(\vv+\epsilon\vt)=\Mv\vv=\vcv
\]
for all $\epsilon\in\mathbb{R}$.
For sufficiently small $\epsilon>0$, the inequalities corresponding to
non-active constraints remain strict, because they are strict at $\vv$.
Therefore,
\[
M(\vv+\epsilon\vt)\le \vc
\quad\text{and}\quad
M(\vv-\epsilon\vt)\le \vc
\]
for sufficiently small $\epsilon$.

Thus both $\vv+\epsilon\vt$ and $\vv-\epsilon\vt$ are feasible.
Moreover,
\[
\vv=\tfrac12(\vv+\epsilon\vt)+\tfrac12(\vv-\epsilon\vt),
\]
and since $\vt\notin\ker(M)$, we have
\[
M(\vv+\epsilon\vt)\neq M\vv,
\]
so the two feasible points are $M$-distinct.

Hence $\vv$ can be written as a non-trivial convex combination of two
$M$-distinct feasible points, and therefore $\vv$ is not a vertex by
\Cref{def: vertex}.

\medskip
\noindent
\textbf{($\Rightarrow$) If $\rank(\Mv)= 2n\ell - \dim(\ker(M))$, then $\vv$ is a vertex.}

Now suppose $\rank(\Mv)=2n\ell - \dim(\ker(M))$.
Assume
\[
\vv=\lambda\vvo+(1-\lambda)\vvt,
\qquad
\lambda\in(0,1),
\]
for feasible points $\vvo,\vvt$.

Since $\vv$ is active on $\Mv$, we have
\[
\Mv\vv=\vcv.
\]
Applying $\Mv$ to the convex combination gives
\[
\Mv\vv
=\lambda\Mv\vvo+(1-\lambda)\Mv\vvt.
\]
Because $\vvo$ and $\vvt$ are feasible,
\[
\Mv\vvo\le \vcv,
\qquad
\Mv\vvt\le \vcv.
\]
Since the convex combination equals $\vcv$, it follows that
\[
\Mv\vvo=\vcv
\quad\text{and}\quad
\Mv\vvt=\vcv.
\]
Therefore,
\[
\Mv(\vv-\vvo)=\mathbf{0},
\qquad
\Mv(\vv-\vvt)=\mathbf{0}.
\]

Because $\rank(\Mv)=\rank(M)$, and the fact the $\Mv$ is a submatrix of $M$ we have
\[
\ker(\Mv)=\ker(M).
\]
Hence,
\[
\vv-\vvo\in\ker(M),
\qquad
\vv-\vvt\in\ker(M).
\]

By \Cref{def: vertex}, any two feasible points differing by a vector
in $\ker(M)$ represent the same extreme point.
Thus $\vv$ cannot be written as a non-trivial convex combination of two
$M$-distinct feasible points, and therefore $\vv$ is a vertex.
\end{proof}

\begin{remark}\label{rem: general uniqness of Mv}
    Since $\rank(M)=\rank(\Mv) = 2n\ell - \dim(\ker(M))$, 
    the solution set of the system $\Mv \vx = \vc$ is  an affine space of the form $\vv + \ker(M)$. 
    Therefore, $\Mv$ uniquely determines the
    $M$-equivalence class $[\vv]$.

    In particular, if two active constraint matrices share the same
    row basis (that is, they consist of the same linearly independent
    rows), then they determine the same solution $\vv$ and thus the same
    $M$-equivalence class $[\vv]$. Consequently, the two active constraint
    matrices must coincide.
\end{remark}

\begin{lemma}
\label{lem: rank of rays are 2nl - dim ker M}
Let $\vr \neq \mathbf 0$ be a feasible vector of the cone $\mathcal{K}=\{\vx: M\vx \le 0\}$.
Let $\Mr$ denote the matrix of all inequalities that are active at $\vr$ as described in \Cref{def: vertex matrix IV version}.
Then $\vr$ is an extreme ray of $\mathcal K$ if and only if the active set contains at least $2n\ell - \dim(\ker(M)) - 1$ linearly independent constraints, i.e.,
\begin{equation*}
    \rank(\Mr) \geq 2n\ell - \dim(\ker(M)) - 1
\end{equation*}
\end{lemma}
\begin{proof}
$(\Rightarrow)$
Assume $\vr$ is an extreme ray and suppose, by contradiction, that $\rank(\Mr)<2n\ell - \dim(\ker(M)) - 1$.  
Then the null space of $\Mr$ has dimension strictly larger than $1+\dim\ker(M)$. 
Hence there exists a nonzero vector
\begin{equation*}
\vt \in \ker(\Mr)\setminus \mathrm{span}\!\left(\vr,\ker(M)\right).
\end{equation*}
Since $\vt\in\ker(\Mr)$, all inequalities that are tight at $\vr$ remain tight along the direction $\vt$. 
Therefore, for sufficiently small
$\epsilon>0$,
\begin{equation*}
M(\vr+\epsilon\vt)\le0,
\qquad
M(\vr-\epsilon\vt)\le0,
\end{equation*}
so both $\vr+\epsilon\vt$ and $\vr-\epsilon\vt$ belong to $\mathcal K$.
Moreover,
\begin{equation*}
\vr=\tfrac12(\vr+\epsilon\vt)+\tfrac12(\vr-\epsilon\vt).
\end{equation*}
Because $\vt\notin \mathrm{span}(\vr,\ker(M))$, at least one of the vectors $\vr\pm\epsilon\vt$ is not contained in $\mathrm{span}(\vr,\ker(M))$, contradicting \Cref{def: extreme ray}.
Hence $\rank(\Mr)\ge 2n\ell - \dim(\ker(M)) - 1$.

$(\Leftarrow)$
Assume $\rank(\Mr)\ge 2n\ell - \dim(\ker(M)) - 1$. Then
\begin{equation*}
\dim\ker(\Mr)\le 1+\dim\ker(M).
\end{equation*}
By construction, $\vr\in\ker(\Mr)$ and $\ker(M)\subseteq\ker(\Mr)$.

Let $\vr^{(1)},\vr^{(2)}\in\mathcal K$ such that
\begin{equation*}
\vr=\vr^{(1)}+\vr^{(2)} .
\end{equation*}
Applying $\Mr$ gives
\begin{equation*}
0=\Mr\vr=\Mr\vr^{(1)}+\Mr\vr^{(2)} .
\end{equation*}
Since $\vr^{(1)},\vr^{(2)}\in\mathcal K$, we have $\Mr\vr^{(1)}\le0$ and $\Mr\vr^{(2)}\le0$, hence
\begin{equation*}
\Mr\vr^{(1)}=\Mr\vr^{(2)}=\mathbf{0},
\end{equation*}
which implies $\vr^{(1)},\vr^{(2)}\in\ker(\Mr)$.

If $\vr\in\ker(M)$, then $\Mr=M$ and therefore $\vr^{(1)},\vr^{(2)}\in\ker(M)$, satisfying \Cref{def: extreme ray}.
Otherwise, $\vr\notin\ker(M)$. 
Since $\ker(M)\subseteq\ker(\Mr)$ and
\begin{equation*}
\dim\ker(\Mr)\le 1+\dim\ker(M),
\end{equation*}
we must have
\begin{equation*}
\ker(\Mr)=\mathrm{span}\!\left(\vr,\ker(M)\right).
\end{equation*}
Thus any feasible decomposition of $\vr$ lies in $\mathrm{span}\!\left(\vr,\ker(M)\right)$, which proves that $\vr$ is an extreme ray.
\end{proof}

\subsection{Vertices and Proofs of Section \ref{sec: llate}}
\propgeneralizedvertex*

\begin{proof}

For convenience we rewrite \Cref{eq: dual_lower_bound_formal} in the equivalent form
\begin{align}
\label{eq: generalized ACE dual optimization problem_app}
\max\quad 
&\sum_{y\in[n]}\sum_{d\in\{0,1\}}\sum_{j\in[\ell]} p_{yd,j}\,x_{yd,j} \\
\text{s.t.}\quad 
&\sum_{j=0}^{\ell-1} x_{y_{i_j} i_j, j} \le \gamma_{y_1}-\gamma_{y_0},
\qquad
\forall y_0,y_1\in[n],\ \forall (i_0,\dots,i_{\ell-1})\in\{0,1\}^{\ell}.
\nonumber
\end{align}

\begin{lemma}[Shift invariance for general version]
\label{lem: invariency of M_x in general}
    Let $\vv$ be feasible for \Cref{eq: generalized ACE dual optimization problem_app}.
    Fix any $t \in \R$ and $j \in [\ell] \setminus \{0\}$, and define $\vu$ by
    \begin{align*}
    u_{ki,0} = v_{ki,0} + t,
    \qquad
    u_{ki,j} = v_{ki,j} - t,
    \qquad
    \forall k \in [n],\ i\in\{0,1\}.
    \end{align*}
    Then $\vy$ is also feasible.
    Moreover, $\Mv = \Mu$ and $\vpt \vu = \vpt \vv$; in particular, $\vv$ and $\vu$ induce the same constraints on $\vp$.
\end{lemma}
\begin{proof}[proof of \Cref{lem: invariency of M_x in general}]
Each inequality in \Cref{eq: generalized ACE dual optimization problem_app} involves a sum of one $(\cdot,0)$-coordinate and one $(\cdot,j)$-coordinate.
Under the transformation above, these sums remain unchanged, hence feasibility is preserved and $\Mv=\Mu$.
The equality $\vpt \vv=\vpt \vu$ follows from the same cancellation due to the fact that $\sum_{ k \in [n],\ i\in\{0,1\}} p_{k i, 0} = \sum_{ k \in [n],\ i\in\{0,1\}} p_{k i, j} = 1$. 
\end{proof}

\begin{remark}[Normalization]
\label{rem: fix last variable to zero in general}
By \Cref{lem: invariency of M_x in general}, for any feasible point in \Cref{eq: generalized ACE dual optimization problem_app}, we may impose the normalization $v_{(n-1)1,j}=0, \forall j \neq 0$ without loss of generality.
\end{remark}

In order to prove \Cref{prop: generalized exponential vertices}, we show (i) feasibility of the constructed $\vw$, and (ii) that $\vw$ is a vertex by exhibiting a family of constraints that are tight at $\vw$ and whose equalities determine $\vw$ uniquely (after fixing a normalization for the shift invariance).

\paragraph{(i) Feasibility.}
Fix $(i_0,\dots,i_{\ell-1})\in\{0,1\}^{\ell}$ and let
$A:=\{j:i_j=0\}$ and $A^c:=\{j:i_j=1\}$.
For any $y_0,y_1\in[n]$, the left-hand side of the corresponding constraint equals
\[
\sum_{j\in A} w_{y_0 0,j} \;+\; \sum_{j\in A^c} w_{y_1 1,j}.
\]
By construction, among the terms $w_{y_1 1,j}$ only the coordinate with $j=a$ may be nonzero, and among the terms
$w_{y_0 0,j}$ only the coordinates with $j\in\{a,s_{y_0}\}$ may be nonzero. Hence the above sum contains at most two
nonzero contributions, at indices $a$ and $s_{y_0}$, and we obtain
\[
\sum_{j\in A} w_{y_0 0,j} + \sum_{j\in A^c} w_{y_1 1,j}
\le 
\max\{w_{y_0 0,a},w_{y_1 1,a}\}
+
\max\{w_{y_0 0,s_{y_0}},w_{y_1 1,s_{y_0}}\}.
\]
Since $w_{y_1 1,s_{y_0}}=0$ (because $s_{y_0}\neq a$), the second maximum equals $w_{y_0 0,s_{y_0}}=\gamma_{n-1}-\gamma_{y_0}$.
Moreover $w_{y_0 0,a}=\gamma_0-\gamma_{n-1}$ and $w_{y_1 1,a}=\gamma_{y_1}-\gamma_{n-1}$, so the first maximum equals $\gamma_{y_1}-\gamma_{n-1}$.
Therefore,
\[
\sum_{j\in A} w_{y_0 0,j} + \sum_{j\in A^c} w_{y_1 1,j}
\le \bigl(\gamma_{y_1}-\gamma_{n-1}\bigr) + \bigl(\gamma_{n-1}-\gamma_{y_0}\bigr)
= \gamma_{y_1}-\gamma_{y_0},
\]
and all constraints in \Cref{eq: generalized ACE dual optimization problem_app} are satisfied. Thus $\vw$ is feasible.

\paragraph{(ii) Vertex Property.}
We next show that the above construction is uniquely determined (up to the standard shift invariance)
by a collection of tight constraints from \Cref{eq: generalized ACE dual optimization problem_app}. 
Concretely, the construction implies that the following
constraints are tight:
\begin{align}
&\sum_{j=0}^{\ell-1} x_{y 0, j} = \gamma_0-\gamma_y, 
&\qquad \forall y\in[n], \label{eq: E1}\\
&\sum_{j=0}^{\ell-1} x_{y 1, j} = \gamma_{y} - \gamma_{n-1}, 
&\qquad \forall y\in[n], \label{eq: E2}\\
&x_{y 0, s_y} + \sum_{j \neq s_y} x_{k 1, j} = \gamma_{k} - \gamma_{y},
&\qquad \forall y,k\in[n], \label{eq: E3}\\
&x_{y 0, s_y} + x_{k 1, a} + \sum_{j \in A} x_{y 0, j} + \sum_{j \in A^c\setminus \{a, s_y\}} x_{k 1, j}
= \gamma_{k} - \gamma_{y},
&\quad \forall y,k,\ \forall A \subseteq [\ell]\setminus\{a,s_y\}.
\label{eq: E4}
\end{align}
After fixing the normalization
\[
x_{(n-1)1,j}=0,\qquad \forall j\neq 0,
\]
we prove that the resulting linear system has a unique solution, hence $\vw$ is a vertex of the feasible polyhedron.

From the normalization $x_{(n-1)1,j}=0$ for all $j\neq 0$ (\Cref{rem: fix last variable to zero in general}) and equality~\eqref{eq: E2} with $y=n-1$, we obtain
\[
\sum_{j=0}^{\ell-1} x_{(n-1)1,j}=\gamma_{n-1}-\gamma_{n-1}=0.
\]
Since all summands except possibly $j=0$ are zero, it follows that $x_{(n-1)1,0}=0$ as well. 
Thus,
\begin{equation}
\label{eq: all_zero_row}
x_{(n-1)1,j}=0,\qquad \forall j\in[\ell].
\end{equation}

Now apply~\eqref{eq: E3} with $k=n-1$:
\[
x_{y0,s_y} + \sum_{j\neq s_y} x_{(n-1)1,j} = \gamma_{n-1}-\gamma_{y}.
\]
Using~\eqref{eq: all_zero_row}, the sum vanishes, and we conclude that for every $y$,
\begin{equation}
\label{eq: xys_explicit}
x_{y0,s_y}=\gamma_{n-1}-\gamma_{y}.
\end{equation}
Substituting~\eqref{eq: xys_explicit} back into~\eqref{eq: E3} for general $k$ gives
\[
\gamma_{n-1}-\gamma_{y} + \sum_{j\neq s_y} x_{k1,j} = \gamma_{k}-\gamma_{y},
\]
hence
\begin{equation}
\label{eq: sum_excluding_sy}
\sum_{j\neq s_y} x_{k1,j} = \gamma_{k}-\gamma_{n-1},\qquad \forall y,k.
\end{equation}

Combining \eqref{eq: E2} with \eqref{eq: sum_excluding_sy} gives
\begin{equation}
\label{eq: xk1_sy_zero}
x_{k1,s_y} = 0
\qquad \forall k,y.
\end{equation}

Define the index sets
\[
B \coloneqq \{j:\exists k \text{ with } s_k=j\},\qquad C\coloneqq B^c.
\]
By construction $s_y\neq a$ for all $y$, so $a\in C$.

Since there exists $y$ such that $s_y=j$ for each $j\in B$, using~\eqref{eq: xk1_sy_zero} implies that 
\begin{equation}
\label{eq: xk1_B_zero}
x_{k1,j}=0,\qquad \forall k,\ \forall j\in B.
\end{equation}
Consequently,~\eqref{eq: E2} becomes
\begin{equation}
\label{eq: sum_C}
\sum_{j\in C} x_{k1,j} = \gamma_{k}-\gamma_{n-1},\qquad \forall k.
\end{equation}

Next, specialize~\eqref{eq: E4} by choosing
$A\subseteq B\setminus\{s_y\}$ (so that $A$ contains no elements of $C$).
Using~\eqref{eq: xys_explicit} and~\eqref{eq: xk1_B_zero}, the terms involving $x_{k1,j}$ for $j\in B$
vanish, and~\eqref{eq: E4} simplifies to
\begin{equation} \label{eq: rewritting E4}
\gamma_{n-1}-\gamma_{y} + \sum_{j\in A} x_{y0,j} = \gamma_{k}-\gamma_{y} - x_{k1,a} - \sum_{j\in C\setminus\{a\}} x_{k1,j}.
\end{equation}
Combining~\eqref{eq: sum_C} with~\eqref{eq: rewritting E4} gives
\[
\sum_{j\in A} x_{y0,j} = 0.
\]
Varying $A$ over all subsets of $B\setminus\{s_y\}$ forces $x_{y0,j}=0$ for all $j\in B\setminus\{s_y\}$,
and combined with~\eqref{eq: xys_explicit} shows that, within $B$, the only possibly nonzero $d=0$
coordinate is at $s_y$:
\begin{equation}
\label{eq: y0_B_structure}
x_{y0,j}=0\ \ (j\in B\setminus\{s_y\}),\qquad x_{y0,s_y}=\gamma_{n-1}-\gamma_{y}.
\end{equation}

Now choose $A\subseteq C\setminus\{a\}$ in~\eqref{eq: E4}. 
With~\eqref{eq: xk1_sy_zero},
the remaining terms lie entirely in $C$. 
\[
\gamma_{n-1} - \gamma_{y} + x_{k 1, a}  + \sum_{j \in A} x_{y 0, j} + \sum_{j \in A^c\cap (C \setminus \{a\})} x_{k 1, j}
= \gamma_{k} - \gamma_{y}
\]
The dependence on $A$ then implies that for each $j\in C\setminus\{a\}$
there exists a scalar $h_j$ such that
\begin{equation}
\label{eq: hj}
x_{k1,j}=x_{y0,j}=h_j,\qquad \forall k,y,\ \forall j\in C\setminus\{a\}.
\end{equation}
Using~\eqref{eq: sum_C}, we obtain for all $k$:
\begin{align*}
x_{k1,a} + \sum_{j\in C\setminus\{a\}} h_j &= \gamma_{k}-\gamma_{n-1}.
\end{align*}
Also, with~\eqref{eq: E1} and~\eqref{eq: y0_B_structure}, we obtain for all $y$:
\[
x_{y0,a} + \sum_{j\in C\setminus\{a\}} h_j = \gamma_0-\gamma_{n-1}
\]
Finally, by the normalization and~\eqref{eq: hj}, for every $j\in C\setminus\{a\}$,
\[
h_j = x_{(n-1)1,j}=0,
\]
so $x_{k1,j}=x_{y0,j}=0$ for all $j\in C\setminus\{a\}$. Plugging this into the two displayed equations yields
\[
x_{y0,a}=\gamma_0-\gamma_{n-1},\qquad x_{k1,a}=\gamma_{k}-\gamma_{n-1}.
\]
Together with~\eqref{eq: xk1_B_zero} and~\eqref{eq: y0_B_structure}, this determines all coordinates of $\vx$ uniquely.
Therefore, after fixing the stated normalization, the linear system
\eqref{eq: E1}-\eqref{eq: E4} admits a unique solution, and hence the feasible point $\vw$ is a vertex of the feasible polyhedral.

\paragraph{(iii) Counting.}
Fix $a\in\{0,\dots,\ell-1\}$. 
The map $\vs$ is a length-$n$ vector taking values in $[\ell]\setminus\{a\}$, hence there are $(\ell-1)^n$ possible maps in total.
Among these, $(\ell-1)$ maps are constant. 
The construction excludes constant $\vs$, and moreover the resulting vertex is unchanged under the standard shift invariance, so the number of \emph{$M$-distinct} (non-constant) choices contributing $M$-distinct vertices is
\[
(\ell-1)^{\,n-1}-(\ell-1).
\]
Multiplying by the $l$ possible choices of $a$ gives the total count
\[
l\bigl((\ell-1)^{\,n-1}-(\ell-1)\bigr).
\]
This completes the proof.
\end{proof}

\thmExponentialATEBounds*
\begin{proof}
By \Cref{prop: generalized exponential vertices}, the feasible set $\mathcal{H}$ contains at least
\[
\ell\Big((\ell-1)^{\,n-1}-(\ell-1)\Big)
\]
$M$-distinct vertices.

Moreover, according to \Cref{rem: general uniqness of Mv} and similar to \Cref{prop: sharpness of vertices}, for each vertex $\vv \in \mathcal{H}$, there exists a pair $(\mathcal{P}, \mathcal{Q})$ satisfying \eqref{eq: relation between q and p} such that $\vv$ uniquely attains the optimum of the dual objective.

It follows that each such vertex $\vv$ induces an achievable ATE bound of the form $\vpt \vv$. Therefore, any complete set of sharp ATE bounds must include at least one bound corresponding to each of these vertices.

Hence, any complete collection of sharp ATE bounds must contain at least
\[
\ell\Big((\ell-1)^{\,n-1}-(\ell-1)\Big)
\]
bounds.
\end{proof}


\subsection{Extreme rays and Proofs of Section \ref{sec: llineq}}

In this setting, the dual feasibility problem \Cref{eq: general dual feasibility test_formal} can be written as
\begin{align}
\label{eq: generalized dual feasibility test}
    \max \quad & \vpt \vx \nonumber \\
    \text{s.t.} \quad
    &\sum_{j \in S} x_{y_0 0, j} + \sum_{j \in S^c} x_{y_1 1, j} \leq 0 
    \quad \forall y_0, y_1 \in \{0, \dots, n-1\},\ \forall S \subseteq \{0, \dots, \ell-1\}.
\end{align}

\propgeneralizediv*

In order to prove \Cref{prop: exponential inequalities general}, we introduce a set of extreme rays that produce necessary IV inequalities.

\begin{proposition}[Extreme rays generating the inequalities]
\label{prop: extreme rays general appendix}
Fix $y' \in [n-1]$ and $j' \in [\ell]$.  
Let $(j_0,\dots,j_{n-1}) \in ([\ell]\setminus\{j'\})^n$ be not identically constant.  
Define $\vr$ by
\begin{align*}
r_{y'1,j'} &= 1, \\
r_{y0,j_y} &= -1, \quad \forall y, \\
r_{y'1,j_y} &= -1, \quad \forall y,
\end{align*}
and all other entries equal to zero.  
Then $\vr$ is an extreme ray of the cone $\mathcal{K}$.
\end{proposition}

\begin{proof}
The proof consists of three steps.

\paragraph{Step 1: Feasibility.}
Fix any constraint indexed by $(y_0,y_1,S)$. 
The left-hand side is
\[
LHS = \sum_{j \in S} r_{y_0 0, j} + \sum_{j \in S^c} r_{y_1 1, j}.
\]

By construction, the only nonzero entries of $\vr$ are:
\[
r_{y'1,j'} = 1, \qquad
r_{y0,j_y} = -1, \qquad
r_{y'1,j_y} = -1.
\]

Hence, $LHS$ contains at most one positive term, namely $r_{y'1,j'}$.
If this term appears (i.e., $y_1 = y'$ and $j' \in S^c$), then since $j_{y_0} \neq j'$, either the term $r_{y_0 0, j_{y_0}} = -1$ or $r_{y' 1, j_{y_0}} = -1$ appear in the sums.
All remaining nonzero terms are nonpositive. 
Therefore,
\[
LHS \le 1 - 1 = 0,
\]
and the constraint is satisfied.
Hence $\vr \in \mathcal{K}$.

\paragraph{Step 2: Rank of the Active Constraints.}
First, note that the ambient dimension is $2n\ell$, and by shift invariance (see \Cref{lem: invariency of M_x in general}), the kernel has dimension at least $\ell - 1$.
Hence, by \Cref{lem: rank of rays are 2nl - dim ker M}, $\vr$ satisfies \Cref{def: extreme ray} if and only if
\[
\rank(\Mr) \geq 2n\ell  - \dim(\ker(M)) - 1.
\]

We now show that the set of tight constraints at $\vr$ has rank at least $2n\ell - \ell$.

Consider the following family of constraints, all of which are active (tight) at $\vr$:
\begin{align}
x_{y' 0, i} + \sum_{j \neq i} x_{y 1, j} &= 0,
&& \forall y \neq y',\ i \neq j_{y'}, \label{eq: B1}\\
\sum_{j} x_{y 1, j} &= 0,
&& \forall y \neq y', \label{eq: B2}\\
x_{y 0, i} + \sum_{j \neq i} x_{y 1, j} &= 0,
&& \forall y \neq y',\ i \neq j_{y}, \label{eq: B3}\\
\sum_{j \neq j', i} x_{y' 0, j} + x_{y' 1, j'} + x_{y' 1, i} &= 0,
&& \forall i \neq j_{y'}, j', \label{eq: B4}\\
\sum_{j \neq j', j_{y'}} x_{y'' 0, j} + x_{y' 1, j'} + x_{y' 1, j_{y'}} &= 0,
&& \forall y'' \text{ with } j_{y''} \neq j_{y'}, \label{eq: B5}\\
\sum_{j \neq j', j_y} x_{y 0, j} + x_{y 0 ,j_y} + x_{y' 1, j'} &= 0,
&& \forall y. \label{eq: B6}
\end{align}

We claim that these constraints determine $\vr$ uniquely up to scaling and addition of a vector in $\ker(M)$.

From \eqref{eq: B2}, we have
\[
\sum_{j} x_{y1,j} = 0 \quad \forall y \neq y'.
\]
Substituting into \eqref{eq: B3} yields, for all $i \neq j_y$,
\[
x_{y0,i} + \sum_{j \neq i} x_{y1,j}
= x_{y0,i} - x_{y1,i} = 0,
\]
hence
\[
x_{y0,i} = x_{y1,i}, \qquad \forall y \neq y',\ i \neq j_y.
\]
Combining with \eqref{eq: B1} yields, for each $i \neq j_y$ there exists a scalar $h_i$ such that,
\begin{equation*}
\label{eq:midproof_equal_blocks2}
h_i = x_{y' 0, i} = x_{y0,i} = x_{y1,i}, \qquad \forall y \neq y',\ i \neq j_y.
\end{equation*}
Applying in \eqref{eq: B2} yields, for all $y \neq y'$,
\begin{equation}
\label{eq:midproof_equal_blocks3}
x_{y 1, j_y} = -\sum_{j \neq j_y} h_j.
\end{equation}
Substituting into \eqref{eq: B6} yields, for all $y \neq y'$,
\begin{equation}
\label{eq:midproof_equal_blocks4}
x_{y 1, j_y} - x_{y 0, j_y} = x_{y' 1, j'} - h_{j'}.
\end{equation}
Combining with \eqref{eq: B5} yields, for all $y''$ with $j_{y''} \neq j_{y'}$,
\[
\sum_{j \neq j_{y'}} x_{y'' 1, j} + x_{y' 1, j_{y'}} = 0.
\]
Considering \eqref{eq: B2}, we obtain
\begin{equation*}
\label{eq: midproof_equal_y_prime}
    x_{y' 1, j_{y'}} = h_{j_{y'}} = x_{y' 0, j_{y'}}
\end{equation*}
Comparing \eqref{eq: B4} and \eqref{eq: B6} for $y = y'$, we obtain
\begin{equation*}
\label{eq: midproof_equal_y_prime2}
    x_{y' 1, i} = x_{y' 0, i} = h_{i}, \qquad \forall i \neq j', j_{y'}.
\end{equation*}
Applying in \eqref{eq: B4} yields
\begin{equation*}
\label{eq: midproof_equal_y_prime3}
    x_{y' 1, j'} = -\sum_{j \neq j'} h_{j}.
\end{equation*}
Combining with \eqref{eq:midproof_equal_blocks3} and \eqref{eq:midproof_equal_blocks4} results in
\begin{equation*}
\label{eq: midproof_equal_y_prime4}
    x_{y 0, j_y} =  h_{j_y}.
\end{equation*}

Thus, by fixing values of $h_i$ for all $i \in [\ell]$, $\vx$ is uniquely determined as follows
\begin{align*}
    x_{y 0, j} &= h_j, \qquad &&\forall y, j; \\
    x_{y 1, j} &= h_j, \qquad &&\forall y \neq y', j \neq j_y; \\
    x_{y 1, j_y} &= -\sum_{j \neq j_y} h_j \qquad &&\forall y \neq y';\\
    x_{y' 1, j} &= h_j, \qquad &&\forall j \neq j'; \\
    x_{y' 1, j'} &= -\sum_{j \neq j'} h_{j};
\end{align*}

Therefore, the solution space has dimension $\ell$.
Since $\dim(\ker(M)) \geq \ell - 1$ and $\vr$ belongs to the solution space, 
\[
\rank(\Mr) \geq 2n\ell - \ell \geq 2n\ell - \dim(\ker(M)) - 1.
\]

This shows that the active constraints have maximal rank for an extreme ray, and hence $\mathbf{x}$ defines an extreme ray.

\paragraph{Step 3: Counting.}
For each $(y',j')$, there are $(\ell-1)^n - (\ell-1)$ non-constant choices of the vector $(j_0,\dots,j_{n-1})$.
Multiplying by $(n-1)\ell$ yields the total number of $M$-distinct rays.
\end{proof}

\begin{proof}[Proof of \Cref{prop: exponential inequalities general}]
Each ray presented in \Cref{prop: extreme rays general appendix} induces the inequality
\[
\vpt \vr \leq 0,
\]
which is exactly inequalities in \Cref{prop: exponential inequalities general}.
In order to prove the necessity of these inequalities, it is suffice to show that for each inequality there exists a vector $\vp$ such that all of inequalities except this specific one is satisfied, which is obtained similarly to \Cref{sec: proof of necessity and sufficiency of extreme rays} since the inequalities are derived from $M$-distinct extreme rays.
Furthermore, as in \Cref{sec: proof of necessity and sufficiency of extreme rays}, the basic constraints  
\[
p_{y d, z} \geq 0 \qquad \forall y, d, z,
\quad \text{and} \quad
\sum_{y, d} p_{y d, 0} = \sum_{y, d} p_{y d, z}, \qquad \forall z \in [\ell],
\]
are implied by extreme rays other than those characterized in \Cref{prop: extreme rays general appendix}.
Hence, it follows that the vector $\vp$ is an observed probability vector, corresponding to some distribution $\mathcal{P}$.
\end{proof}

\section{Analytical Example}
\label{sec: example}

In this section, we present an explicit example for $n = \ell = 2$. 
We construct the corresponding matrix $M$, the set $S$, the set of vertices $\mathcal{V}$, 
the ATE bounds based on the distribution $\mathcal{P}$, set of extreme rays, and the associated IV inequalities.

\subsection{Formulation} \label{App: example1}

We begin by listing all the constraints given in \Cref{prop:relation}:
    \begin{align*}
        p_{0 0, 0} = q_{0 0,0 0} + q_{0 0,0 1} + q_{0 1,0 0} + q_{0 1,0 1}
        \\
        p_{1 0, 0} = q_{1 0,0 0} + q_{1 0,0 1} + q_{1 1,0 0} + q_{1 1,0 1}
        \\
        p_{0 1, 0} = q_{0 0,1 0} + q_{0 0,1 1} + q_{1 0,1 0} + q_{1 0,1 1}
        \\
        p_{1 1, 0} = q_{0 1,1 0} + q_{0 1,1 1} + q_{1 1,1 0} + q_{1 1,1 1}
        \\
        p_{0 0, 1} = q_{0 0,0 0} + q_{0 0,1 0} + q_{0 1,0 0} + q_{0 1,1 0}
        \\
        p_{1 0, 1} = q_{1 0,0 0} + q_{1 0,1 0} + q_{1 1,0 0} + q_{1 1,1 0}
        \\
        p_{0 1, 1} = q_{0 0,0 1} + q_{0 0,1 1} + q_{1 0,0 1} + q_{1 0,1 1}
        \\
        p_{1 1, 1} = q_{0 1,0 1} + q_{0 1,1 1} + q_{1 1,0 1} + q_{1 1,1 1}.
    \end{align*}
In addition, we impose the nonnegativity constraints
\[
q_{ij,\mathbf{d}} \ge 0 
\qquad \text{for all } i,j \in \{0,1\}, \; \mathbf{d} \in \{0,1\}^{2}.
\]
Our objective is to maximize or minimize the linear functional 
\[
\mathbf{c}^\top \mathbf{q},
\]
subject to the constraints $\mathbf{p} = M^\top\mathbf{q}$ and $\mathbf{q} \ge 0$, where
\setlength{\arraycolsep}{1pt}
\renewcommand{\arraystretch}{0.9}
\[
\mathbf{p} =
\begin{pmatrix}
p_{00,0} \\
p_{10,0} \\
p_{01,0} \\
p_{11,0} \\
p_{00,1} \\
p_{10,1} \\
p_{01,1} \\
p_{11,1}
\end{pmatrix}
\qquad
\mathbf{q} =
\begin{pmatrix}
q_{00,00} \\
q_{01,00} \\
q_{10,00} \\
q_{11,00} \\
q_{00,01} \\
q_{01,01} \\
q_{10,01} \\
q_{11,01} \\
q_{00,10} \\
q_{01,10} \\
q_{10,10} \\
q_{11,10} \\
q_{00,11} \\
q_{01,11} \\
q_{10,11} \\
q_{11,11}
\end{pmatrix}
\qquad
\mathbf{c} =
\begin{pmatrix}
0 \\
\gamma_1-\gamma_0 \\
\gamma_0 - \gamma_1 \\
0 \\
0 \\
\gamma_1-\gamma_0 \\
\gamma_0 - \gamma_1 \\
0 \\
0 \\
\gamma_1-\gamma_0 \\
\gamma_0 - \gamma_1 \\
0 \\
0 \\
\gamma_1-\gamma_0 \\
\gamma_0 - \gamma_1 \\
0
\end{pmatrix}
\]

\[
M =
\left[
\begin{array}{c|cccccccc}
& p_{00,0} & p_{10,0} & p_{01,0} & p_{11,0}
& p_{00,1} & p_{10,1} & p_{01,1} & p_{11,1} \\ \hline
q_{00,00} & 1 & 0 & 0 & 0 & 1 & 0 & 0 & 0 \\
q_{01,00} & 1 & 0 & 0 & 0 & 1 & 0 & 0 & 0 \\
q_{10,00} & 0 & 1 & 0 & 0 & 0 & 1 & 0 & 0 \\
q_{11,00} & 0 & 1 & 0 & 0 & 0 & 1 & 0 & 0 \\
q_{00,01} & 1 & 0 & 0 & 0 & 0 & 0 & 1 & 0 \\
q_{01,01} & 1 & 0 & 0 & 0 & 0 & 0 & 0 & 1 \\
q_{10,01} & 0 & 1 & 0 & 0 & 0 & 0 & 1 & 0 \\
q_{11,01} & 0 & 1 & 0 & 0 & 0 & 0 & 0 & 1 \\
q_{00,10} & 0 & 0 & 1 & 0 & 1 & 0 & 0 & 0 \\
q_{01,10} & 0 & 0 & 0 & 1 & 1 & 0 & 0 & 0 \\
q_{10,10} & 0 & 0 & 1 & 0 & 0 & 1 & 0 & 0 \\
q_{11,10} & 0 & 0 & 0 & 1 & 0 & 1 & 0 & 0 \\
q_{00,11} & 0 & 0 & 1 & 0 & 0 & 0 & 1 & 0 \\
q_{01,11} & 0 & 0 & 0 & 1 & 0 & 0 & 0 & 1 \\
q_{10,11} & 0 & 0 & 1 & 0 & 0 & 0 & 1 & 0 \\
q_{11,11} & 0 & 0 & 0 & 1 & 0 & 0 & 0 & 1
\end{array}
\right]
\]
\subsection{Example of Submatrix Selection}\label{App: submatrix}

We illustrate the notation $M(\alpha,\beta)$ on the following matrix for \Cref{def: submatrixx}:

\[
\left[
\begin{array}{c|cccccccc}
& p_{00,0} & p_{10,0} & p_{01,0} & p_{11,0}
& p_{00,1} & p_{10,1} & p_{01,1} & p_{11,1} \\ \hline
q_{00,00} & \cellcolor{gray!15}1 & 0 & 0 & \cellcolor{gray!15}0 & 1 & 0 & 0 & 0 \\
q_{01,00} & \cellcolor{gray!15}1 & 0 & 0 & \cellcolor{gray!15}0 & 1 & 0 & 0 & 0 \\
q_{10,00} & \cellcolor{gray!15}0 & 1 & 0 & \cellcolor{gray!15}0 & 0 & 1 & 0 & 0 \\
q_{11,00} & \cellcolor{gray!15}0 & 1 & 0 & \cellcolor{gray!15}0 & 0 & 1 & 0 & 0 \\
q_{00,01} & \cellcolor{gray!15}1 & 0 & 0 & \cellcolor{gray!15}0 & \cellcolor{gray!40}1 & \cellcolor{gray!40}0 & \cellcolor{gray!40}1 & \cellcolor{gray!40}0 \\
q_{01,01} & \cellcolor{gray!15}1 & 0 & 0 & \cellcolor{gray!15}0 & \cellcolor{gray!40}0 & \cellcolor{gray!40}0 & \cellcolor{gray!40}0 & \cellcolor{gray!40}1 \\
q_{10,01} & \cellcolor{gray!15}0 & 1 & 0 & \cellcolor{gray!15}0 & 0 & 0 & 1 & 0 \\
q_{11,01} & \cellcolor{gray!15}0 & 1 & 0 & \cellcolor{gray!15}0 & 0 & 0 & 0 & 1 \\
q_{00,10} & \cellcolor{gray!15}0 & 0 & 1 & \cellcolor{gray!15}0 & \cellcolor{gray!40}1 & \cellcolor{gray!40}0 & \cellcolor{gray!40}0 & \cellcolor{gray!40}0 \\
q_{01,10} & \cellcolor{gray!15}0 & 0 & 0 & \cellcolor{gray!15}1 & \cellcolor{gray!40}1 & \cellcolor{gray!40}0 & \cellcolor{gray!40}0 & \cellcolor{gray!40}0 \\
q_{10,10} & \cellcolor{gray!15}0 & 0 & 1 & \cellcolor{gray!15}0 & 0 & 1 & 0 & 0 \\
q_{11,10} & \cellcolor{gray!15}0 & 0 & 0 & \cellcolor{gray!15}1 & 0 & 1 & 0 & 0 \\
q_{00,11} & \cellcolor{gray!15}0 & 0 & 1 & \cellcolor{gray!15}0 & 0 & 0 & 1 & 0 \\
q_{01,11} & \cellcolor{gray!15}0 & 0 & 0 & \cellcolor{gray!15}1 & 0 & 0 & 0 & 1 \\
q_{10,11} & \cellcolor{gray!15}0 & 0 & 1 & \cellcolor{gray!15}0 & 0 & 0 & 1 & 0 \\
q_{11,11} & \cellcolor{gray!15}0 & 0 & 0 & \cellcolor{gray!15}1 & 0 & 0 & 0 & 1
\end{array}
\right]
\]

The dark gray entries correspond exactly to the submatrix
\[
M(0*01 \cup 0*10, **1).
\]
The light grey entries correspond exactly to the submatrix
\[
M(****, 110\cup000).
\]

\subsection{Vertices and ATE Bounds}\label{App: example2}

We can represent each $\vB\in\{0,1\}^{2\times2\times2}\in S$ (defined in \Cref{def: S define}) as the horizontal vector ordered as
\[
(b_{000},\, b_{100},\, b_{010},\, b_{11,0},\,
  b_{001},\, b_{101},\, b_{011},\, b_{111}).
\]
in
\[
S=\left\{
\begin{aligned}
S_1=\left\{ 
\begin{aligned}
&(1,1,0,1,\,1,1,1,0),\\
&(1,1,1,0,\,1,1,0,1),\\
\end{aligned}
\right\}\\
S_2=\left\{ 
\begin{aligned}
&(0,1,1,0,\,1,1,1,1),\\
&(1,1,1,0,\,0,1,1,1),\\
&(0,1,1,1,\,1,1,1,0),\\
&(1,1,1,1,\,0,1,1,0),\\
\end{aligned}
\right\}\\
S_2=\left\{ 
\begin{aligned}
&(0,1,1,1,\,1,0,1,1),\\
&(1,0,1,1,\,0,1,1,1)
\end{aligned}
\right\}\\
\end{aligned}
\right\}
\]

Now based on \Cref{def: U define}, we can create the vertex set $\mathcal{V}$ defined in \Cref{thm:main PI result}:
\[
\mathcal{V}=\left\{
\begin{aligned}
&(\phantom{-}0,-1,-1,\phantom{-}1,\phantom{-}0,-1,\phantom{-}0,-1),\\
&(\phantom{-}0,-1,\phantom{-}0,-1,\phantom{-}0,-1,-1,\phantom{-}1),\\
&(-1,-1,-1,-1,\phantom{-}1,\phantom{-}0,\phantom{-}0,\phantom{-}1),\\
&(\phantom{-}0,-1,-1,-1,\phantom{-}0,\phantom{-}0,\phantom{-}0,\phantom{-}1),\\
&(-1,-1,-1,\phantom{-}0,\phantom{-}1,\phantom{-}0,\phantom{-}0,\phantom{-}0),\\
&(\phantom{-}0,-1,-1,\phantom{-}0,\phantom{-}0,\phantom{-}0,\phantom{-}0,\phantom{-}0),\\
&(-1,\phantom{-}0,-1,\phantom{-}0,\phantom{-}1,-1,-1,\phantom{-}0),\\
&(\phantom{-}1,-1,-1,\phantom{-}0,-1,\phantom{-}0,-1,\phantom{-}0).
\end{aligned}
\right\}
\]
So based on \Cref{thm: PI ATE bounds} we have
\[
\max\left\{
\begin{aligned}
&-p_{10,0}-p_{01,0}+p_{11,0}-p_{10,1}-p_{11,1},\\[4pt]
&-p_{10,0}-p_{11,0}-p_{10,1}-p_{01,1}+p_{11,1},\\[4pt]
&-p_{00,0}-p_{10,0}-p_{01,0}-p_{11,0}+p_{00,1}+p_{11,1},\\[4pt]
&-p_{10,0}-p_{01,0}-p_{11,0}+p_{11,1},\\[4pt]
&-p_{00,0}-p_{10,0}-p_{01,0}+p_{00,1},\\[4pt]
&-p_{10,0}-p_{01,0},\\[4pt]
&-p_{00,0}-p_{01,0}+p_{00,1}-p_{10,1}-p_{01,1},\\[4pt]
&\phantom{-}p_{00,0}-p_{10,0}-p_{01,0}-p_{00,1}-p_{01,1}
\end{aligned}
\right\}
\le \text{ATE} 
\]
and
\[
\text{ATE} \le
\min\left\{
\begin{aligned}
&+p_{11,0}+p_{00,0}-p_{10,0}+p_{11,1}+p_{10,1},\\[4pt]
&+p_{11,0}+p_{10,0}+p_{11,1}+p_{00,1}-p_{10,1},\\[4pt]
&+p_{01,0}+p_{11,0}+p_{00,0}+p_{10,0}-p_{01,1}-p_{10,1},\\[4pt]
&+p_{11,0}+p_{00,0}+p_{10,0}-p_{10,1},\\[4pt]
&+p_{01,0}+p_{11,0}+p_{00,0}-p_{01,1},\\[4pt]
&+p_{11,0}+p_{00,0},\\[4pt]
&+p_{01,0}+p_{00,0}-p_{01,1}+p_{11,1}+p_{00,1},\\[4pt]
&-p_{01,0}+p_{11,0}+p_{00,0}+p_{01,1}+p_{00,1}
\end{aligned}
\right\}
\]
which are equivalent to the bounds introduced in \cite{balke1995probabilistic}.
\subsection{Extreme Rays and Testable Implications}\label{App: example3}

We now specialize \Cref{thm: iv inequalities} to the case $n=\ell=2$. 
By \Cref{thm: iv inequalities}, the IV model is valid if and only if the observed distribution $\vp$ satisfies the finite family of inequalities induced by the extreme rays of $\mathcal{K}$.

We represent each non trivial extreme ray as the horizontal vector ordered with respect to
\[
(p_{00,0},\, p_{10,0},\, p_{01,0},\, p_{11,0},\, p_{00,1},\, p_{10,1},\, p_{01,1},\, p_{11,1}).
\]
\begin{align*}
    &(\phantom{-}0, \phantom{-}0, \phantom{-}1, \phantom{-}0, -1, -1, -1, \phantom{-}0), \\
    &(-1, \phantom{-}0, -1, -1, \phantom{-}1, \phantom{-}0, \phantom{-}0, \phantom{-}0), \\
    &(\phantom{-}0, -1, -1, -1, \phantom{-}0, \phantom{-}1, \phantom{-}0, \phantom{-}0), \\
    &(-1, -1, -1, \phantom{-}0, \phantom{-}0, \phantom{-}0, \phantom{-}1, \phantom{-}0).
\end{align*}

\paragraph{Explicit Inequalities for $n=2$.}

Instantiating above extreme rays into $\vpt \vr \leq 0$ (equivalently \eqref{eq: necessary and sufficient iv inequalities} for $n=2$), we obtain the following system:



\begin{align*}
    - p_{0 0, 1} - p_{1 0, 1} + p_{0 1, 0} - p_{0 1, 1} &\leq 0;
\end{align*}
\begin{align*}
    - p_{0 1, 0} - p_{1 1, 0} + p_{0 0, 1} - p_{0 0, 0} &\leq 0, \\
    - p_{0 1, 0} - p_{1 1, 0} + p_{1 0, 1} - p_{1 0, 0} &\leq 0;
\end{align*}
\begin{align*}
    - p_{0 0, 0} - p_{1 0, 0} + p_{0 1, 1} - p_{0 1, 0} &\leq 0;
\end{align*}
which by using 
$(p_{00,1}+p_{10,1}+p_{01,1}+p_{11,1})
 =
(p_{00,0}+p_{10,0}+p_{01,0}+p_{11,0}) = 1$
are equivalent to 
\begin{align*}
     p_{0 1, 0} +  p_{1 1, 1} &\leq 1, \\
     p_{0 0, 1} +  p_{1 0, 0} &\leq 1, \\
     p_{1 0, 1} +  p_{0 0, 0} &\leq 1, \\
     p_{0 1, 1} +  p_{1 1, 0} &\leq 1.
\end{align*}

Therefore, producing a total of $2^{n+1}-4 = 4$ inequalities as follow:
\begin{align*}
    p_{0 1, 0} +  p_{1 1, 1} &\leq 1, \\
    p_{0 0, 1} +  p_{1 0, 0} &\leq 1, \\
    p_{1 0, 1} +  p_{0 0, 0} &\leq 1, \\
    p_{0 1, 1} +  p_{1 1, 0} &\leq 1.
\end{align*}

which are equivalent to the bounds introduced in \cite{pearl1995testability}.



\end{document}